\documentclass[final,1p,times]{elsarticle}

\usepackage{amsmath,amsthm,amssymb,amsfonts,amscd}
\usepackage[T1]{fontenc}
\usepackage{graphicx}

\def\R{\mathbb R}
\def\N{\mathbb N}
\def\C{\mathbb C}

\renewcommand{\d}{{\mathrm d}}
\renewcommand{\i}{{\mathrm i}}
\def\e{\mathrm e}
\def\O{\mathcal O}

\def\S{\mathcal{S}}

\def\P{\mathcal P}
\def\U{\mathcal{U}}

\DeclareMathOperator{\cotg}{cotg}
\DeclareMathOperator{\diag}{diag}

\newtheorem{lemma}{Lemma}[section]
\newtheorem{proposition}[lemma]{Proposition}

\newtheorem{observation}[lemma]{Observation}

\theoremstyle{definition}

\theoremstyle{remark}
\newtheorem{remark}[lemma]{Remark}
\newtheorem{notation}[lemma]{Notation}

\begin{document}

\begin{frontmatter}

\title{Potential-controlled filtering in quantum star graphs}
%
%
\author{Ond\v{r}ej Turek} \ead{ondrej.turek@kochi-tech.ac.jp}
\author{Taksu Cheon} \ead{taksu.cheon@kochi-tech.ac.jp}
\address{Laboratory of Physics, Kochi University of Technology, Tosa Yamada, Kochi 782-8502, Japan}

\date{\today}
\begin{abstract}

We study the scattering in a quantum star graph with a  F\"ul\"op--Tsutsui coupling in its vertex and with external potentials on the lines. We find certain special couplings for which the probability of the transmission between two given lines of the graph is strongly influenced by the potential applied on another line. On the basis of this phenomenon we design a tunable quantum band-pass spectral filter. The transmission from the input to the output line is governed by a potential added on the controlling line. The strength of the potential directly determines the passband position, which allows to control the filter in a macroscopic manner. Generalization of this concept to quantum devices with multiple controlling lines proves possible. It enables the construction of spectral filters with more controllable parameters or with more operation modes. In particular, we design a band-pass filter with independently adjustable multiple passbands. We also address the problem of the physical realization of F\"ul\"op--Tsutsui couplings and demonstrate that the couplings needed for the construction of the proposed quantum devices can be approximated by simple graphs carrying only $\delta$~potentials.

\end{abstract}


\end{frontmatter}

\section{Introduction}

Quantum mechanics on graphs 
is recently attracting more and more attention due to its current and prospective use in nanosciences, especially nanoelectronics \cite{AG05,EKST08,LBHS10}. 
In particular, since quantum graphs serve as effective models of graph-like structures of submicron sizes, they allow to study single electron devices based on interconnected nanoscale wires.
One of the first such applications emerged in the spectral filtering.
It is known that a line carrying the $\delta$-interaction is usable as a high-pass filter, and similarly, a line carrying the $\delta'$-interaction works as a low-pass filter. More complex spectral filters have been developed as well, for instance a trident filter~\cite{SCdL03}, or a Y-shaped branching filter, functionning as a high-pass/low-pass junction \cite{CET09}.

If a quantum device is easily controllable, its usefulness is markedly enhanced.
In principle, any device based on a quantum graph, including the aforementionned filters, can be controlled by adjusting its internal properties, for example, by varying the parameters of the couplings in its vertices. However, such a controllability is hard to be implemented in practice since it requires a real-time modification of a nanoscale object. It is, therefore, highly desirable to develop a device, controllable in a simpler manner.

In this paper we propose a quantum filtering device that can be controlled without affecting its internal structure. We consider a model in which constant potentials are applied to the lines of a star graph, and the control is achieved through the variation of the potential strengths.
The design of the device essentially relies on a special type of point interaction used in the star graph vertex, namely, on a scale invariant (F\"ul\"op--Tsutsui) coupling.
It should be emphasized that the F\"ul\"op--Tsutsui point interaction is very different from usual free coupling and the $\delta$~potential.
Its exotic scattering properties together with the external potentials on the lines induce a threshold resonance effect, on which the filtering function is based.

It is known that quantum graph vertices with F\"ul\"op--Tsutsui couplings can be approximately constructed from vertices carrying $\delta$-couplings \cite{CET10,CT10}. In this paper we improve the existing approximation schemes, and, in addition, demonstrate 
its convergence in terms of the scattering matrices. Since the vertices carrying the $\delta$-couplings can be 
approximated by regular, smooth potentials~\cite{Ex96b}, the realization of the required F\"ul\"op--Tsutsui couplings is within reach of the technology in the near future. Therefore, our quantum graph model of a quantum device is of a practical relevance, not just a mathematical abstraction.

The paper is organized as follows:
In Section~\ref{Section: Preliminaries} we prepare the necessary theoretical basis, including the scattering formalism of quantum graphs with external potentials on the lines.
In Section~\ref{Section: Device} we introduce the concept of a potential-controllable quantum filter.
The first device of this type is designed in Section~\ref{Section: n=3}, namely a band-pass spectral filter 
with one passband tunable by adjusting the potential on the controlling line.
Before generalizing the result to more complex devices, we devote Section~\ref{Section: Potential} to a practical question. 
Although the control potential is assumed 
to be constant on the whole line, we demonstrate that its small perturbation in a segment close to the vertex, which can be expected in practice, does not distort the filtering function.
The idea of Section~\ref{Section: n=3} is extended to filters with two controlling lines in Sections~\ref{Section: n=4} and~\ref{Section: band-pass}. The filter proposed in Section~\ref{Section: n=4} has two passbands, and the position of each of them is independently tunable by a control potential. The filter developed in Section~\ref{Section: band-pass} has one passband of adjustable bandwidth: its upper and lower cutoff energies are both tunable by the potentials on the controlling lines.
In Section~\ref{Section: n=2r} we introduce a spectral filter with $r$ controlling lines, $r\geq2$.  The filter has $r$ passbands, each of 
which can be independently adjusted (as well as disabled) by the potentials on the controlling lines. This result is achieved by generalizing the idea of Section~\ref{Section: n=4}. Finally, a filter with multiple independent outputs is studied in Section~\ref{Section: Branching}. The outputs work as individually tunable band-pass filters; the position of the passband at each output is adjustable by tuning a potential on a dedicated controlling line.

The construction of the F\"ul\"op--Tsutsui couplings is discussed in Sections~\ref{Section: Approximation} and~\ref{Section: S-matrix}. In Section~\ref{Section: Approximation} we demonstrate how they can be approximately represented by a small web carrying $\delta$-couplings in the vertices and vector potentials on the lines. In Section~\ref{Section: S-matrix} we calculate the on-shell scattering matrix of the approximating graph and prove its convergence to the scattering matrix of the required F\"ul\"op--Tsutsui vertex in the small size limit. The paper is concluded by Section~\ref{Section: Conclusion} in which we discuss possible extensions and related open problems.

\section{Preliminaries}\label{Section: Preliminaries}

A wave function of a particle confined to a graph having $N$ lines $E_1,\ldots,E_N$ of the lengths $\ell(e_1),\ldots,\ell(e_N)$ has $N$ components, $\Psi=(\psi_1,\psi_2,\ldots,\psi_N)^T$, and the corresponding Hilbert space $\mathcal{H}$ is given by $\bigoplus_{j=1}^N L^2(\ell(E_j))$. Let us assume that there are potentials $V_1,\ldots,V_N$ and vector potentials $A_1,\ldots,A_N$ imposed on the graph lines. The Hamiltonian action is given by
\begin{equation}
H
\begin{pmatrix}
\psi_1\\
\vdots\\
\psi_N
\end{pmatrix}
=\frac{1}{2m}
\begin{pmatrix}
\left(-\i\hbar\frac{\d}{\d x}-qA_1\right)^2\psi_1+V_1\cdot\psi_1\\
\vdots\\
\left(-\i\hbar\frac{\d}{\d x}-qA_N\right)^2\psi_N+V_N\cdot\psi_N
\end{pmatrix}
\,,
\end{equation}
where $m$ is the mass of the particle and $q$ is its charge. Note that if there is no vector potential, the Hamiltonian takes a simpler form
\begin{equation}
H
\begin{pmatrix}
\psi_1\\
\vdots\\
\psi_N
\end{pmatrix}
=
\begin{pmatrix}
-\frac{\hbar^2}{2m}\psi_1''+V_1\cdot\psi_1\\
\vdots\\
-\frac{\hbar^2}{2m}\psi_N''+V_N\cdot\psi_N
\end{pmatrix}\,.
\end{equation}

\subsection{Vertex couplings}

Let us consider a graph vertex at which $n$ lines are coupled. We assume that the wave function components at these lines are denoted by $\psi_1(x_1),\ldots,\psi_n(x_n)$ and that the coordinates are choosen such that $x_j=0$ corresponds to the vertex for all $j=1,\ldots,n$. We introduce the vectors
\begin{equation}
\Psi(0)=
\begin{pmatrix}
\psi_1(0) \\
\vdots \\
\psi_n(0)
\end{pmatrix}
\qquad\text{and}\qquad
\Psi'(0)=
\begin{pmatrix}
\psi'_1(0) \\
\vdots \\
\psi'_n(0)
\end{pmatrix} \, .
\end{equation}
Since the Hamiltonian is a second-order
differential operator, the boundary conditions at the vertex couple $\Psi(0)$ and $\Psi'(0)$. The most general form of the boundary conditions is
\begin{equation}\label{b.c.}
A\Psi(0)+B\Psi'(0)=0\,,
\end{equation}
where $A$ and $B$ are complex $n\times n$ matrices.

To ensure the self-adjointness of the Hamiltonian, which is in physical
terms equivalent to conservation of the probability current at the
vertex, the matrices $A$ and $B$ cannot be arbitrary but have
to satisfy the requirements
\begin{equation}\label{KS}
\begin{split}
\bullet \quad & \mathrm{rank}(A|B)=n,\\
\bullet \quad & \text{the matrix $AB^*$ is self-adjoint},
\end{split}
\end{equation}
where $(A|B)$ denotes the $n\times2n$ matrix with $A,B$ forming
the first and the second $n$ columns, respectively \cite{KS99}. The relation
\eqref{b.c.} together with the requirements~\eqref{KS} describe all possible vertex
boundary conditions giving rise to a self-adjoint Hamiltonian.

The requirements~\eqref{KS} can be partly included into the boundary conditions~\eqref{b.c.} by writing the matrices $A,B$ in a special form. One of the ways was discovered by Harmer~\cite{Ha00} and independently by Kostrykin and Schrader~\cite{KS00} who transformed \eqref{b.c.} with~\eqref{KS} into ``compact'' boundary conditions
\begin{equation}\label{V}
(\U-I)\Psi(0)+\i(\U+I)\Psi'(0)=0\,,
\end{equation}
where $\U$ is a unitary $n\times n$ matrix.
Other way was proposed in~\cite{CET10}. It consists in transforming~\eqref{b.c.} into the block form
\begin{equation}\label{ST}
\left(\begin{array}{cc}
I^{(r)} & T \\
0 & 0
\end{array}\right)\Psi'(0)=
\left(\begin{array}{cc}
S & 0 \\
-T^* & I^{(n-r)}
\end{array}\right)\Psi(0)\,,
\end{equation}
where $r\in\{0,1,\ldots,n\}$, $I^{(n)}$ is the identity matrix $n\times n$, $T$ is a general complex $r\times n-r$ matrix and $S$ is a Hermitian matrix of the order $r$. We call \eqref{ST} \emph{$ST$-form} of boundary conditions. It is simple and unique, but on the other hand it requires an appropriately chosen line numbering, see~\cite{CET10} for details.

In this paper we primarily use the $ST$-form~\eqref{ST} that best meets our needs. However, any form of boundary conditions has its pros and cons and is suitable for certain applications.

\subsection{$\delta$-couplings}

The $\delta$-coupling (or ``$\delta$~potential'') in a vertex of a quantum graph is characterized by boundary conditions
\begin{equation}\label{delta}
\psi_j(0)=\psi_\ell(0)=:\psi(0)\,, \quad j,\ell=1\ldots,n\,, \qquad
\sum^{n}_{j=1}\psi_j'(0)=\alpha\psi(0)\,,
\end{equation}
where $\alpha\in\R$ is a parameter of the coupling.
The $\delta$-coupling is the second most natural singular interaction (after the free coupling) by reason of its simple interpretation:
It can be understood as a limit case of properly scaled smooth potentials~\cite{Ex96b}.

\begin{remark}\label{delta V}
Strictly speaking, the $\delta$~potential of the strength $V_0$ ($V(x)=V_0\delta(x)$) is characterized by boundary conditions
\begin{equation}
\psi_j(0)=\psi_\ell(0)=:\psi(0)\,, \quad j,\ell=1\ldots,n\,, \qquad
\sum^{n}_{j=1}\psi_j'(0)=\frac{2m}{\hbar^2}V_0\psi(0)\,,
\end{equation}
where $m$ is the particle mass. However, since the particle in question is almost eclusively electron and, therefore, $m=m_e=\mathrm{const.}$, the term $\frac{2m}{\hbar^2}V_0$ is commonly written as one single constant $\alpha$, usually called parameter.
\end{remark}

\subsection{Scattering matrix}\label{Subsection: S-matrix}

If a quantum particle with mechanical energy $E$ living on the star graph comes in the vertex from the $\ell$-th line, it is scattered at the vertex into all the lines. The $j$-th component of the final-state wave function is given by
\begin{eqnarray}
\label{psi_ji}
\Psi^{(\ell)}_{j}(x)=\left\{\begin{array}{cl}
\frac{1}{\sqrt{k_j}}\e^{-{\rm i} k_j x}+\S_{jj}\frac{1}{\sqrt{k_j}}\e^{{\rm i} k_j x} & \quad\text{for } j=\ell\,, \\ \\
\S_{j\ell}\frac{1}{\sqrt{k_j}}\e^{{\rm i} k_j x} & \quad\text{for } j\neq \ell\,,
\end{array}\right. 
\end{eqnarray}
where $\S_{j\ell}$ are scattering amplitudes, $k_j$ are angular wavenumbers on the corresponding lines, and coefficients $1/\sqrt{k_j}$ are involved for proper normalization. For any $j$, the wavenumber $k_j$ is equal to
\begin{equation}\label{k_j}
k_j=\sqrt{\frac{2m(E-V_j)}{\hbar^2}}\,,
\end{equation}
where $V_j$ is the potential on the $j$-th line. The matrix $\S=\{\S_{j\ell}\}$ is the \emph{scattering matrix} of the graph. For a normalized wave function coming in from the $j$-th line, $\S_{ij}$ is interpreted as the complex amplitude of transmission into the $j$-th line (for $j\neq \ell$), whereas $\S_{\ell\ell}$ represents the complex amplitude of reflection.
The matrix $\S$ depends, in addition to the internal properties of the vertex, on the potentials $V_1,V_2,\ldots,V_n$, and on the particle energy $E$.
In order to express this dependence, we will often use the full symbol
\begin{equation}
\S(E;V_1,\ldots,V_n)\,.
\end{equation}
In case that $V_j=0$ for all $j=1,\ldots,n$, we will simplify the notation by omitting $V_j$, i.e.,
\begin{equation}
\S(E)\equiv\S(E;0,\ldots,0)\,.
\end{equation}

In order to derive a formula for $\S$, we define matrices $Y$ and $Y'$ such that $Y_{j\ell}=\Psi^{(\ell)}_{j}(0)$, $Y'_{j\ell}=(\Psi^{(\ell)}_{j})'(0)$.
With regard to~\eqref{psi_ji}, it holds
\begin{equation}\label{M}
\begin{split}
Y &= D^{-1} + D^{-1}\S=D^{-1}(I+\S), \\
Y' &=\i D^2(-D^{-1} + D^{-1}\S)=\i D(-I+\S),
\end{split}
\end{equation}
where $D$ is the diagonal matrix
\begin{equation}\label{K}
D = \diag\left(\sqrt{k_1},\ldots,\sqrt{k_n}\right) \, ,
\end{equation}
for $k_j$ given by~\eqref{k_j}.
Any wave function $\Psi(x)=(\psi_1(x),\ldots,\psi_n(x))^T$ (the superscript $T$ denotes the transposition) on the graph obeys the boundary conditions~\eqref{b.c.} determining the vertex. In particular, b.~c. must be satisfied by the final-state wave function $\Psi^{(\ell)}_{j}$ determined in~\eqref{psi_ji} for all $j$, hence $AY + BY' = 0$. When we substitute for $Y$ and $Y'$ from~\eqref{M}, we obtain
\begin{equation}
AD^{-1}(I+\S)+\i BD(-I+\S)=0\,,
\end{equation}
which
leads to the sought expression for the scattering matrix $\S(E;V_1,\ldots,V_n)$:
\begin{equation}\label{S}
\S(E;V_1,\ldots,V_n) = - (AD^{-1} + \i B D)^{-1}(AD^{-1} - \i B D)\,.
\end{equation}
Matrix $\S$ is always unitary, and squared moduli of its elements have the following interpretation: $|\S_{j\ell}|^2$ for $j\ne \ell$ represents the probability of transmission 
from the $\ell$-th to the $j$-th line, $|\S_{\ell\ell}|^2$ is the probability of reflection on the $\ell$-th line.

Note that if there are no potentials on the graph lines, i.e., $V_j=0$ for all $j=1,\ldots,n$, then $k_j=\sqrt{2mE/\hbar^2}$ for all $j$, and the scattering matrix can be written in the simpler form
\begin{equation}\label{S0}
\S(E) = -(A+\i kB)^{-1}(A-\i kB)\,, \qquad\text{where $k=\sqrt{2mE/\hbar^2}$}\,.
\end{equation}
Equation~\eqref{S0} allows to find a relation between the boundary conditions~\eqref{b.c.} and their form~\eqref{V}:
\begin{observation}
The boundary conditions~\eqref{b.c.} are equivalent to
\begin{equation}
\left(\S(E_1)-I\right)\Psi(0)+\i\left(\S(E_1)+I\right)\Psi'(0)=0\qquad\text{for } E_1=\frac{\hbar^2}{2m}\,.
\end{equation}
\end{observation}
\begin{proof}
It suffices to substitute for $\S(E_1)$ from~\eqref{S0}, and then to multiply the equation by the matrix $A+\i B$ from left.
\end{proof}

\subsection{F\"ul\"op--Tsutsui couplings}

Vertex couplings with $\S(E)$ independent of $E$ are called \emph{scale invariant}, or \emph{F\"ul\"op--Tsutsui}, couplings~\cite{FT00, NS00, SS02}. In this paper we use the latter term, since the name ``scale invariant'' sometimes leads to misinterpretations.
It can be easily shown that the scattering matrix $\S(E)$ corresponding to the boundary conditions~\eqref{ST} is energy-independent if and only if the Hermitian matrix $S$ in~\eqref{ST} satisfies $S=0$, i.e., iff the boundary conditions take the form
\begin{equation}\label{FT}
\left(\begin{array}{cc}
I^{(r)} & T \\
0 & 0
\end{array}\right)\Psi'(0)=
\left(\begin{array}{cc}
0 & 0 \\
-T^* & I^{(n-r)}
\end{array}\right)\Psi(0)\,.
\end{equation}
The scattering matrix corresponding to b.c.~\eqref{FT} can be explicitely expressed using formula~\eqref{S} to which we substitute
$A=-\begin{pmatrix}
0 & 0 \\
-T^* & I^{(n-r)}
\end{pmatrix}$,
$B=\begin{pmatrix}
I^{(r)} & T \\
0 & 0
\end{pmatrix}$:
\begin{equation}\label{SFT}
\S(E)=\left(\begin{array}{cc}
\left(I^{(r)}+TT^*\right)^{-1} \left(I^{(r)}-TT^*\right) & \left(I^{(r)}+TT^*\right)^{-1}2 T \\
\left(I^{(n-r)}+T^*T\right)^{-1}  2 T^* & -\left(I^{(n-r)}+T^*T\right)^{-1} \left(I^{(n-r)}-T^*T\right)
\end{array}\right)\,.
\end{equation}
This result can be found for example in~\cite{CT10}.

\subsection{Filtering property of the $\delta$-interaction}

The $\delta$-interaction can serve for a simple spectral filtering (Fig.~\ref{d-line}). This effect can be demonstrated with the help of the scattering matrix.
\begin{figure}[h]
\begin{center}
  \includegraphics[width=4.5cm]{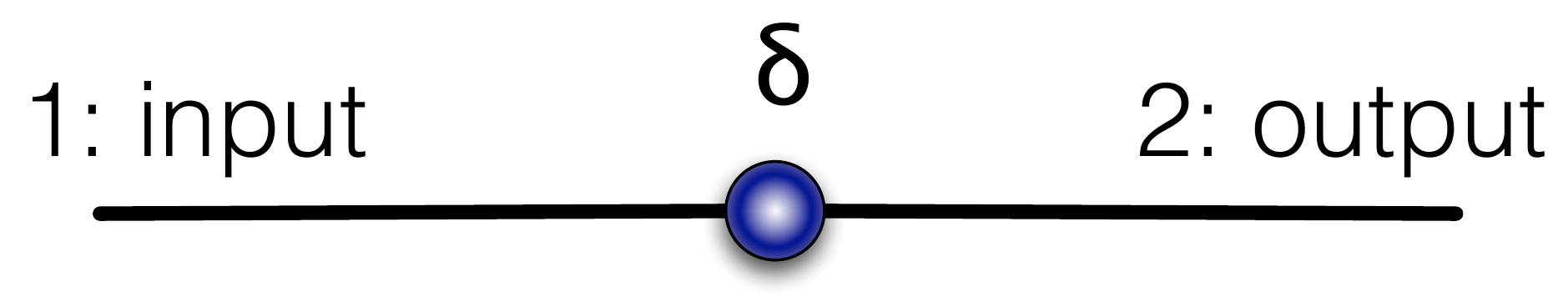}
\caption{$\delta$-interaction on the line as a high-pass spectral filter.}
\label{d-line}
\end{center}
\end{figure}
The boundary conditions~\eqref{delta} are equivalent to $-\begin{pmatrix}\alpha&0\\-1&1\end{pmatrix}\Psi(0)+\begin{pmatrix}1&1\\0&0\end{pmatrix}\Psi'(0)=0$, hence, using equation~\eqref{S0},
\begin{equation}
\S(E)=\left(\begin{array}{cc}
\frac{-\i\alpha}{2\sqrt{E}+\i\alpha} & \frac{2\sqrt{E}}{2\sqrt{E}+\i\alpha} \\
\frac{2\sqrt{E}}{2\sqrt{E}+\i\alpha} & \frac{-\i\alpha}{2\sqrt{E}+\i\alpha}
\end{array}\right)\,,
\end{equation}
where $\alpha\in\R\backslash\{0\}$ is the parameter of the $\delta$-interaction.
The transmission amplitude input~$\to$~output (1~$\to$~2) is given by $\S_{21}(E)$, hence we obtain the transmission probability
\begin{equation}
\P(E)=|\S_{21}(E)|^2=\frac{4E}{4E+\alpha^2} \, ,
\end{equation}
which satisfies
\begin{equation}
\text{$\P(E)\approx0$ \quad for $E\ll\alpha^2$ \qquad and \qquad $\P(E)\approx1$ \quad for $E\gg\alpha^2$}\,.
\end{equation}
To sum up, the system exhibits full reflection for small energies and full transmission for high energies, and, therefore, can be used as a high-pass spectral filter.

\section{Potential-controlled quantum device}\label{Section: Device}

The main goal of this paper is to design a quantum filter that is controllable by an external potential, as schematically illustrated in Figure~\ref{Device}. Any such 
device has, therefore, three (or more) lines that are used in the following way:
\begin{itemize}
\item Line 1 is \emph{input}. Particles of various energies are coming in the device along this line.
\item Line 2 is \emph{output}. Particles passed through the device are gathered on this line.
\item Line 3 is \emph{controlling line}. We assume that this line is subjected to a constant (but adjustable) external potential $U$. Adjustment of the potential directly controls the flow from the input to the output.
\end{itemize}
The concept can be naturally extended to devices having more than one controlling line, if more device parameters are to be controlled. Similarly, devices with several output lines, or several input lines, can be considered. We will deal with such cases in Sections~\ref{Section: n=4}--\ref{Section: Branching}.

\begin{figure}[h]
  \centering
  \includegraphics[width=4.5cm]{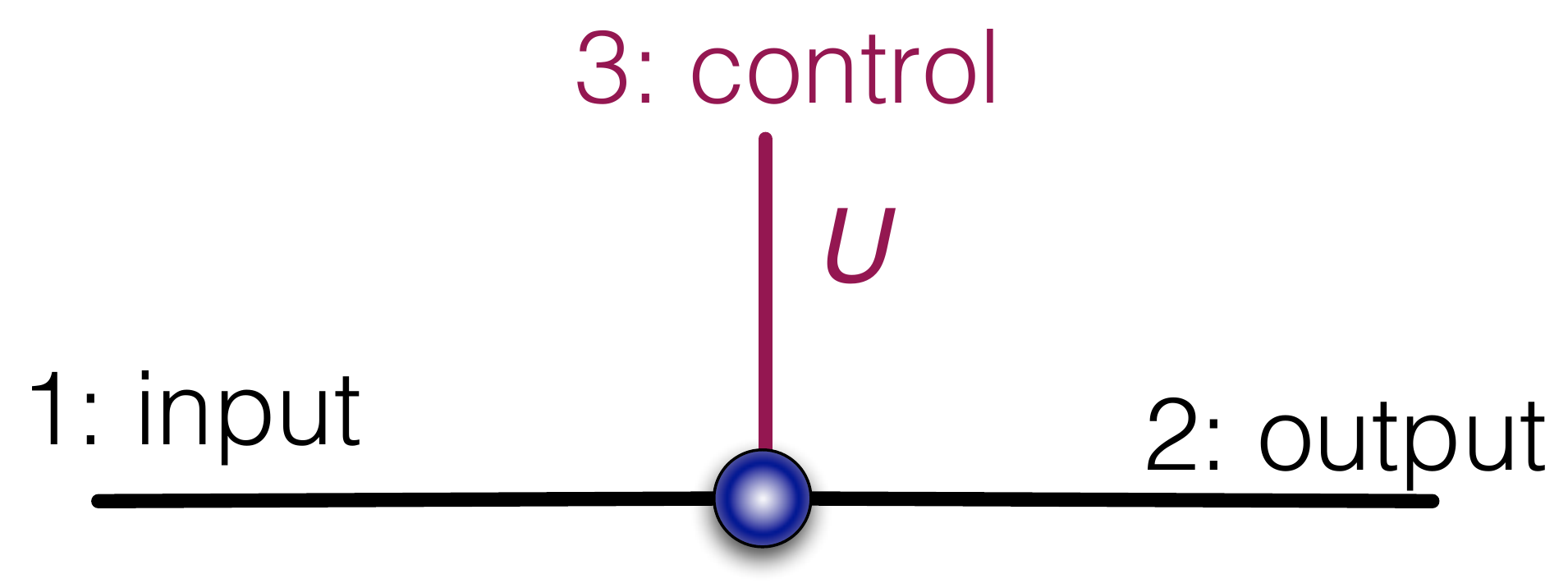}
  \caption{Schematic depiction of a quantum device controllable by an external potential. The transmission through the channel 1~$\to$~2 is being controlled by adjusting a potential $U$ on the controlling line 3.}
  \label{Device}
\end{figure}
It is desirable that the device is technically simple. We satisfy this requirement by designing the device 
modeled by a quantum star graph. In its vertex, we assume a special, F\"ul\"op--Tsutsui coupling. The parameters of the coupling will be specified later. They determine the transmission characteristics of the device.

In order to describe the effect of the adjustable control potentials on the scattering properties of the graph, we need to generalize the scattering matrix formula~\eqref{SFT}.

\begin{proposition}\label{S-matrix FT}
Let the vertex coupling in the center of a star graph is given by~\eqref{FT} and there be constant potentials $V_1,\ldots,V_n$ on the graph lines. Then the scattering matrix is given by
\begin{equation}\label{SFTpot}
\S(E;V_1,\ldots,V_n)=
-I^{(n)}+2\left(\begin{array}{c}
Q_{(1)} \\
Q_{(2)}T^*
\end{array}\right)
\left(Q_{(1)}^2+TQ_{(2)}^2T^*\right)^{-1}
\left(\begin{array}{cc}
Q_{(1)} & TQ_{(2)}
\end{array}\right)\,,
\end{equation}
where
\begin{equation}
Q_{(1)}=\diag\left(\sqrt[4]{1-\frac{V_1}{E}},\ldots,\sqrt[4]{1-\frac{V_r}{E}}\right)\,,\quad Q_{(2)}=\diag\left(\sqrt[4]{1-\frac{V_{r+1}}{E}},\ldots,\sqrt[4]{1-\frac{V_n}{E}}\right)\,.
\end{equation}
\end{proposition}
\begin{proof}
The result can be obtained from equation~\eqref{S} in a similar manner as formula~\eqref{SFT}.
\end{proof}

\section{Potential-controlled spectral band-pass filter}\label{Section: n=3}

The simplest device we are going to construct is based on a star graph with just $3$ lines. The F\"ul\"op--Tsutsui coupling in its center is described by the boundary condition~\eqref{FT} with two additional assumptions: (i) $r=1$, (ii) $T$ is real. The reason for assuming $T$ real will become obvious later in Section~\ref{Section: Approximation}. Hence $T=(a \quad b)$, $a,b\in\R$, and the boundary conditions read
\begin{equation}
\label{bc}
  \begin{pmatrix}
  1 & a & b \\ 0 & 0 & 0 \\ 0 & 0 & 0
  \end{pmatrix}
  \begin{pmatrix}
  \psi'_1(0) \\ \psi'_2(0) \\ \psi'_3(0)
  \end{pmatrix}
=
  \begin{pmatrix}
  0 & 0 & 0 \\ -a & 1 & 0 \\ -b & 0 & 1
  \end{pmatrix}
  \begin{pmatrix}
  \psi_1(0) \\ \psi_2(0) \\ \psi_3(0)
  \end{pmatrix}\,.
\end{equation}
The graph is schematically illustrated in Fig.~\ref{fig:m1}.
The roles of individual lines are as stated in Section~\ref{Section: Device}: 1 = input, 2 = output, 3 = controlling line.
\begin{figure}[h]
  \centering
  \includegraphics[width=4.5cm]{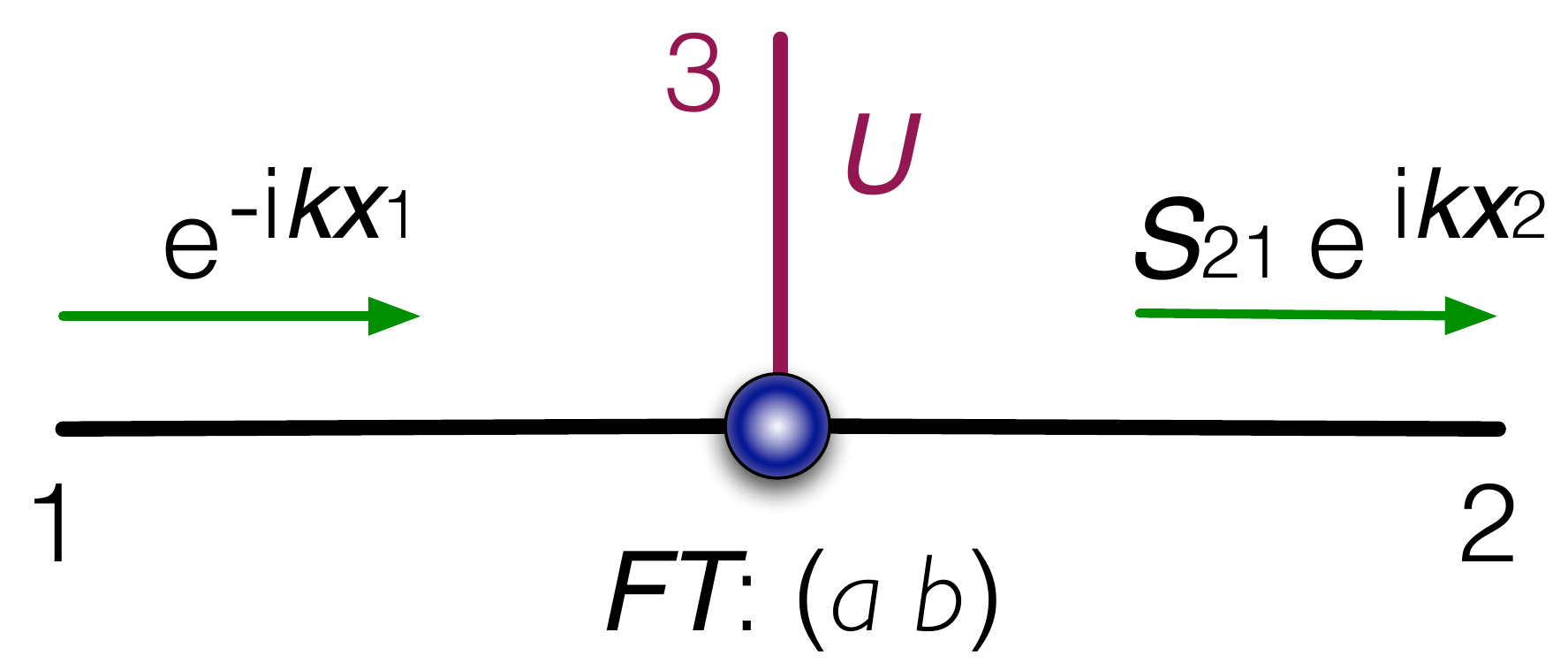}
  \caption{Scheme of a quantum spectral filter controllable by an external potential $U$ on the line 3.}
  \label{fig:m1}
\end{figure}
The scattering properties of the graph in question are determined by the scattering matrix $\S(E;0,0,U)$, which can be calculated using Proposition~\ref{S-matrix FT}. If we substitute $T=(a \quad b)$ and $V_1=V_2=0$, $V_3=U$ into equation~\eqref{SFTpot}, we easily obtain
\begin{equation}\label{S-matrix 3}
\S(E;0,0,U)=\frac{1}{1+a^2+b^2\xi}
\left(\begin{array}{ccc}
1-a^2-b^2\xi & 2a & 2b\sqrt{\xi} \\
2a & a^2-1-b^2\xi & 2ab\sqrt{\xi} \\
2b\sqrt{\xi} & 2ab\sqrt{\xi} & b^2\xi-1-a^2
\end{array}\right)\,,
\end{equation}
where $\xi=\sqrt{1-U/E}$ (thus $\sqrt{\xi}=\sqrt[4]{1-U/E}$).

A quantum particle with energy $E$ coming in the vertex from the input line~1 is scattered at the vertex into all the lines. The corresponding scattering amplitudes are given by the entries in the first column of~\eqref{S-matrix 3}. Let us denote for simplicity $\S_{j1}^{(U)}(E)\equiv[\S(E;0,0,U)]_{j1}$ for all $j=1,2,3$. The transmission amplitude input~$\to$~output equals
\begin{equation}
\label{S21}
\S_{21}^{(U)}(E) = \frac{2a}{ 1+a^2+ b^2 \sqrt{1-\frac{U}{E} } }\,,
\end{equation}
the reflection amplitude is
\begin{equation}
\S_{11}^{(U)}(E) = \frac{1-a^2- b^2 \sqrt{1-\frac{U}{E} } }{ 1+a^2+ b^2 \sqrt{1-\frac{U}{E} } }\,,
\end{equation}
and the transmission amplitude 1~$\to$~3 is given by
\begin{equation}\label{S31}
\S_{31}^{(U)}(E) = \frac{2b  \sqrt[4]{1-\frac{U}{E}} }
{ 1+a^2+ b^2 \sqrt{1-\frac{U}{E} } }\,\Theta(E-U)\,.
\end{equation}
The Heaviside step function $\Theta(E-U)$ is added in equation~\eqref{S31} to make the expression valid for all energies $E$, including $E<U$. It represents asymptotically no transmission to the line 3 below the threshold energy $E_\mathrm{th}=U$.

We are interested above all in the transmission probability from the input line~1 into the output line~2, which we denote by $\P^{(U)}(E)$.
Since $\P^{(U)}(E)=|\S_{21}^{(U)}(E)|^2$,
we have from~\eqref{S21}
\begin{equation}\label{P:3}
\P^{(U)}(E)=\left\{\begin{array}{cl}
\frac{4a^2}{\left(1+a^2+b^2\sqrt{1-U/E}\right)^2} & \quad\text{for } E>U, \\ [1em]
\frac{4a^2}{(1+a^2)^2+b^4(U/E-1)} & \quad\text{for } E<U.
\end{array}\right.
\end{equation}
We observe that for a given constant potential on the line~3, $\P^{(U)}(E)$ as a function of $E$ grows in the interval $(0,U)$,
decreases in the interval $(U,\infty)$,
and satisfies
\begin{align}
&
\lim_{E\to0}\P^{(U)}(E) = 0\,,
 \\
&
\lim_{E\to U}\P^{(U)}(E) = \left(\frac{2a}{1 + a^2}\right)^2 \label{P(sqrt(U),U)}\,,
\\
&
\lim_{E\to\infty}\P^{(U)}(E) = \left(\frac{2a}{1 + a^2 + b^2}\right)^2\,.
\end{align}
Therefore, if the parameters $a,b$ are chosen so that $b^4 \gg 4a^2$, the function $\P^{(U)}(E)$ has a sharp peak at $E=U$. Furthermore, equation~\eqref{P(sqrt(U),U)} implies that the peak is highest possible (attaining $1$) for $a=1$.
To sum up, the device from Figure~\eqref{fig:m1} built for $a,b$ satisfying
\begin{equation}\label{a=1,b>>1}
a=1 \quad \text{and} \quad b^4/4\gg1
\end{equation}
has the following property: The transmission probability input~$\to$~output is high for particles having energies $E\approx U$ and perfect for $E=U$, while it is very small for particles with other energies. It means that such a device works as an adjustable band-pass spectral filter, allowing to control the passband position by the potential put on the controlling line 3. The situation is numerically illustrated in Figure~\ref{fig:m2}.

Let~\eqref{a=1,b>>1} be satisfied. Filtering properties of the device, in particular the sharpness of the peak, are determined by the quantity $\beta:=b^2/2$.
With regard to the above results, it holds
$\P^{(U)}(E)\approx 0$ for $E\to0$ and $\P^{(U)}(E)\approx 1/(1+\beta)^2$ for $E\to\infty$.
The bandwidth $W$ of the filter, i.e., the width of the interval of energies $E$ for which $\P^{(U)}(E)>1/2$, can be calculated from equation~\eqref{P:3}. The bandwidth depends on $U$ and $\beta^2$, and equals
\begin{equation}
W=\frac{2\left(2-\sqrt{2}\right)\beta^2}{\left(\beta^2-3+2\sqrt{2}\right)(\beta^2+1)}U\,.
\end{equation}
Since $\beta^2\gg1$, it holds $W\approx 2(2-\sqrt{2})U/\beta^2 \approx 1.17 U/\beta^2$.

\begin{figure}[h]
  \centering
  \includegraphics[width=6.5cm]{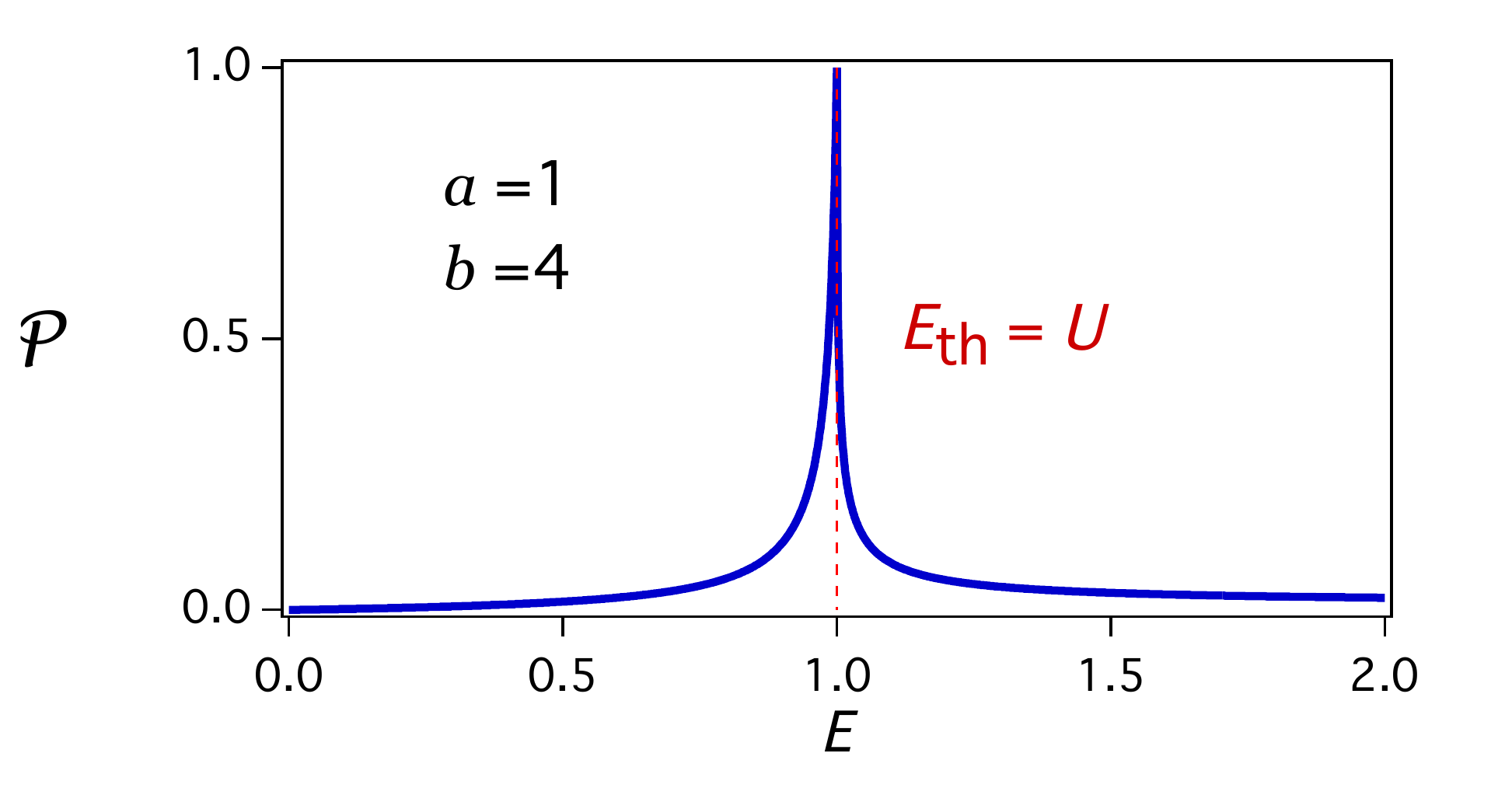} \quad \includegraphics[width=6.5cm]{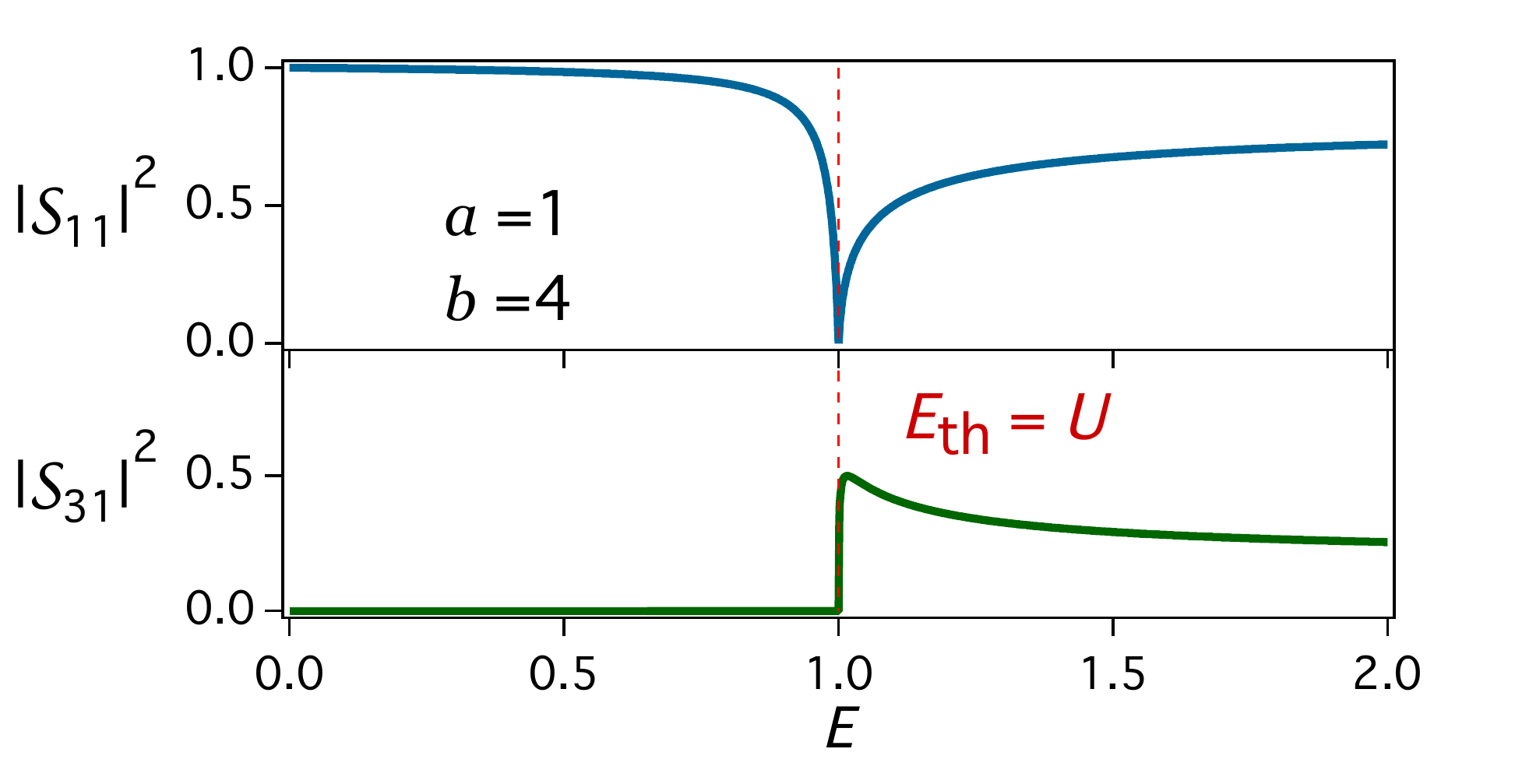}
  \caption{Scattering characteristics of the graph from Fig.~\ref{fig:m1} with parameters $a=1$, $b=4$ (i.e., $\beta=8$), plotted for the control potential set to $U=1$. The transmission probability $\P^{(U)}(E)$ as a function of $E$ is plotted in the left figure. The right figure shows the reflection probability $|\S_{11}^{(U)}(E)|^2$ and the probability of transmission to the controlling line $|\S_{31}^{(U)}(E)|^2$.}
  \label{fig:m2}
\end{figure}

\begin{remark}
The resonance at the threshold energy $E_\mathrm{th}=U$
is related to the pole of the scattering matrix which is located on the positive real axis at
\begin{equation}
E_{\mathrm{pol}}=\frac{b^4}{b^4-(1+a^2)^2}U
\end{equation}
on the unphysical Riemann surface which is connected to the physical Riemann surface at 
$E=U$.
\end{remark}

\begin{remark}
A \emph{band-stop} spectral filter can be constructed in a similar way. If we begin with $r=2$ and $T=\begin{pmatrix} c \\ d \end{pmatrix}$, and choose $c,d\in\R$ such that $c=d\gg1$, we obtain
\begin{equation}
\P^{(U)}(E)=\frac{4c^4}{\left|2c^2+\sqrt{\frac{E}{E-U}}\right|^2}\,,
\end{equation}
hence
$\P^{(U)}(E)\approx1$ everywhere except a certain narrow interval around $E=U$,
and $\lim_{E\to U}\P^{(U)}(E)\approx0$.
\end{remark}

\section{The role of the control potential}\label{Section: Potential}

The filtering property of the device proposed in Section~\ref{Section: n=3} has been derived on the assumption that the control potential $U$ is constant on the whole line 3. A natural question arises: What happens when the control potential is not constant precisely up to the junction and/or $U(x)$ is supported only by a certain segment of the line 3?
In this section we show that under certain conditions, the filtering property is not distorted. To obtain the result, we prove, for a general potential $U(x)$, that the transmission probability $\P^{(U(x))}(E)$ can be expressed in terms of the reflection amplitude calculated for the potential $U(x)$ on the real line. The formula we obtain will allow not only to answer the above question, but also to better understand the function of the filter.

In order to find the transmission amplitude $\S^{(U(x))}_{21}(E)$ for a given potential $U(x)$ on the controlling line, we need to apply boundary conditions~\eqref{bc} on the wave function components $\psi_1(x)$, $\psi_2(x)$, $\psi_3(x)$. The particle motion on the lines 1 and 2 is free, thus the boundary values $\psi_1(0),\psi_1'(0)$, $\psi_2(0),\psi_2'(0)$ at the vertex can be expressed in terms of the scattering amplitudes, cf. equation~\eqref{psi_ji},
\begin{equation}\label{psi12}
\begin{array}{ccc}
\psi_1(0)=\frac{1}{\sqrt{k}}(-1+\S_{11})\,, & \qquad & \psi_2(0)=\frac{1}{\sqrt{k}}\S_{21}\,, \\
\psi_1'(0)=\i\sqrt{k}(-1+\S_{11})\,, & \qquad & \psi_2'(0)=\i\sqrt{k}\S_{21}\,,
\end{array}
\end{equation}
where $\S_{11}$ is the reflection amplitude to the input line 1, $\S_{21}$ is the transmission amplitude to the output line 2, and
\begin{equation}
k=\frac{\sqrt{2mE}}{\hbar}\,.
\end{equation}
Our next aim is to express the boundary values $\psi_3(0)$ and $\psi_3'(0)$ in a convenient way. For that purpose, 
let us consider a line with the potential
\begin{equation}
\tilde{U}(x)=\left\{\begin{array}{cl}
U(x) & \text{for } x>0\,, \\
0 & \text{for } x\leq0\,,
\end{array}\right.
\end{equation}
see Figure~\ref{Fig.potential}. For $\psi_3(x)$ being the given wave function component on the controlling line, we define the function \begin{equation}\label{tilde psi}
\tilde{\psi}(x)=\left\{\begin{array}{cl}
\psi_3(x) & \text{for } x\geq0\,, \\
C\left(\frac{1}{\sqrt{k}}\e^{\i kx}+\mathcal{R}\frac{1}{\sqrt{k}}\e^{-\i kx}\right) & \text{for } x\leq0\,,
\end{array}\right.
\end{equation}
where $C$ and $\mathcal{R}$ are determined by the continuity of $\tilde{\psi}(x)$ and $\tilde{\psi}'(x)$ at $x=0$:
\begin{equation}\label{psi3}
\frac{1}{\sqrt{k}}C\left(1+\mathcal{R}\right)=\psi_3(0)\,, \qquad \i \sqrt{k}C\left(1-\mathcal{R}\right)=\psi_3'(0)\,.
\end{equation}
\begin{figure}[h]
  \centering
  \includegraphics[width=5.0cm]{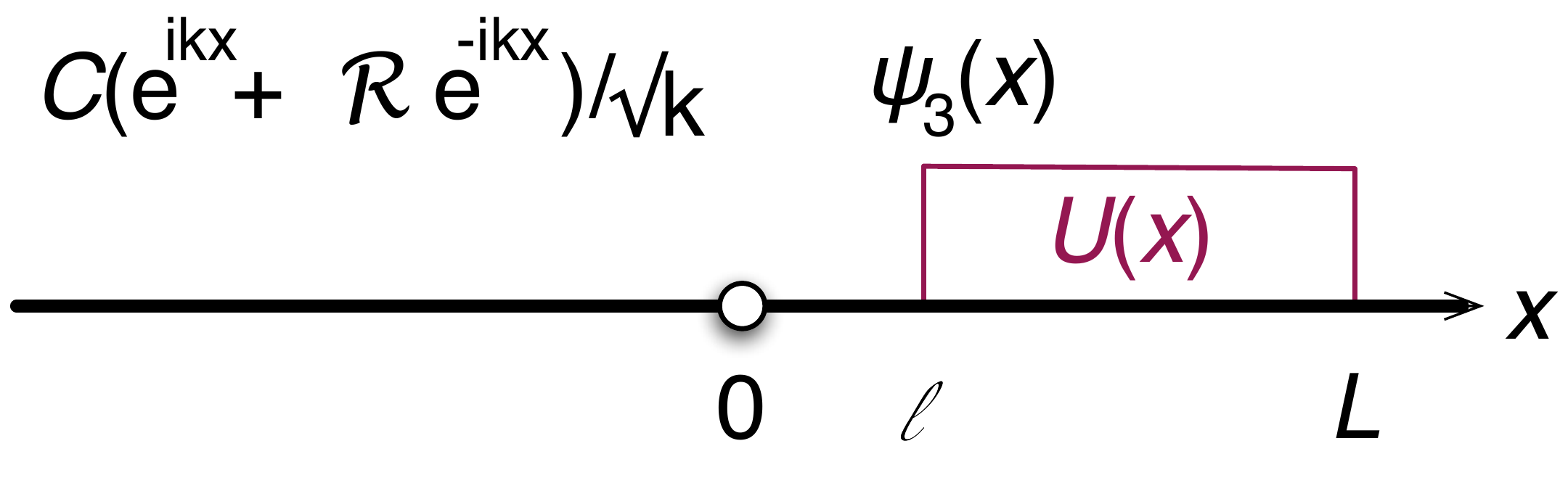}
  \caption{A line with the potential $U(x)$ in the part $x>0$.}
  \label{Fig.potential}
\end{figure}
The function $\tilde{\psi}(x)$ is obviously a solution of the Schr\"odinger equation for a particle on the line with the potential $\tilde{U}(x)$.
The term $C(\e^{\i kx}/\sqrt{k}+\mathcal{R}\e^{-\i kx}/\sqrt{k})$ corresponding to $x\leq0$ in definition~\eqref{tilde psi} implies that $\tilde{\psi}(x)$ is the wave function of a particle of energy $E$ moving along the line from the part $x<0$ to the part $x>0$, and $\mathcal{R}$ represents the reflection amplitude at the point $x=0$.

\begin{remark}\label{psi'/psi}
From equations~\eqref{psi3} it follows
$\mathcal{R}=\left(\i k\psi_3(0)-\psi_3'(0)\right)/\left(\i k\psi_3(0)+\psi_3'(0)\right)$. This formula can be generalized. It is easy to show that for any $x_0\geq0$, it holds
\begin{equation}
\mathcal{R}(x_0)=\frac{\i k\psi_3(0)-\psi_3'(0)}{\i k\psi_3(0)+\psi_3'(0)}\,,
\end{equation}
where $\mathcal{R}(x_0)$ represents the reflection amplitude corresponding to the shifted and truncated potential $U_{x_0}(x)=U(x-x_0)\cdot\Theta(x-x_0)$, measured at the point $x=0$.
\end{remark}

Now we use equations~\eqref{psi12} and~\eqref{psi3} to subsitute for $\psi_j(0)$ and $\psi_j'(0)$, $j=1,2,3$, in the boundary conditions~\eqref{bc}. After an easy manipulation we obtain the system of equations
\begin{equation}
\begin{pmatrix}
1 & a & b(1-\mathcal{R}) \\
a & -1 & 0 \\
b & 0 & -(1+\mathcal{R})
\end{pmatrix}
\begin{pmatrix}
\S_{11} \\
\S_{21} \\
C
\end{pmatrix}
=
\begin{pmatrix}
1 \\
-a \\
-b
\end{pmatrix}\,,
\end{equation}
therefore, the transmission amplitude input~$\to$~output equals
\begin{equation}
\S^{(U(x))}_{21}(E)=\frac{2a(1+\mathcal{R})}{(1+a^2)(1+\mathcal{R})+b^2(1-\mathcal{R})}\,.
\end{equation}
If $a=1$ and $\beta:=b^2/2$ as introduced in Section~\ref{Section: n=3}, then
\begin{equation}\label{P U(x)}
\P^{(U(x))}(E)=\left|\frac{1+\mathcal{R}}{1+\mathcal{R}+\beta(1-\mathcal{R})}\right|^2\,.
\end{equation}
Equation~\eqref{P U(x)} relates the effect of the control potential, encapsulated in the reflection amplitude $\mathcal{R}$,
and the actual transmission probability $\P^{(U(x))}(E)$ of the filter.

Since $\beta\gg1$ by assumption, we observe from equation~\eqref{P U(x)} that the filter is open for the particles that would be reflected by the potential $U(x)$ with the amplitude $\mathcal{R}\approx+1$, and closed for particles with $|\beta(1-\mathcal{R})|\gg1$. We can say that the control potential opens the filter by reflecting the particles entering the line 3 back to the vertex with the amplitude $+1$. In particular, if the particle has high enough energy to pass through the potential barrier $U(x)$, we have $\mathcal{R}\approx0$, hence $\P^{(U(x))}(E)\approx1/(1+\beta)^2\approx0$.

In case that the potential is finitely supported and/or separated from the vertex by a gap, the reflection amplitude $\mathcal{R}$ is usually oscillating, thus $\P^{(U(x))}$ is oscillating as well. Furthermore, if the constant part of the potential is long enough, the oscillations are very rapid and, therefore, imperceptible (the ``semiclassical'' short wavelength limit). In such cases we mollify the quickly oscillating function $\P^{(U(x))}(E)$ by a convolution with the Friedrichs' mollifier $\omega_\epsilon$ for a certain small $\epsilon>0$, where
\begin{equation}
\omega_\epsilon(z)=\left\{\begin{array}{cl}
\frac{1}{\epsilon}C_1\e^{-\frac{\epsilon^2}{\epsilon^2-z^2}} & \text{for } |z|<\epsilon\,, \\
0 & \text{otherwise}
\end{array}\right.
\end{equation}
and $C_1$ is chosen such that $\int_{-1}^1 C_1\e^{-\frac{1}{1-z^2}}\d z=1$. As we will see in examples below, the obtained mollified function
\begin{equation}
\P^{(U(x))}_\mathrm{moll}(E)=\left(\omega_\epsilon\ast\P^{(U(x))}\right)(E)
\end{equation}
is essentially identical with the transmission probability input~$\to$~output in the ideal case when the control potential is constant on the whole line 3.

\subsubsection*{Example 1: Control potential supported by a finite segment $[0,L]$}

If $U_{[0,L]}=U\cdot\chi_{[0,L]}$ for $\chi_{[0,L]}$ being the characteristic function of the interval $[0,L]$, we have
\begin{equation}\label{R[0,L]}
\mathcal{R}_{[0,L]}=\frac{(1-\xi^2)\sin\xi kL}{(1+\xi^2)\sin\xi kL+2\i\xi\cos\xi kL}\,,
\end{equation}
where $\xi=\sqrt{1-U/E}$.
It can be shown that the transmission probability, given by equation~\eqref{P U(x)} with $\mathcal{R}$ obeying~\eqref{R[0,L]}, grows in the interval of energies $[0,U]$. Furthermore, if the support of $U(x)$ is long enough in the sense $L\gg\hbar/\sqrt{2mU}$, then $kL\gg1$ for all particles with energies $E>U$, and one can prove that
\begin{itemize}
\item $\lim_{E\to U}\P^{(U_{[0,L]})}(E)\approx1$,
\item in the region of energies $E\in[U,\infty)$, the transmission probability $\P^{(U_{[0,L]})}(E)$ quickly oscillates between $\P_{\min}\approx1/(1+\beta)^2\approx0$ and $\P_{\max}\approx1/(1+(1-U/E)\beta)^2$.
\end{itemize}
The situation is illustrated in Figure~\ref{Fig. Finite supp} (left). The figure shows also the mollified transmission probability $\P^{(U_{[0,L]})}_\mathrm{moll}(E)=(\omega_\epsilon\ast\P^{(U_{[0,L]})})(E)$. The function $\P^{(U_{[0,L]})}_\mathrm{moll}(E)$
is very close to the transmission characteristics obtained in Section~\ref{Section: n=3} for the control potential constant on the whole line~3.

The condition $L\gg\hbar/\sqrt{2mU}\equiv\lambda_U/(2\pi)$ means that the support of the control potential is substantially longer than the de Broglie wavelength of the particle with energy $U$. Such a condition is usually easily satisfied by any macroscopical length $L$. Indeed, if we consider an electron and, for instance, $U=0.1$~eV, then $\hbar/\sqrt{2mU}\approx 2\cdot10^{-10}\,\text{m}=0.6\,\text{nm}$.

\begin{figure}[h]
  \centering
  \includegraphics[width=6.5cm]{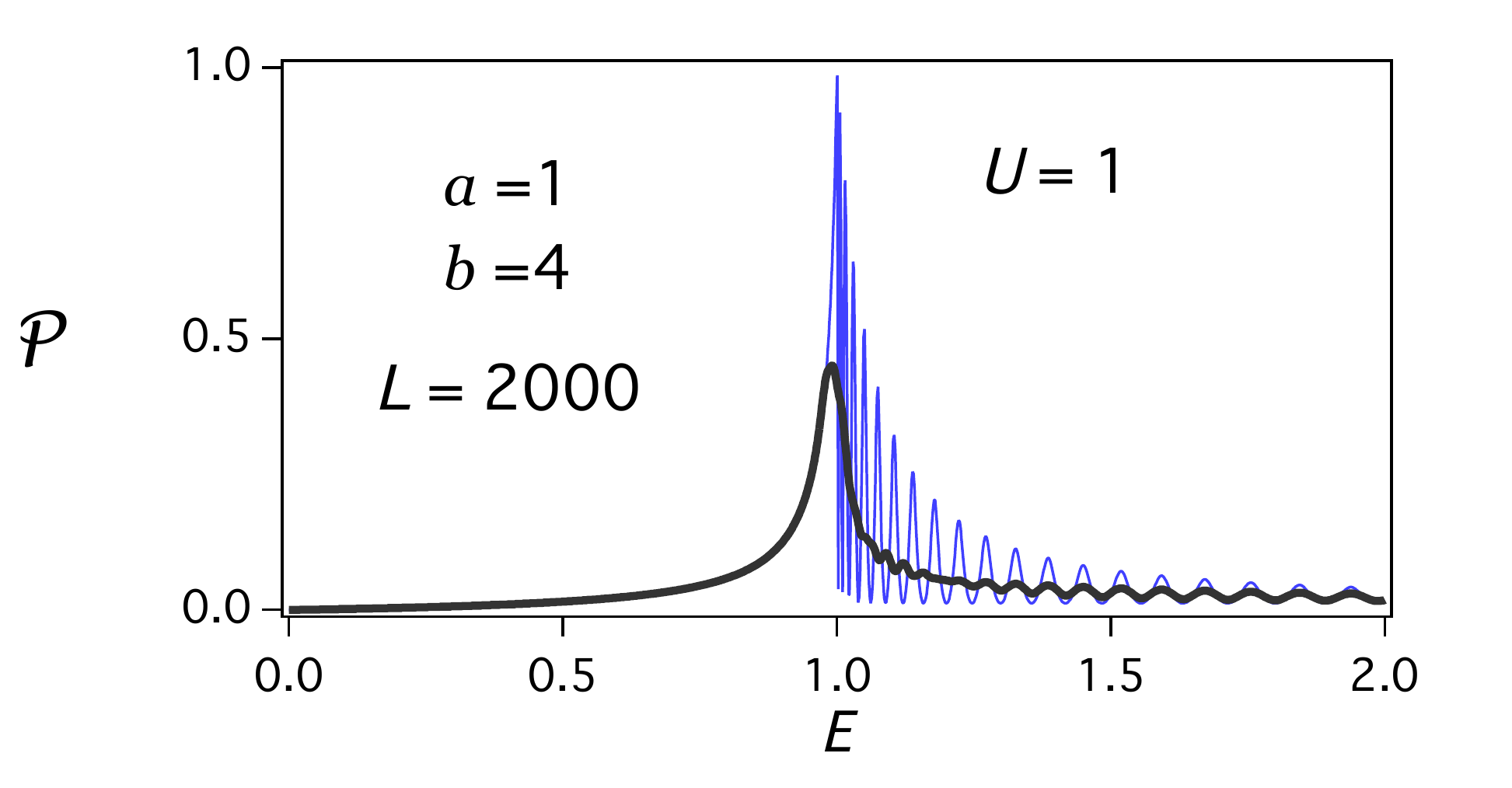} \quad \includegraphics[width=6.5cm]{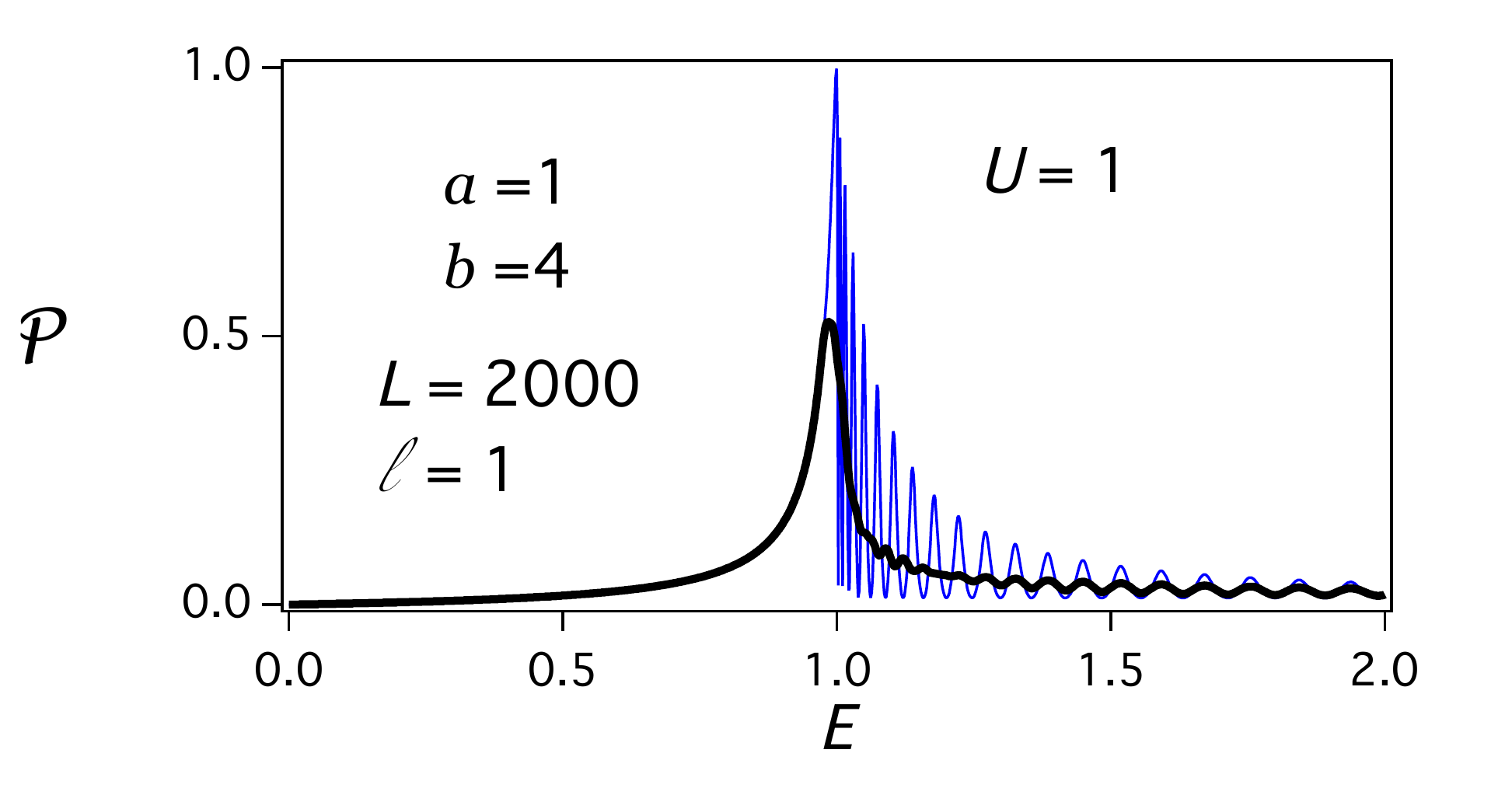}
  \caption{The transmission characteristics of the quantum filter designed in Sect.~\ref{Section: n=3}, controlled by a potential supported by the finite segment $[0,L]$ (left) and $[\ell,L]$ (right). 
  The mollified transmission probabilities are plotted by thick black lines.
  The parameters are chosen as $U=1$, $L=2000$, $\ell=1$.  If we assume the particle to be an electron, the energy scale to be given by {meV}, milielectronvolt, and the length scale to be specified by {nm}, nanometer, then $U=1 \, {\rm meV}$ corresponds to the wavelength $\lambda = h/\sqrt{2m_e U} = 38.8 \, {\rm nm}$. For $L\gg\lambda/(2\pi)\gg d$, the mollified transmission probability is essentially identical to the transmission characteristics obtained in the ideal case $L=\infty$, cf. Fig.~\ref{fig:m2}.
  }
  \label{Fig. Finite supp}
\end{figure}

\subsubsection*{Example 2: Control potential supported by a finite segment $[\ell,L]$, $\ell>0$}

In the second example we focus on a control potential with the support separated from the vertex by a gap of the length $d$, i.e., $U_{[\ell,L]}=U\cdot\chi_{[\ell,L]}$; the corresponding coefficient $\mathcal{R}_{[\ell,L]}$ equals
\begin{equation}
\mathcal{R}_{[\ell,L]}=\frac{(1-\xi^2)\sin\xi k(L-\ell)}{(1+\xi^2)\sin\xi k(L-\ell)+2\i\xi\cos\xi k(L-\ell)}\cdot\e^{2\i k\ell}\,.
\end{equation}
If the lengths $L,\ell$ satisfy
\begin{equation}\label{dL}
L\gg\frac{\hbar}{\sqrt{2mU}} \qquad\text{and}\qquad \ell\ll\frac{\hbar}{\sqrt{2mU}}\,,
\end{equation}
then the values of $\mathcal{R}_{[\ell,L]}$ are close to the values of $\mathcal{R}_{[0,L]}$ obtained in Example 1 (cf. eq.~\eqref{R[0,L]}). Indeed,
\begin{itemize}
\item for $E<U$ or $E\approx U$, we have $k=\sqrt{2mE}/\hbar\leq\sqrt{2mU}/\hbar$, hence $2k\ell\ll1$ due to the second equation~\eqref{dL}, hence $\e^{2\i k\ell}\approx1$;
\item for $E\gg U$, it holds $\xi\approx1$, hence $\mathcal{R}_{[\ell,L]}\approx0\approx\mathcal{R}_{[0,L]}$.
\end{itemize}
Therefore, the transmission characteristics of the filter with $U_{[\ell,L]}=U\cdot\chi_{[\ell,L]}$ are similar to those described in Example 1, see Figure~\ref{Fig. Finite supp} (right). Namely, the transmission probability $\P^{(U_{[\ell,L]})}(E)$ is quickly oscillating, but the mollified transmission probability $\P^{(U_{[\ell,L]})}_\mathrm{moll}(E)=(\omega_\epsilon\ast\P^{(U_{[\ell,L]})})(E)$ well approximates the ideal characteristics obtained in Section~\ref{Section: n=3}.
The conditions~\eqref{dL} mean that $\ell\ll\lambda_U/(2\pi)\ll L$, where $\lambda_U=h/\sqrt{2mU}$ is the de Broglie wavelength of the particle with energy $U$.

\medskip

In practical implementations, the profile of the control potential $U(x)$ might differ from those studied in Examples 1 and 2 above. In particular, the potential may be constant in the interval $[\ell,L]$ and taper off to zero in the segment $[0,\ell]$. 
Generally, if the conditions~\eqref{dL} are satisfied, the result is essentially independent of the shape of the fall-off of $U(x)$. To demonstrate 
this, let us distinguish two situations.
\begin{itemize}
\item If $E<U$ or $E\approx U$, the local wavelength $\lambda(x)=h/\sqrt{2m(E-U(x))}$ of the particle in every point $x\in[0,\ell]$ satisfies $\lambda(x)/(2\pi)\gg\ell$ (the ``long wavelength limit''), hence $\psi_3(0)\approx\psi_3(\ell)$, $\psi_3'(0)\approx\psi_3'(\ell)$. 
Therefore, with regard to Remark~\ref{psi'/psi},
\begin{equation}
\mathcal{R}=\frac{\i k\psi_3(0)-\psi_3'(0)}{\i k\psi_3(0)+\psi_3'(0)}\approx\frac{\i k\psi_3(\ell)-\psi_3'(\ell)}{\i k\psi_3(\ell)+\psi_3'(\ell)}=\mathcal{R}_{[0,L-\ell]}\,,
\end{equation}
where $\mathcal{R}_{[0,L-\ell]}$ is the reflection amplitude corresponding to the potential $U_{[0,L-\ell]}=U\cdot\chi_{[0,L-\ell]}$.
Since moreover $L-\ell\gg\lambda_U/(2\pi)$, it follows from Example 1 that the corresponding mollified transmission probability is similar to the transmission characteristics found in Section~\ref{Section: n=3}.
\item If $E\gg U$, the particle is almost fully transmitted through the potential barrier, hence $\mathcal{R}\approx0$. Therefore, $\P^{(U(x))}(E)\approx1/(1+\beta)^2$ due to equation~\eqref{P U(x)}, which coincides with the result obtained in Section~\ref{Section: n=3}.
\end{itemize}

To sum up, the filter works well also for a finitely supported potential which is not constant precisely up to the junction, on condition that the support of the potential is long and the segment in which the potential tapers off is short, both with respect to the de Broglie wavelength of the particle with energy $U$.
A similar result holds true also for the other filtering devices 
described in subsequent sections.


\section{Spectral filter with two passbands}\label{Section: n=4}

In this section we design a band-pass quantum filter with two passbands such that their positions are independently controllable by two external potentials.

The device will be naturally based on a star graph with $4$ lines that have the following meaning: 1 = input, 2 = output, 3 and 4 = controlling lines subjected to constant external potentials $U,V$; see Figure~\ref{Fig.4}.

The vertex coupling in the graph center is of the F\"ul\"op--Tsutsui type in accordance with the concept introduced in Section~\ref{Section: Device}. This time we start with an ansatz $r=2$, i.e., $T=\begin{pmatrix}a & b \\ c & d\end{pmatrix}$, and, as in the previous section, assume that $T$ is real, thus $a,b,c,d\in\R$. Therefore, the vertex is described by the boundary conditions
\begin{equation}\label{b.c. 4}
  \begin{pmatrix}
  1 &  0 & a & b \\ 0 & 1 & c & d \\ 0 & 0 & 0 & 0 \\ 0 & 0 & 0 & 0
  \end{pmatrix}
  \begin{pmatrix}
  \psi'_1(0) \\ \psi'_2(0) \\ \psi'_3(0) \\ \psi'_4(0)
  \end{pmatrix}
=
  \begin{pmatrix}
  0 & 0 & 0 & 0 \\ 0 & 0 & 0 & 0 \\ -a & -c & 1 & 0 \\ -b & -d & 0 & 1
  \end{pmatrix}
  \begin{pmatrix}
  \psi_1(0) \\ \psi_2(0) \\ \psi_3(0) \\ \psi_4(0)
  \end{pmatrix}\,.
\end{equation}
\begin{figure}[h]
  \centering
  \includegraphics[width=4.5cm]{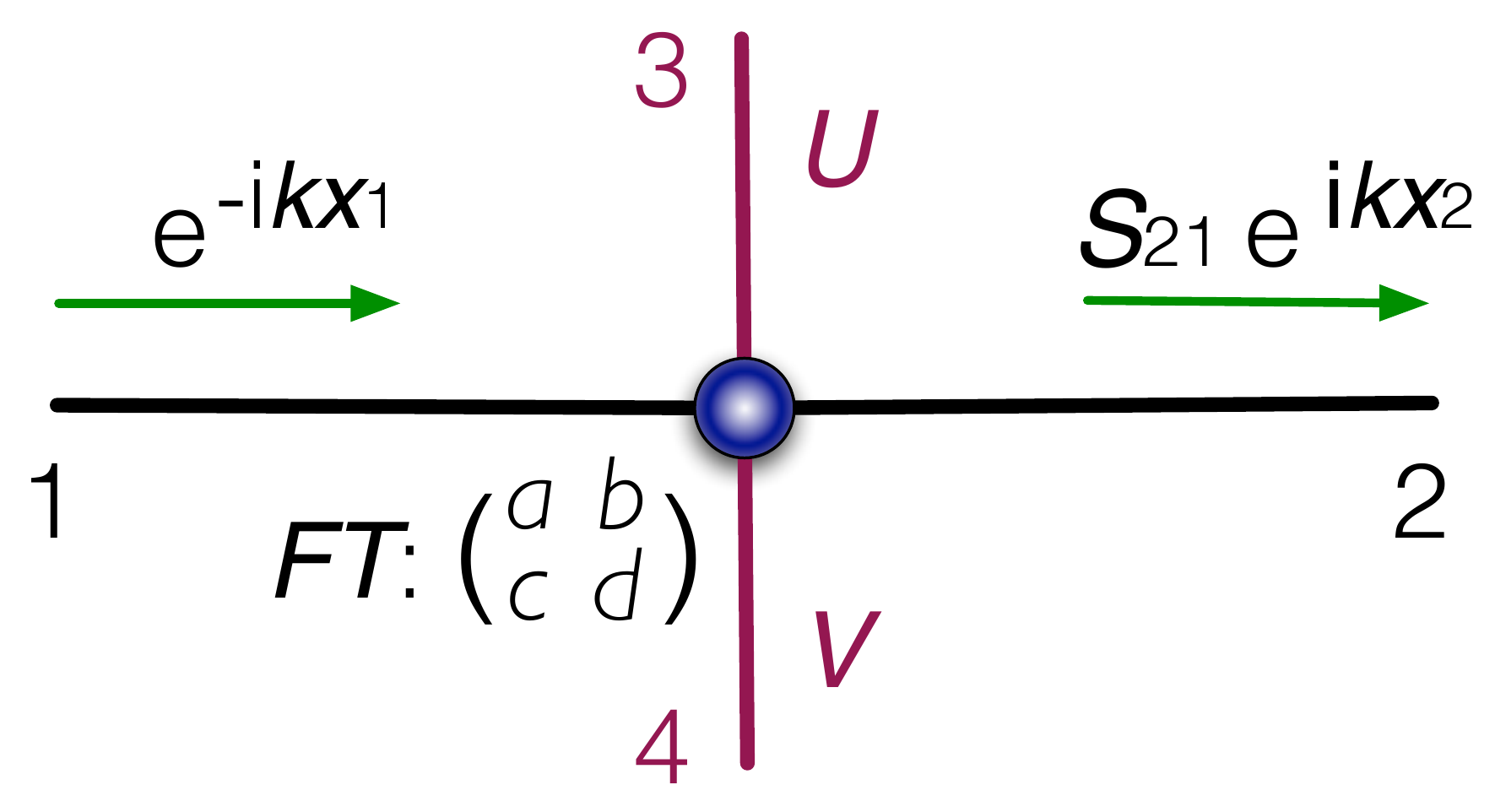}
  \caption{Scheme of a quantum spectral filter controllable by two external potentials $U$ and $V$ on the controlling lines 3 and 4, respectively.}
  \label{Fig.4}
\end{figure}

The calculation of the transmission amplitude $\S_{21}^{(U,V)}(E)\equiv[\S(E;0,0,U,V)]_{21}$ can be performed in the same way as in Section~\ref{Section: n=3} using formula~\eqref{SFTpot}; we find
\begin{equation}\label{S_21 4}
\S_{21}^{(U,V)}(E)=\frac{-2\left(ac\sqrt{1-\frac{U}{E}}+bd\sqrt{1-\frac{V}{E}}\right)}{1+(a^2+c^2)\sqrt{1-\frac{U}{E}}+(b^2+d^2)\sqrt{1-\frac{V}{E}}+(ad-bc)^2\sqrt{1-\frac{U}{E}}\sqrt{1-\frac{V}{E}}}\,.
\end{equation}
The transmission probability input~$\to$~output is given as $\P^{(U,V)}(E)=|\S_{21}^{(U,V)}(E)|^2$. Since 
the filter we seek
shall have two passbands at $E=U$ and $E=V$, we require zero transmission probability for $E\to\infty$ and high transmission probabilities (preferably $\approx1$) at $E=U$ and $E=V$. The first requirement (i.e., $\lim_{E\to\infty}\S_{21}^{(U,V)}(E)=0$) leads to the condition
\begin{equation}\label{cond1}
ac+bd=0\,.
\end{equation}
Then we proceed to the requirement $\P^{(U,V)}(U)\approx1$, $\P^{(U,V)}(V)\approx1$. Let us assume without loss of generality $V<U$. Since it holds
\begin{align}
\P^{(U,V)}(U)&=\left(\frac{2bd\sqrt{1-\frac{V}{U}}}{1+(b^2+d^2)\sqrt{1-\frac{V}{U}}}\right)^2\,, \\
\P^{(U,V)}(V)&=\frac{4a^2c^2\left(\frac{U}{V}-1\right)}{1+(a^2+c^2)^2\left(\frac{U}{V}-1\right)}\,,
\end{align}
we arrive at the condition
\begin{equation}\label{cond1.5}
b^2=d^2\gg1\,,\qquad a^2=c^2\gg1\,.
\end{equation}
Conditions~\eqref{cond1} and~\eqref{cond1.5} together lead to eight possible expressions for $T$, namely
\begin{equation}\label{sol T a}
\pm\begin{pmatrix}a & a \\ a & -a\end{pmatrix},\;\pm\begin{pmatrix}a & a \\ -a & a\end{pmatrix},\;\pm\begin{pmatrix}a & -a \\ a & a\end{pmatrix},\;\pm\begin{pmatrix}-a & a \\ a & a\end{pmatrix}
\end{equation}
for $a>0$ such that $a^2\gg1$. We choose the first of them, i.e.,
\begin{equation}\label{T 4 a}
T=\begin{pmatrix}a & a \\ a & -a\end{pmatrix}\,,
\end{equation}
but it is easy to show that all the solutions~\eqref{sol T a} result in the same transmission probabilities, thus a different choice would not make any difference.
Now we use the formula~\eqref{SFTpot} to calculate the corresponding input~$\to$~output transmission amplitude $\S_{21}^{(U,V)}(E)$, the reflection amplitude $\S_{11}^{(U,V)}(E)$, and the remaining transmission amplitudes $\S_{31}^{(U,V)}(E)$, $\S_{41}^{(U,V)}(E)$:
\begin{align}
S_{21}^{(U,V)}(E)&=\frac{-2a^2\left(\sqrt{1-\frac{U}{E}}-\sqrt{1-\frac{V}{E}}\right)}{\left(1+2a^2\sqrt{1-\frac{U}{E}}\right)\left(1+2a^2\sqrt{1-\frac{V}{E}}\right)}\,, \label{S_ 4,2} \\
S_{11}^{(U,V)}(E)&=\frac{1-4a^4\sqrt{1-\frac{U}{E}}\sqrt{1-\frac{V}{E}}}{\left(1+2a^2\sqrt{1-\frac{U}{E}}\right)\left(1+2a^2\sqrt{1-\frac{V}{E}}\right)}\,,\\
S_{31}^{(U,V)}(E)&=\frac{2a\sqrt[4]{1-\frac{U}{E}}}{1+2a^2\sqrt{1-\frac{U}{E}}}\,\Theta(E-U)\,,\\
S_{41}^{(U,V)}(E)&=\frac{2a\sqrt[4]{1-\frac{V}{E}}}{1+2a^2\sqrt{1-\frac{V}{E}}}\,\Theta(E-V)\,. \label{S_ 4,4}
\end{align}
The input~$\to$~output transmission probability is equal to $\P^{(U,V)}(E)=|\S_{21}^{(U,V)}(E)|^2$, thus
\begin{equation}\label{P(U,V) 4}
\P^{(U,V)}(E)=\left\{\begin{array}{ll}
\frac{4a^4\left(\sqrt{\frac{U}{E}-1}-\sqrt{\frac{V}{E}-1}\right)^2}{\left(1+4a^4\left(\frac{U}{E}-1\right)\right)\left(1+4a^4\left(\frac{V}{E}-1\right)\right)} & \text{ for } E<V, \\ [1em]
\frac{4a^4}{\left(1+4a^4\left(\frac{U}{E}-1\right)\right)\left(1+2a^2\sqrt{1-\frac{V}{E}}\right)^2}\,\frac{U-V}{E} & \text{ for } V<E<U, \\ [1em]
\left(\frac{2a^2\left(\sqrt{1-\frac{U}{E}}-\sqrt{1-\frac{V}{E}}\right)}{\left(1+2a^2\sqrt{1-\frac{U}{E}}\right)\left(1+2a^2\sqrt{1-\frac{V}{E}}\right)}\right)^2 & \text{ for } E>U.
\end{array}\right.
\end{equation}
If $U,V$ satisfy $2a^2\sqrt{1-V/U}\gg1$ and $4a^4(U/V-1)\gg1$, then
\begin{align}
&\lim_{E\to V}\P^{(U,V)}(E)=\frac{4a^4\left(\frac{U}{V}-1\right)}{1+4a^4\left(\frac{U}{V}-1\right)}\approx1\,, \\
&\lim_{E\to U}\P^{(U,V)}(E)=\frac{4a^4\left(1-\frac{V}{U}\right)}{\left(1+2a^2\sqrt{1-\frac{V}{U}}\right)^2}\approx1\,, \\
&\lim_{E\to0}\P^{(U,V)}(E)=0\,, \qquad \lim_{E\to\infty}\P^{(U,V)}(E)=0\,, \\
&\P^{(U,V)}(E)\approx0 \text{ for all $E$ except for certain small neighborhoods of $U$ and $V$}.
\end{align}
Therefore, the function $\P^{(U,V)}(E)$ has two sharp peaks, located at the energies $E=U$ and $E=V$,
see Figure~\ref{Fig. P a} (left).

The situation is different when one of the potentials vanishes. If $V=0$ and $U>0$, then
\begin{align}
&\lim_{E\to 0}\P^{(U,V)}(E)=\frac{1}{(1+2a^2)^2}\approx0\,, \\
&\lim_{E\to\infty}\P^{(U,V)}(E)=0\,, \\
&\lim_{E\to U}\P^{(U,V)}(E)=\frac{4a^4}{\left(1+2a^2\right)^2}\approx1\,, \\
&\P^{(U,V)}(E)\approx0 \text{ for all $E$ except for a certain small neighborhood of $U$}.
\end{align}
Therefore, turning one control potential down to zero (e.g., $V=0$) \emph{does not} result in a peak at $E=0$, but causes the peak corresponding to $V$ to vanish, whereas the peak located at $E=U$ remains essentially unaffected; see Figure~\ref{Fig. P a} for an illustration.

To sum up, the device constructed according to the scheme in Figure~\ref{Fig.4} for $T$ given by~\eqref{T 4 a} with $2a^2\gg1$ works as a band-pass spectral filter with $2$ passbands. The positions of the passbands are directly controlled by the external potentials on the controlling lines, and moreover, one band can be suppressed by setting one of the potentials to zero.

\begin{figure}[h]
  \begin{tabular}{cc}
  \includegraphics[width=6.5cm]{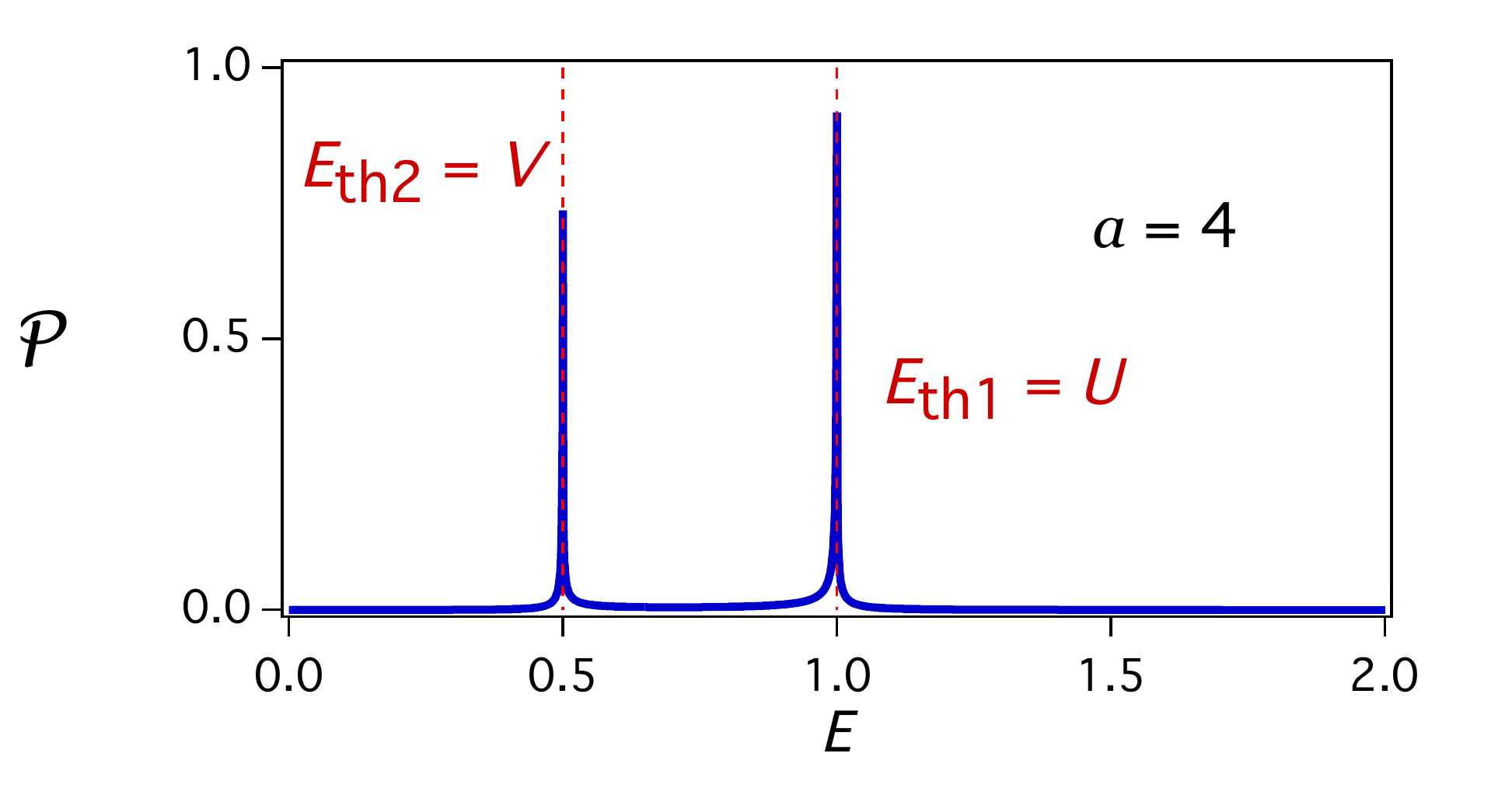} & \includegraphics[width=6.5cm]{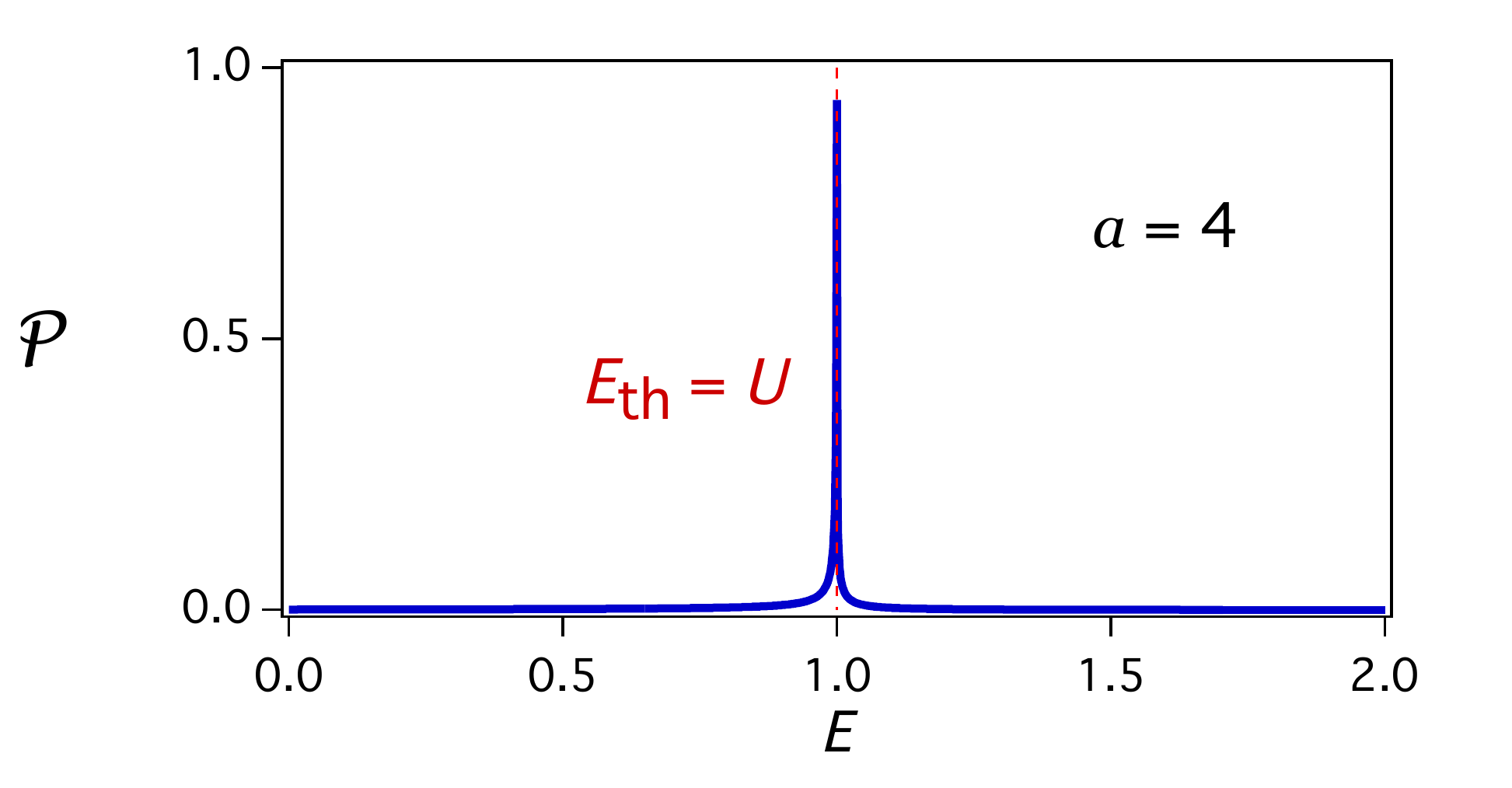} \\
  \includegraphics[width=6.5cm]{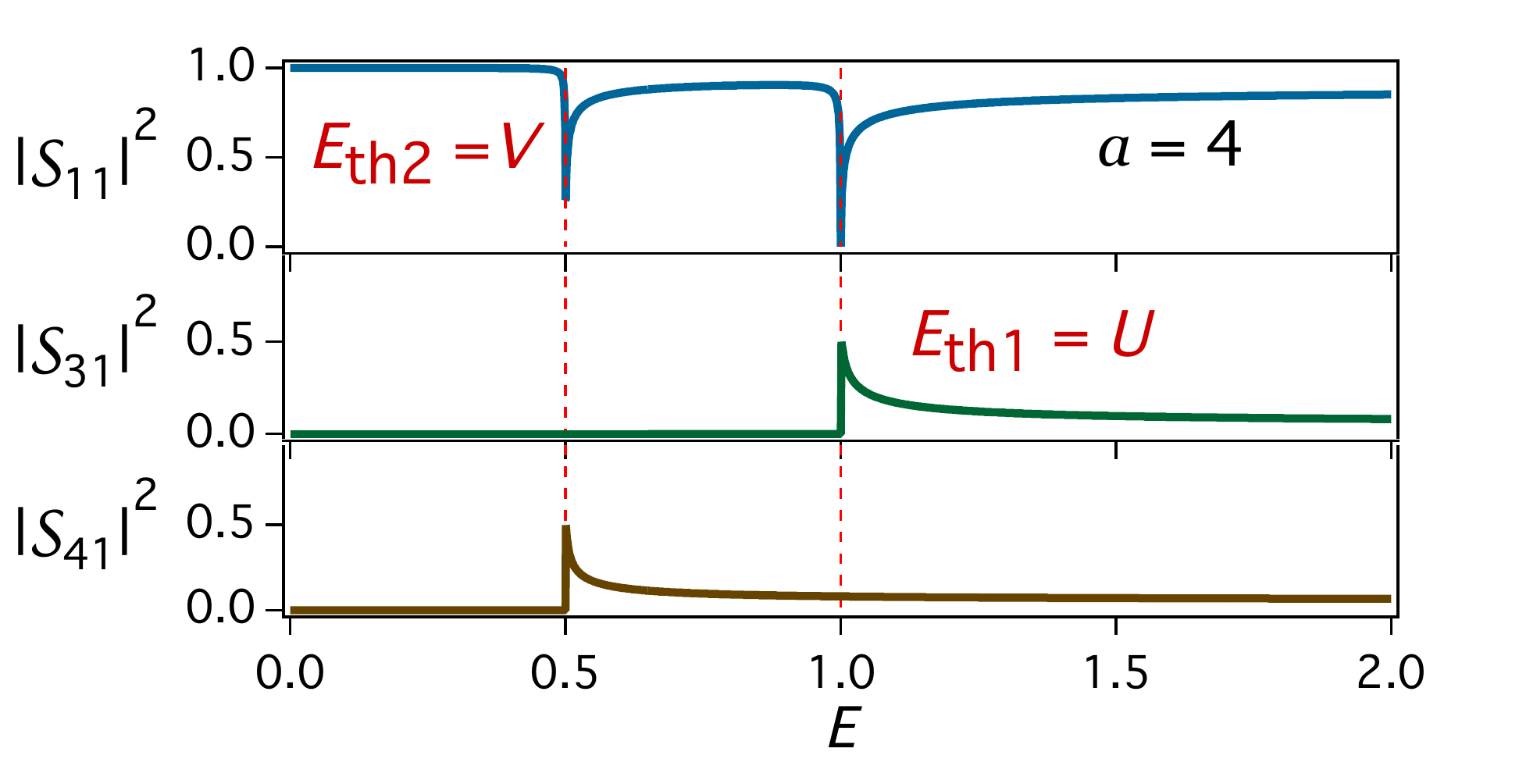} & \includegraphics[width=6.5cm]{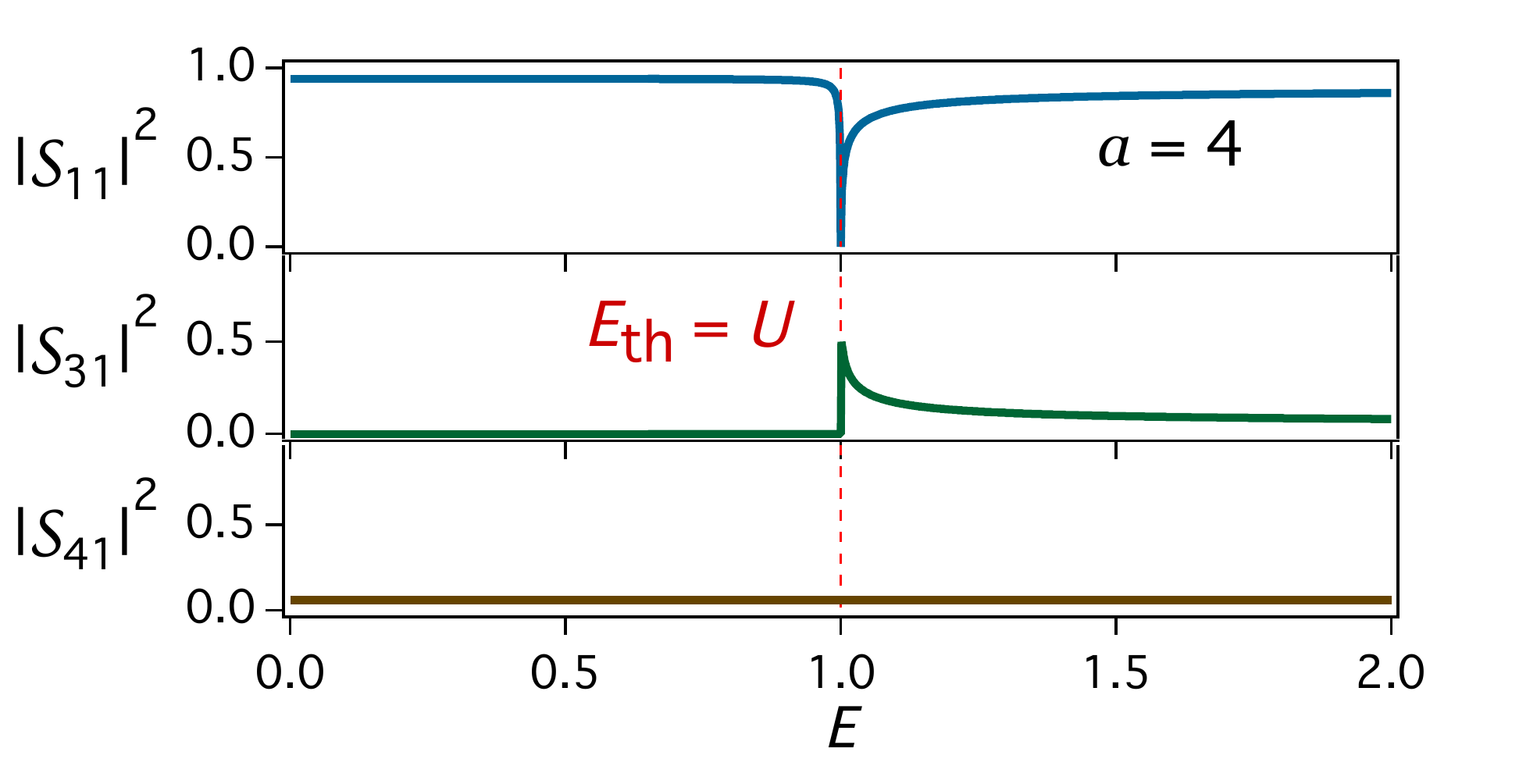}
  \end{tabular}
  \caption{Characteristics of the spectral filter obtained from the graph in Fig.~\ref{Fig.4} for $a=b=c=-d=4$. The left figure, plotted for $U=1$, $V=0.5$, illustrates the standard dual-band regime. The right figure, plotted for $U=1$, $V=0$, shows the effect of setting one of the potentials to zero; the filter is switched to the single-band regime. The top part of every figure displays the transmission probability $\P^{(U,0)}(E)$ as a function of $E$, the lower part shows the reflection probability $|S_{11}^{(U,V)}(E)|^2$ and the probabilities of transmission to the controlling lines $|S_{31}^{(U,V)}(E)|^2$ and $|S_{41}^{(U,V)}(E)|^2$.}
  \label{Fig. P a}
\end{figure}

It can be shown that the bandwidth of the passband at the energy $U>0$ satisfies $W_U\approx(1-1/\sqrt{2})U/a^4$. Similarly, if $V>0$, then $W_V=(1-1/\sqrt{2})V/a^4$.

The peaks of $\P^{(U,V)}(E)$ at $E_{\mathrm{th}1}=U$ and $E_{\mathrm{th}2}=V$ are related to the poles of the scattering matrix in the unphysical Riemann plane at $E^{\mathrm{pol}}_{1}=\frac{4a^4}{4a^4-1}U$ and $E^{\mathrm{pol}}_{2}=\frac{4a^4}{4a^4-1}V$, respectively.

\section{Fully tunable band-pass spectral filter}\label{Section: band-pass}

The filter designed in Section~\ref{Section: n=3} has one controllable parameter, namely the passband position which is adjustable by an external potential put on the controlling line. The aim of this section is to construct a filtering device that allows to control both the passband position and the bandwidth.

Similarly 
to Section~\ref{Section: n=4}, the goal will be achieved using the graph depicted in Figure~\ref{Fig.4} with the F\"ul\"op--Tsutsui vertex coupling given by boundary conditions~\eqref{b.c. 4}. The control potentials $U$ and $V$ will adjust the upper and the lower cutoff energies.

The transmission amplitude input~$\to$~output follows from the results of the previous section, and is given by~\eqref{S_21 4}.
Since the sought filter shall be of the band-pass type, we require zero transmission probability for $E\to\infty$, i.e., $\lim_{E\to\infty}\S_{21}^{(U,V)}(E)=0$. Hence we obtain the condition~\eqref{cond1}, i.e., $ac+bd=0$.

Now let us consider special situation 
where the lower cutoff energy is zero; in this regime the device shall work as a low-pass spectral filter. If we set $V=0$ in~\eqref{S_21 4} and use~\eqref{cond1} to simplify the numerator, we obtain
\begin{equation}\label{S_21(0) 4}
\S_{21}^{(U,0)}(E)=-\frac{2bd\left(1-\sqrt{1-\frac{U}{E}}\right)}{1+b^2+d^2+(a^2+c^2+(ad-bc)^2)\sqrt{1-\frac{U}{E}}}\,.
\end{equation}
Note that for $E\in[V,U]=[0,U]$, the expression $\sqrt{1-U/E}$ is imaginary, and consequently, the transmission probability for $E<U$ (in the intended passband) equals
\begin{equation}\label{P 4}
\P^{(U,0)}(E)=|\S_{21}^{(U,0)}(E)|^2=\frac{4b^2d^2\left(1+\left(1-\frac{U}{E}\right)\right)}{(1+b^2+d^2)^2+(a^2+c^2+(ad-bc)^2)^2\left(1-\frac{U}{E}\right)}\,.
\end{equation}
We observe that there is a special choice of $a,b,c,d$ that leads to a constant transmission probability in the whole interval $(0,U)$, i.e., to a flat passband. Indeed, if the parameters $a,b,c,d$ satisfy condition~\eqref{cond1} and at the same time
\begin{equation}\label{cond2}
1+b^2+d^2=a^2+c^2+(ad-bc)^2\,,
\end{equation}
then
\begin{equation}\label{P(0) 4}
\P^{(U,0)}(E)=\frac{4b^2d^2}{(1+b^2+d^2)^2}=\left(\frac{2bd}{1+b^2+d^2}\right)^2 \qquad\text{for all } E<U\,.
\end{equation}
Therefore, we impose both conditions~\eqref{cond1} and \eqref{cond2}, and require that the transmission probability~\eqref{P(0) 4} is as high as possible (our aim is to minimize the attenuation for $E$ inside the passband). In other words, we are to solve the optimization problem
\begin{equation}\label{optimization}
\text{maximize}\quad F(a,b,c,d)=\left(\frac{2bd}{1+b^2+d^2}\right)^2 \qquad\text{subject to eq. \eqref{cond1}\,\&\,\eqref{cond2}}\,.
\end{equation}
At first, let us simplify the problem by eliminating $a$ and $c$. We express $c$ from~\eqref{cond1}, substitute it to~\eqref{cond2},
\begin{equation}
1+b^2+d^2=a^2+(bd/a)^2+(ad+b^2d/a)^2\,,
\end{equation}
then multiply this equation by $a^2$ and rewrite it as a polynomial equation in $a$:
\begin{equation}
(1+d^2)a^4+(2b^2d^2-1-b^2-d^2)a^2+b^2d^2(1+b^2)=0\,.
\end{equation}
This is a biquadratic equation 
for $a$ with parameters $b,d$; it can be shown that it has a real solution $a$ if and only if $b,d$ satisfy
\begin{equation}
1+b^2+d^2-8b^2d^2\geq0\,.
\end{equation}
Therefore, the optimization problem~\eqref{optimization} is equivalent to the problem
\begin{equation}\label{optimization2}
\text{maximize}\quad F(b,d)=\left(\frac{2bd}{1+b^2+d^2}\right)^2 \qquad\text{subject to}\quad 1+b^2+d^2-8b^2d^2\geq0\,.
\end{equation}
One can easily find its solution:
\begin{equation}\label{bd}
|b|=|d|=1/\sqrt{2}\,.
\end{equation}
We substitute $b,d$ from~\eqref{bd} into conditions~\eqref{cond1} and \eqref{cond2}, and in this way find $a,c$. Altogether, we obtain eight solutions for $T$, namely
\begin{equation}\label{sol T}
\pm\frac{1}{\sqrt{2}}\begin{pmatrix}1 & 1 \\ 1 & -1\end{pmatrix},\;\pm\frac{1}{\sqrt{2}}\begin{pmatrix}1 & 1 \\ -1 & 1\end{pmatrix},\;\pm\frac{1}{\sqrt{2}}\begin{pmatrix}1 & -1 \\ 1 & 1\end{pmatrix},\;\pm\frac{1}{\sqrt{2}}\begin{pmatrix}-1 & 1 \\ 1 & 1\end{pmatrix}
\,.
\end{equation}
We choose the first of them,
\begin{equation}\label{T 4}
T=\frac{1}{\sqrt{2}}\begin{pmatrix}1 & 1 \\ 1 & -1\end{pmatrix}\,,
\end{equation}
but the choice actually makes no difference, as all the solutions~\eqref{sol T} lead to the same transmission probabilities.

We observe that the solution~\eqref{T 4} is nothing but a special case of the matrix $T$ found in Section~\ref{Section: n=4} (see eq.~\eqref{T 4 a}) for $a=1/\sqrt{2}$.
Consequently, the corresponding input~$\to$~output transmission amplitude $\S_{21}^{(U,V)}(E)$, the reflection amplitude $\S_{11}^{(U,V)}(E)$, and the remaining transmission amplitudes $\S_{31}^{(U,V)}(E)$, $\S_{41}^{(U,V)}(E)$ immediately follow from equations~\eqref{S_ 4,2}--\eqref{S_ 4,4}, i.e.,
\begin{equation}
S_{21}^{(U,V)}(E)=-\frac{\sqrt{1-\frac{U}{E}}-\sqrt{1-\frac{V}{E}}}{\left(1+\sqrt{1-\frac{U}{E}}\right)\left(1+\sqrt{1-\frac{V}{E}}\right)}\,,
\end{equation}
etc.
Let us
assume without loss of generality that $V<U$. The input~$\to$~output transmission probability, given by $\P^{(U,V)}(E)=|\S_{21}^{(U,V)}(E)|^2$, follows from equation~\eqref{P(U,V) 4} for $a=1/\sqrt{2}$:
\begin{equation}
\P^{(U,V)}(E)=\left\{\begin{array}{ll}
\left(\sqrt{\frac{U}{E}-1}-\sqrt{\frac{V}{E}-1}\right)^2\cdot\frac{E^2}{UV} & \text{ for } E<V, \\ [1em]
\frac{1}{\left(1+\sqrt{1-\frac{V}{E}}\right)^2}\left(1-\frac{V}{U}\right) & \text{ for } V<E<U, \\ [1em]
\left(\frac{\sqrt{1-\frac{U}{E}}-\sqrt{1-\frac{V}{E}}}{\left(1+\sqrt{1-\frac{U}{E}}\right)\left(1+\sqrt{1-\frac{V}{E}}\right)}\right)^2 & \text{ for } E>U.
\end{array}\right.
\end{equation}
The function $\P^{(U,V)}(E)$ has the following properties:
\begin{itemize}
\item If $V=0$, then
\begin{align}
&\P^{(U,0)}(E)=1/4 \qquad \text{for } E\in(0,U)\,, \\
&\P^{(U,0)}(E) \text{ as a function of $E$ quickly falls off to zero at $E>U$}\,, \\
&\lim_{E\to\infty}\P^{(U,0)}(E)=0\,.
\end{align}
\item If $0<V<U$, then
\begin{align}
&\P^{(U,V)}(V)=1-\frac{V}{U}\,,\quad \P^{(U,V)}(U)=\frac{1}{\left(1+\sqrt{1-\frac{V}{U}}\right)^2}\left(1-\frac{V}{U}\right)\,, \\
&\P^{(U,V)}(E) \text{ decreases in $(V,U)$}\,, \\
&\P^{(U,V)}(E) \text{ quickly decreases for $E>U$ and grows for $E<V$}\,, \\
&\lim_{E\to0}\P^{(U,V)}(E)=0\,, \qquad \lim_{E\to\infty}\P^{(U,V)}(E)=0\,.
\end{align}
\end{itemize}
The behaviour of $\P^{(U,V)}(E)$ is illustrated in Figure~\ref{Fig. P 4} in both situations $V>0$, $V=0$.
\begin{figure}[h]
  \begin{tabular}{cc}
  \includegraphics[width=6.5cm]{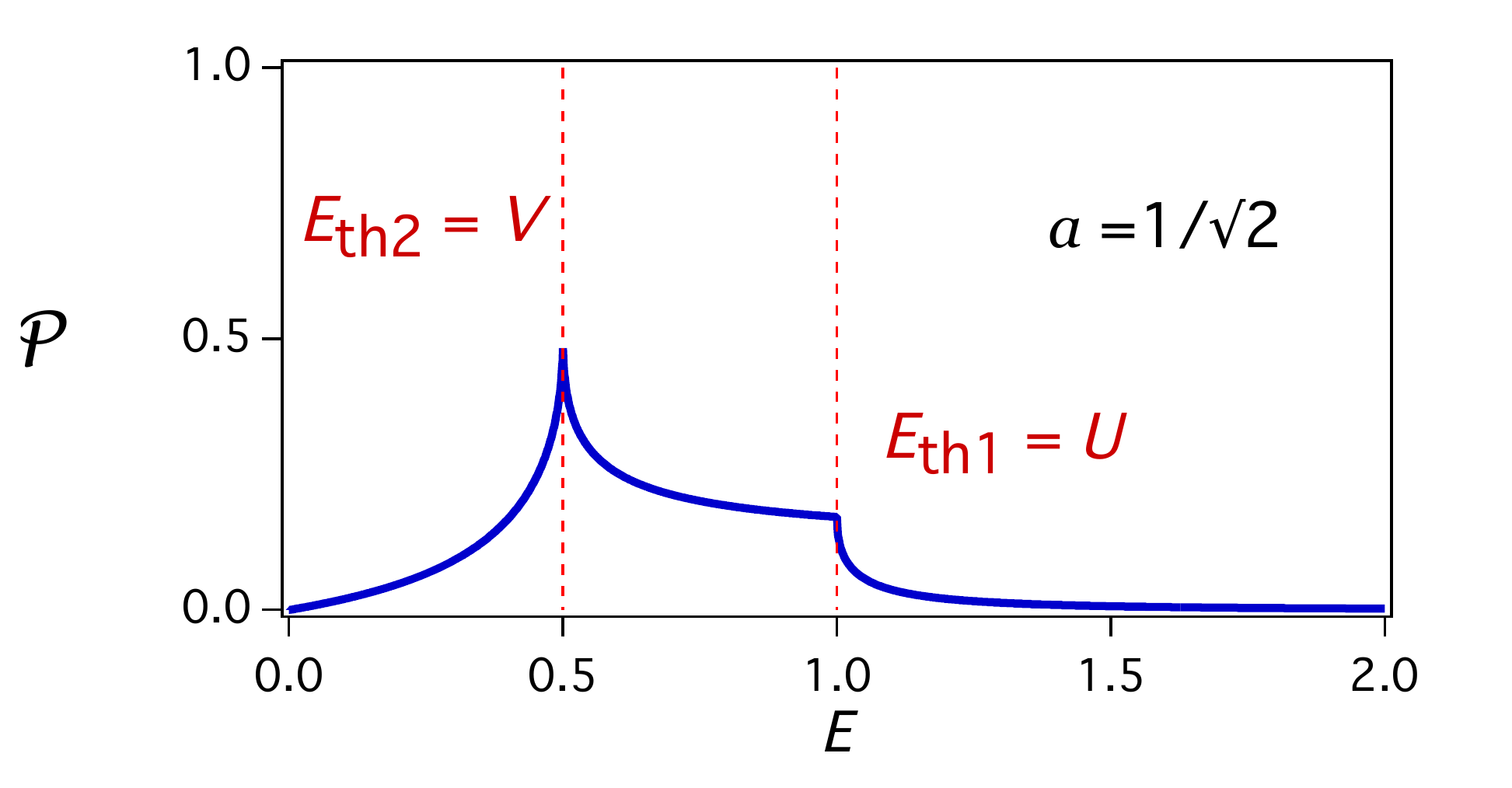} & \includegraphics[width=6.5cm]{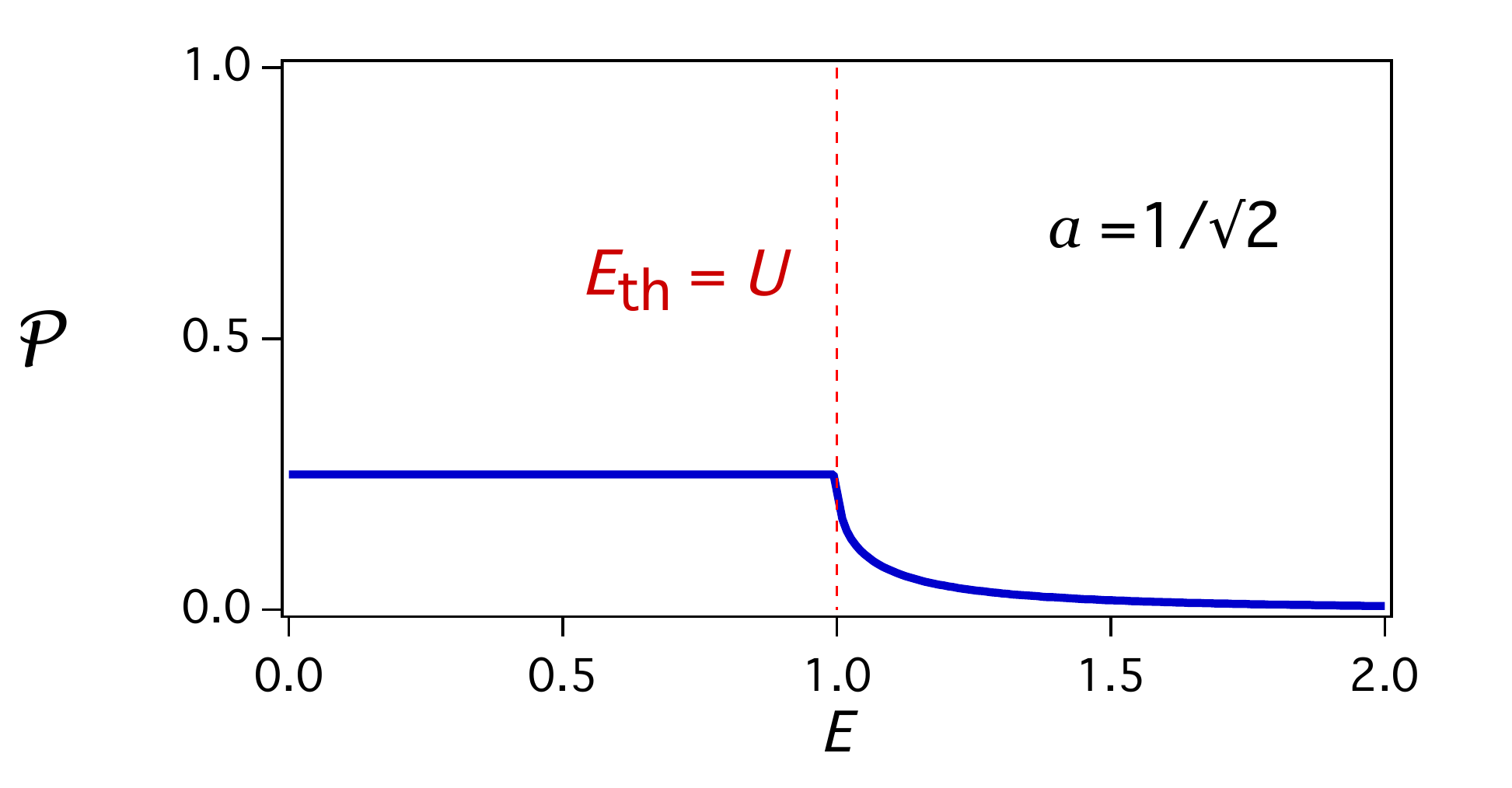} \\
  \includegraphics[width=6.5cm]{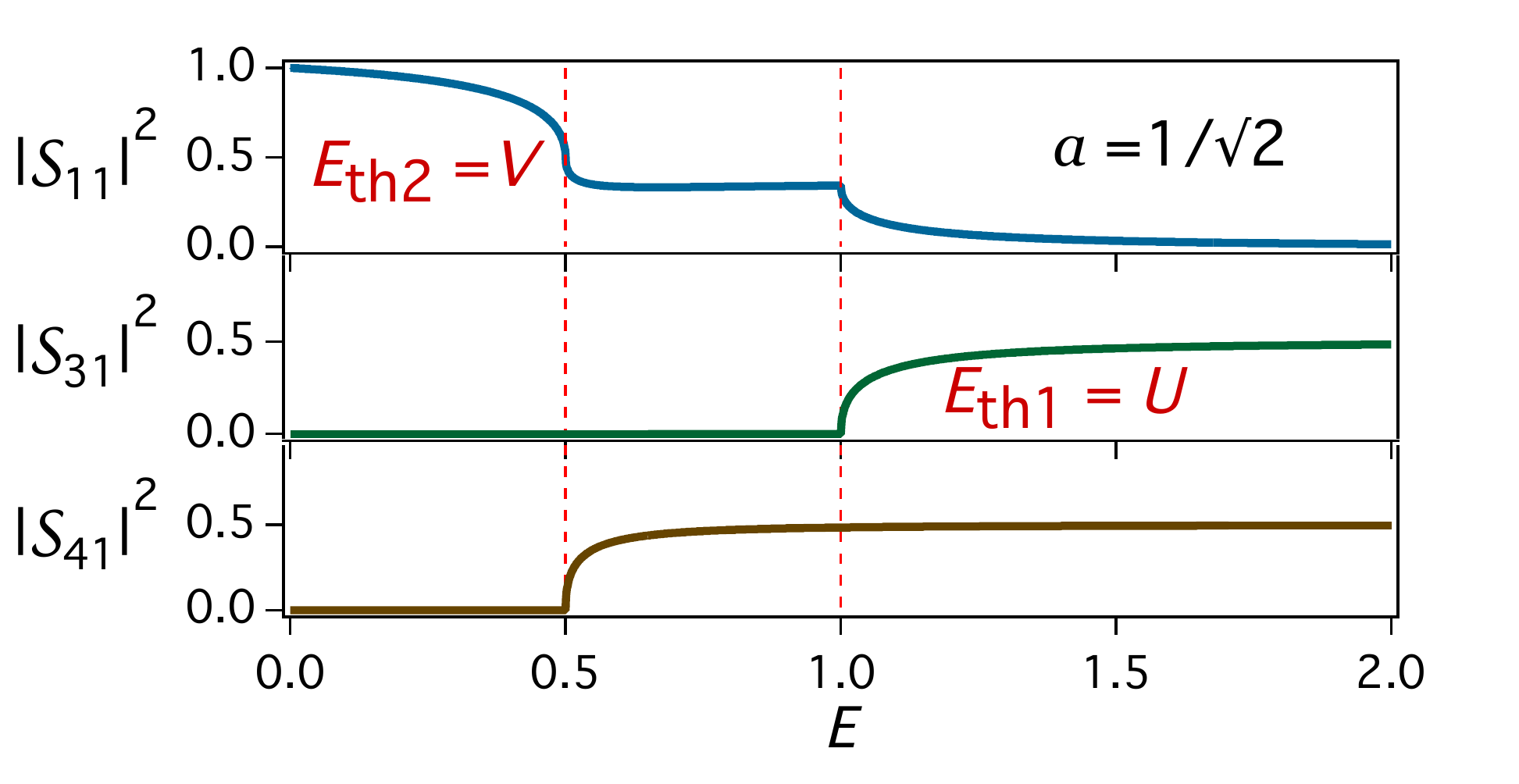} & \includegraphics[width=6.5cm]{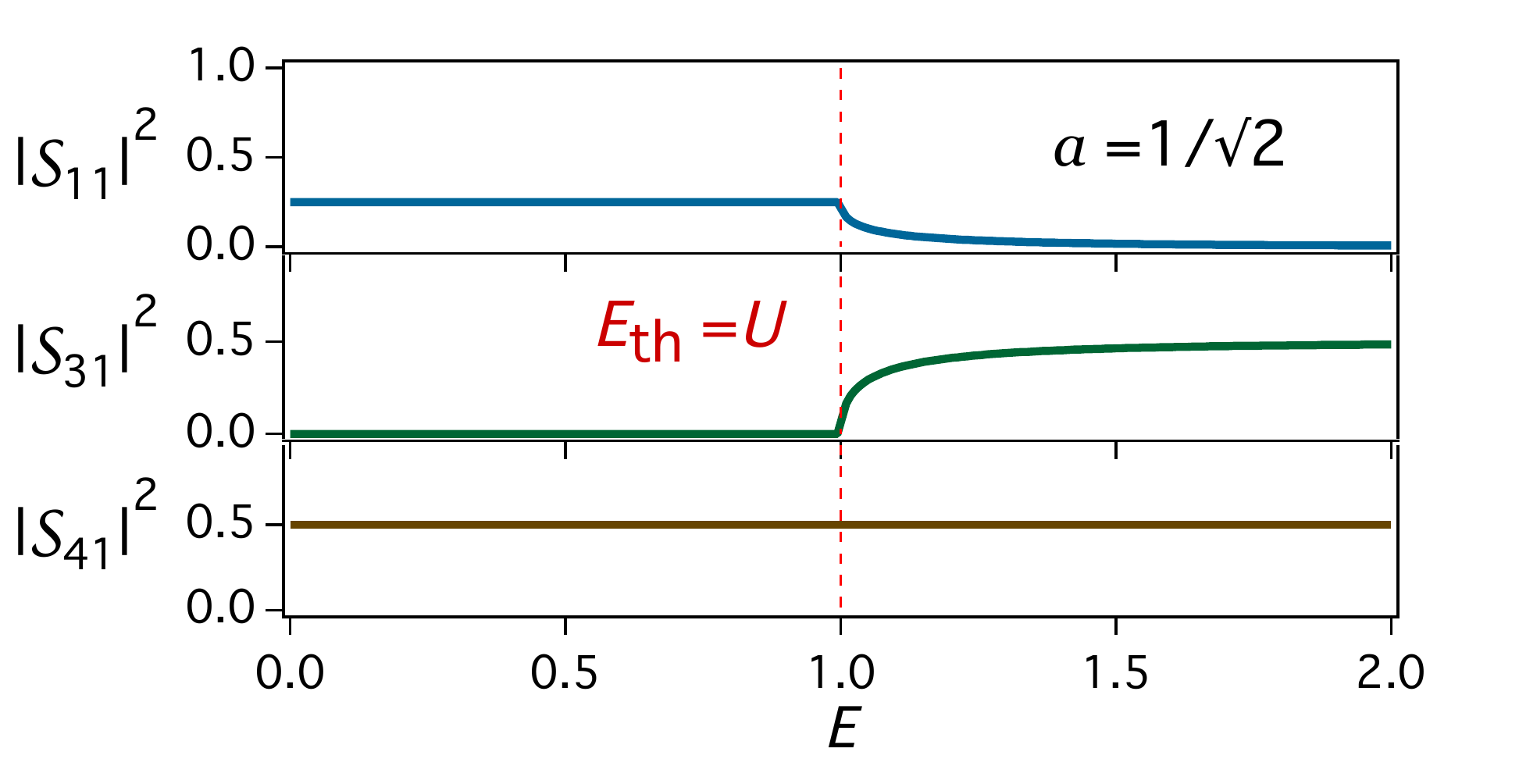}
  \end{tabular}
  \caption{Characteristics of the band-pass spectral filter obtained from the graph in Fig.~\ref{Fig.4} for $a=b=c=-d=1/\sqrt{2}$. The left figure illustrates the situation $U=1$ and $V=0.5$, the right figure the situation $U=1$ and $V=0$ (the flat-passband low-pass regime). The top part of every figure displays the transmission probability $\P^{(U,V)}(E)$ as a function of $E$, the lower part shows the reflection probability $|S_{11}^{(U,V)}(E)|^2$ and the probabilities of transmission to the controlling lines $|S_{31}^{(U,V)}(E)|^2$ and $|S_{41}^{(U,V)}(E)|^2$.}
  \label{Fig. P 4}
\end{figure}

\begin{remark}
The term ``passband'' is used in this section in a weakened sense, because the transmission probability for $E\in[V,U]$ mostly does not exceed $1/2$. However, considering the characteristics of the filter (cf. Fig.~\ref{Fig. P 4}), it is apparently reasonable to regard the interval $\approx[V,U]$ as a passband, especially in the regime $V=0$. 
\end{remark}

\begin{remark}
If the extremality requirement~\eqref{P(0) 4} is left out, conditions~\eqref{cond1}\,\&\,\eqref{cond2}, inducing the flat passband in the low-pass regime $V=0$, have various other solutions, such as
\begin{align}
T &= \begin{pmatrix}  a & \sqrt{1-a^2} \\ \sqrt{1-a^2} & -a \end{pmatrix} \quad \text{for $a\in(0,1)$} \quad\Rightarrow\quad\left.\P^{(U,0)}(E)\right|_{E\in(0,U)}=a^2(1-a^2)\,, \\ 
T &= \frac{1}{\sqrt{2}}\begin{pmatrix}a & a \\ 1/a & -1/a\end{pmatrix} \quad \text{for $a>0$} \quad\Rightarrow\quad\left.\P^{(U,0)}(E)\right|_{E\in(0,U)}=\frac{4a^4}{(a^2+1)^4}\,.
\end{align}
\end{remark}

\subsubsection*{Potential-controlled quantum sluice-gate}

As we have seen, when there is no potential on the line 4 ($V=0$), the device behaves as a low-pass spectral filter with a flat passband that transmits (with the probability of $1/4$) quantum particles with energies $E \in [0,U]$ to the output, whereas particles with higher energies are diverted to the other lines, mainly to 3 and 4.
This property allows to use the device also as a quantum flux controller:
If many particles described by the energy distribution $\rho(E)$ are sent along the line 1, the flux $J$ to the line 2 is given by
\begin{equation}
J(U) = \int_E \rho(E)\P^{(U,0)}(E)\d E\,.
\end{equation}
Assuming the Fermi distribution with Fermi energy $E_F$ larger than our range of operation of $U$, we can set $\rho(E)=\rho=\mathrm{const}$. With the approximation $\P^{(U,0)}(E) \approx \frac{1}{4} \Theta(U-E)$, we obtain
$
J(U) \approx \frac{1}{4}\rho U
$,
which indicates a linear flux control.
Therefore,
the device in the operation mode $V=0$ can be used as a quantum sluice-gate, linearly adjustable by the potential $U$ applied to the controlling line 3. In this regime the line 4 serves as a \emph{drain}.

\section{Spectral filter with multiple passbands}\label{Section: n=2r}

In this section we generalize the idea of Section~\ref{Section: n=4}. For any $r>1$ we design a spectral filter which has $r$ passbands such that their positions are directly controllable by potentials on $r$ controlling lines.

The design is based on a star graph with $n=2r$ lines, as depicted in Figure~\ref{Fig.2r}. The individual lines have the following meanings:
\begin{itemize}
\item Line 1 is \emph{input}.
\item Line 2 is \emph{output}.
\item Lines $3,\ldots,r$ are \emph{auxiliary lines}; their possible use will be discussed later on.
\item Lines $r+1,\ldots,n$ are \emph{controlling lines}, subjected to adjustable external potentials $U_1,\ldots,U_r$.
\end{itemize}

The F\"ul\"op--Tsutsui vertex coupling in the graph center is given by the boundary conditions written in the $ST$-form~\eqref{FT}, and the matrix $T$ is chosen as $T=aG$ for an $a>0$ and a certain unitary matrix $G$ or the order $r$, i.e.,
\begin{equation}\label{b.c.2r}
\left(\begin{array}{cc}
I^{(r)} & aG \\
0 & 0
\end{array}\right)\Psi'(0)=
\left(\begin{array}{cc}
0 & 0 \\
-aG^* & I^{(r)}
\end{array}\right)\Psi(0)\,.
\end{equation}
\begin{figure}[h]
\begin{center}
\includegraphics[width=4.5cm]{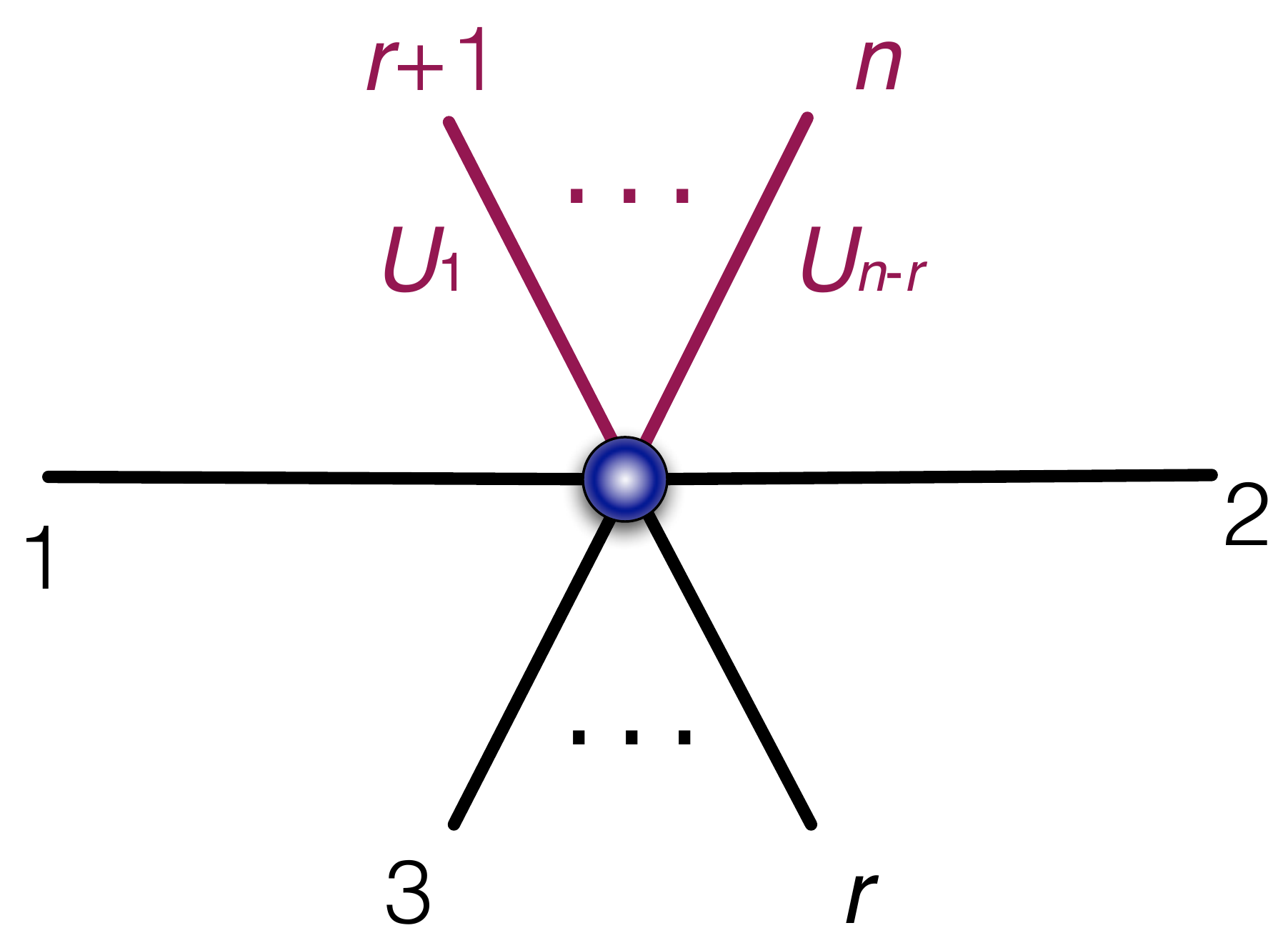}
\end{center}
\caption{Scheme of a quantum spectral filter based on a $n=2r$ star graph. The transmission from the intput 1 to the output 2 is controlled by $r$ external potentials $U_1,\ldots,U_r$ on the controlling lines $r+1,\ldots,n$, respectively.}
\label{Fig.2r}
\end{figure}

For any $j,\ell\leq r$, $j\neq\ell$, the transmission amplitude from the $\ell$-th line to the $j$-th line can be easily obtained from equation~\eqref{SFTpot}. If we denote, for the sake of brevity, $\S_{j\ell}^{(U)}(E)\equiv[\S(E;0,\ldots,0,U_1,\ldots,U_r)]_{j\ell}$, then
\begin{equation}\label{S_jl 2r}
\S_{j\ell}^{(U)}(E)=\frac{2g_{j1}\overline{g_{\ell 1}}}{1+a^2\sqrt{1-\frac{U_1}{E}}}+\frac{2g_{j2}\overline{g_{\ell 2}}}{1+a^2\sqrt{1-\frac{U_2}{E}}}+\cdots+\frac{2g_{jr}\overline{g_{\ell r}}}{1+a^2\sqrt{1-\frac{U_r}{E}}}\,.
\end{equation}
The transmission amplitude input~$\to$~output corresponds to the choice $\ell=1$, $j=2$ in~\eqref{S_jl 2r}. It holds
\begin{equation}
\lim_{E\to\infty}\S_{21}^{(U)}(E)=2g_{21}\overline{g_{11}}+2g_{22}\overline{g_{12}}+\cdots+2g_{2r}\overline{g_{1r}}=2[GG^*]_{21}=2[I^{(r)}]_{21}=0\,,
\end{equation}
and if the control potentials $U_1,\ldots,U_r$ are all nonzero and mutually different, then
\begin{align}
\lim_{E\to0}\S_{21}^{(U)}(E)&=0+0+\cdots+0=0\,, \\
\lim_{E\to U_j}\S_{21}^{(U)}(E)&=2g_{2j}\overline{g_{1j}}+\sum_{\substack{1\leq \ell\leq r\\ \ell\neq j}}\frac{2g_{2\ell}\overline{g_{1\ell}}}{1+a^2\sqrt{1-\frac{U_{\ell}}{U_j}}}\,.
\end{align}
If $a\gg1$ and, moreover, the ``isolated potentials'' condition
\begin{equation}\label{not close U}
a^2\sqrt{\left|1-\frac{U_{\ell}}{U_j}\right|}\gg1 \qquad \text{for all $j,\ell=1,\ldots,r$, $j\neq\ell$},
\end{equation}
is satisfied, we obtain
\begin{equation}\label{S(U_j)}
\lim_{E\to U_j}\S_{21}^{(U)}(E)=2g_{2j}\overline{g_{1j}} \quad\text{for all $j=1,\ldots,r$}\,,
\end{equation}
and at the same time
\begin{equation}
\S_{21}^{(U)}(E)\approx0 \quad\text{for all $E$ except for certain small neighborhoods of $U_1\ldots,U_r$}\,.
\end{equation}
To sum up, the absolute value of the function $\S_{21}^{(U)}(E)$ for $a\gg1$ has sharp peaks at $E=U_1,\ldots,U_n$. Since the input~$\to$~output transmission probability equals $|\S_{21}^{(U)}(E)|^2$, the studied device works as a spectral filter with multiple passbands positioned at the energies determined by the control potentials. The situation is illustrated in Figure~\ref{Fig. P 2r} in the case of the device constructed for
\begin{equation}\label{2r ex.}
r=4, \quad G=\frac{1}{2}\begin{pmatrix}1&1&1&1\\1&-1&1&-1\\1&1&-1&-1\\1&-1&-1&1\end{pmatrix}, \quad a=4\,.
\end{equation}
\begin{figure}[h]
  \begin{tabular}{cc}
  \includegraphics[width=6.5cm]{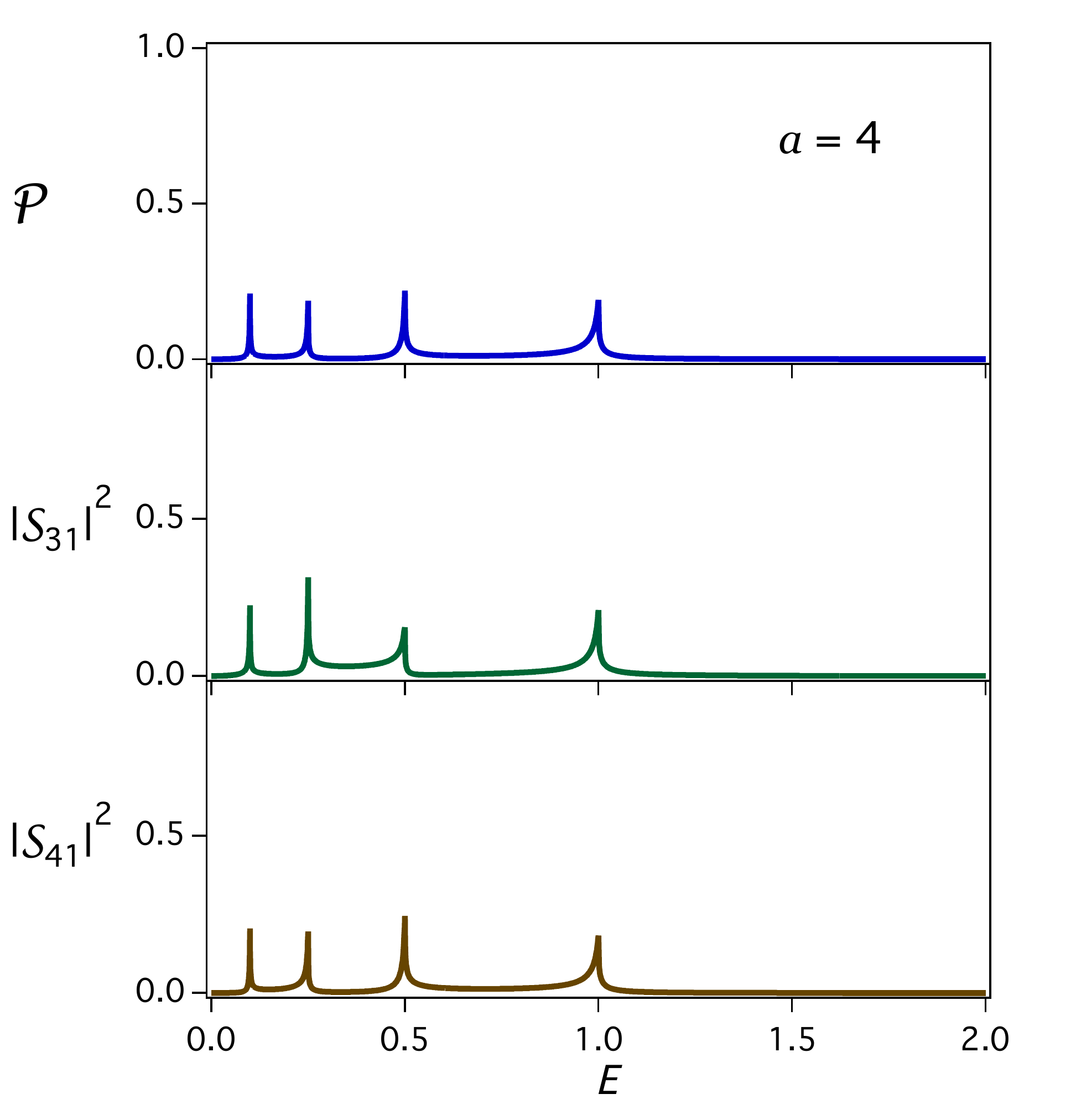} & \includegraphics[width=6.5cm]{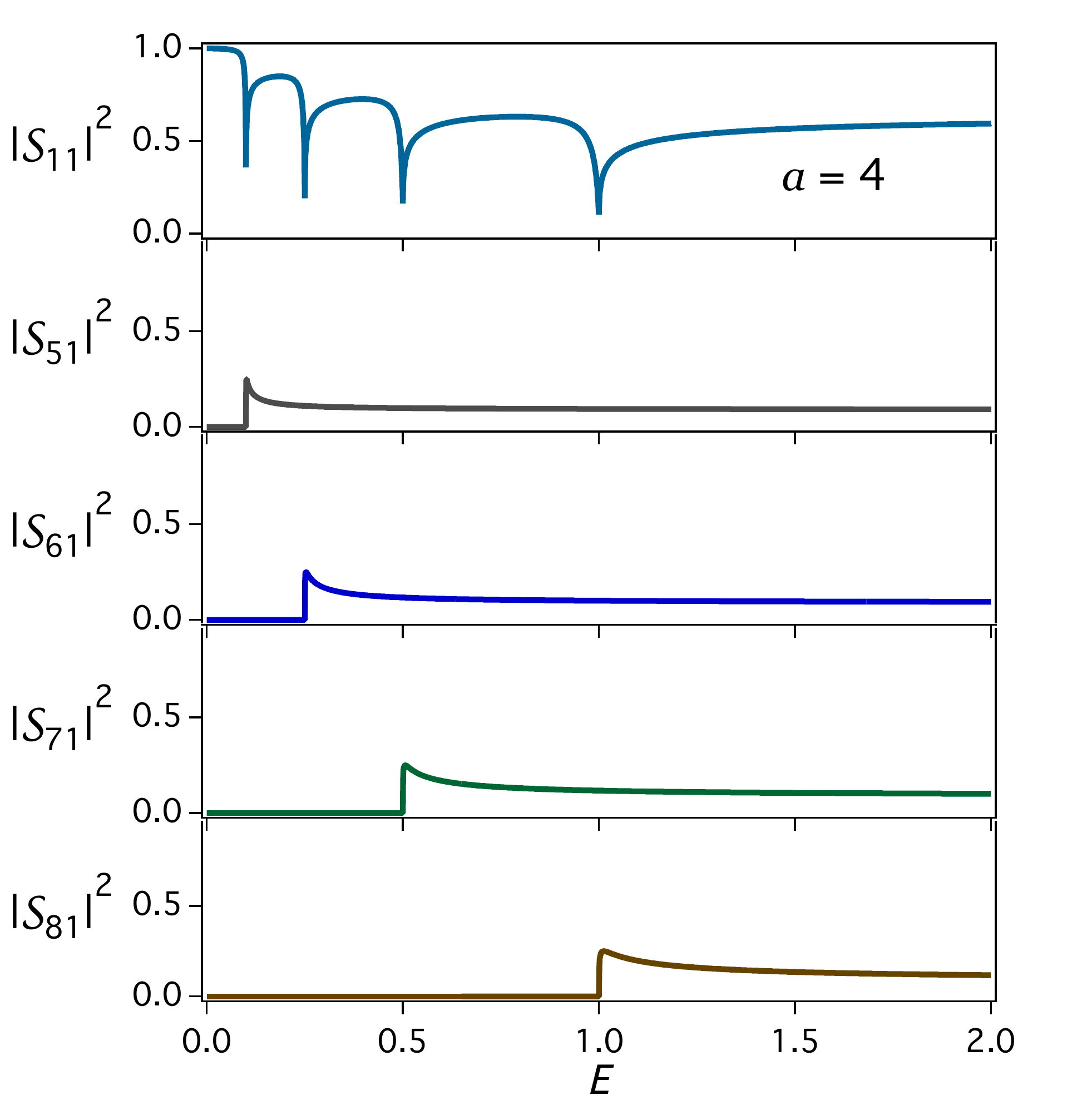}
  \end{tabular}
  \caption{Characteristics of the filter with multiple passbands, obtained from the graph in Fig.~\ref{Fig.2r} with the parameters given by eq.~\eqref{2r ex.}. The control potentials are set to $U_1=0.1$, $U_2=0.25$, $U_3=0.5$, $U_4=1$. The left figure shows the transmission probabilities $|\S_{j1}^{(\{U\})}(E)|^2$ between the lines $1 \to j$; due to the choice of $G$, all the peak heights are almost identical. The right figure shows the reflection probability $|S_{11}^{(U)}(E)|^2$ and the probabilities of transmission to the controlling lines $|S_{\ell 1}^{(U)}(E)|^2$ for $\ell=r+1,\ldots,n$.}
  \label{Fig. P 2r}
\end{figure}

The peaks of $\P^{(\{U\})}(E)$ at the energies $E=U_j$ ($j=1,\ldots,r$) are related to the poles of the scattering matrix in the unphysical Riemann plane at 
$
E^{\mathrm{pol}}_{j}=\frac{a^4}{a^4-1}U_j
$.

\subsubsection*{On the heights of the peaks and the choice of $G$}

Let $U_1,\ldots,U_r$ be all nonzero and satisfying the ``isolated potentials'' condition~\eqref{not close U}.
With regard to equation~\eqref{S(U_j)}, the heights of the probability peaks at $E=U_j$ equal
\begin{equation}\label{P(U_j)}
\P^{(\{U\})}(U_j)=|\S_{21}^{(U)}(U_j)|^2\approx4|g_{2j}|^2|g_{1j}|^2\,,
\end{equation}
thus, in general, they may be mutually different. However, if all the entries of $G$ have the same moduli (equal to $1/\sqrt{r}$), then all the peaks are (approximately) of the uniform height
\begin{equation}\label{peak heights}
\P^{(\{U\})}(U_j)\approx\frac{4}{r^2}\,.
\end{equation}
A unitary matrix $G$ with this property is nothing but $G=\frac{1}{\sqrt{r}}H^{(r)}$, where $H^{(r)}$ is an \emph{Hadamard matrix} of the order $r$. Regarding the existence of Hadamard matrices of order $r$, it holds:
\begin{itemize}
\item a complex Hadamard matrix exists for every $r\in\N$, for example the matrix with the entries
\begin{equation}
[H^{(r)}]_{jk}=\e^{\frac{2\pi\i}{r}(j-1)(k-1)}\,,
\end{equation}
\item a real Hadamard matrix (having elements $\pm1$) of the order $r$ is conjectured to exist if and only if $r=1$, $r=2$ or $r$ is a multiple of $4$.
\end{itemize}
Consequently, a filter with $r$ passbands of the equal height $4/r^2$ can be constructed for every $r\geq2$, but for most numbers $r$ the corresponding matrices $G$ need to be complex.

\medskip

Since the studied device is a generalization of the filter designed in Section~\ref{Section: n=4}, it has similar properties. A more detailed analysis follows.

\paragraph*{Reduction of the number of passbands}

The studied filter allows not only to control the passband positions, but also to reduce their number to any $r'<r$, which is achieved by setting $r-r'$ control potentials to zero. Indeed, if $r-r'$ controlling lines carry zero potentials (we may assume without loss of generality $U_{r'+1}=U_{r'+2}=\cdots=U_{r}=0$), then
\begin{equation}
\S_{21}^{(U)}(E)=\sum_{\ell=1}^{r'}\frac{2g_{2\ell}\overline{g_{1\ell}}}{1+a^2\sqrt{1-\frac{U_\ell}{E}}}+\frac{1}{1+a^2}\sum_{\ell=r'+1}^{r}2g_{2\ell}\overline{g_{1r}}\approx\sum_{\ell=1}^{r'}\frac{2g_{2\ell}\overline{g_{1\ell}}}{1+a^2\sqrt{1-\frac{U_\ell}{E}}}
\end{equation}
with regard to the assumption $a\gg1$. Hence, the effect of setting $U_{r'+1}=U_{r'+2}=\cdots=U_{r}=0$ essentially consists in suppressing $r-r'$ peaks. More precisely speaking, $\P^{(\{U\})}(0)$ is very slightly increased from $0$ to
\begin{equation}
\lim_{E\to0}\P^{(\{U\})}(E)\approx\frac{4}{(1+a^2)^2}\left|\sum_{\ell=1}^r g_{2\ell}\overline{g_{1\ell}}\right|^2\,,
\end{equation}
and the heights of the remaining $r'$ peaks at $U_j$, $j=1,\ldots,r'$, are very slightly changed from $\approx4|g_{21}|^2|\overline{g_{11}}|^2$ to
\begin{equation}
\lim_{E\to U_j}\P^{(\{U\})}(E)\approx4\left|g_{2j}\overline{g_{1j}}+\frac{1}{1+a^2}\sum_{\ell=1}^r g_{2\ell}\overline{g_{1\ell}}\right|^2\,,
\end{equation}
but the effect is not significant due to $a\gg1$.
Note, however, that reducing the number of passbands in this way (when the device is in operation) is not equivalent to using a device constructed for a smaller $r$, because the peak heights generally depend on $r$ (cf. equation~\eqref{peak heights}, related to the Hadamard case, and the discussion above it).

If all the control potentials are zero, then $\P^{(\{U\})}(E)\equiv0$, which follows from equation~\eqref{S_jl 2r} together with the unitarity of $G$.

\paragraph*{Effect of identical potentials}

When several control potentials are set to the same value (or if their values are very close), it results in 
the merging of 
the corresponding passbands into one single passband, and, therefore, to the reduction of the peaks of $\P^{(\{U\})}(E)$. In contrast to the reduction achieved via setting several potentials to zero (see above), this manner influences the height of the merged peak.
If, for instance, $U_1=\cdots=U_j$ ($j<r$), then
\begin{equation}
\S_{21}^{(U)}(E)=\frac{\sum_{\ell=1}^{j}2g_{2\ell}\overline{g_{1\ell}}}{1+a^2\sqrt{1-\frac{U}{E}}}+\sum_{\ell=j+1}^{r}\frac{2g_{2\ell}\overline{g_{1\ell}}}{1+a^2\sqrt{1-\frac{U_\ell}{E}}}\,,
\end{equation}
therefore, the height of the peak at $E=U_1$ is changed from $\approx4|g_{21}\overline{g_{11}}|^2$ (see~\eqref{P(U_j)}) to
\begin{equation}
\P^{(\{U\})}(U_1)\approx4\left|\sum_{\ell=1}^j g_{2\ell}\overline{g_{1\ell}}\right|^2 \, ,
\end{equation}
which may be a significantly different value, depending on the matrix $G$. The other peaks remain essentially unaffected.

If all the control potentials are equal, it holds $\P^{(\{U\})}(E)\equiv0$ due to equation~\eqref{S_jl 2r} and the unitarity of $G$.

A similar situation happens when some of the potentials are close to each other, $U_1\approx\ldots\approx U_j$ ($j\leq r$).

\paragraph*{Remark on the auxiliary lines}

It follows from equation~\eqref{S_jl 2r} that the lines $3,\ldots,r$ can in principle serve as output lines. If the particles passed throught the vertex are gathered on all of the lines $2,3,\ldots,r$, then the total (aggregated) transmission probability 
is
\begin{equation}
\P^{(\{U\})}_\mathrm{tot}(E)=\sum_{j=2}^r|\S_{j1}^{(U)}(E)|^2\,,
\end{equation}
therefore, the cummulative peaks are higher, namely
\begin{equation}
\lim_{E\to U_\ell}\P^{(\{U\})}_\mathrm{tot}(E)\approx4\sum_{j=2}^r |g_{2j}|^2|g_{1j}|^2\,.
\end{equation}
For instance, the cummulative peaks in the Hadamard case $G=(1/\sqrt{r})H^{(r)}$ satisfy $\P^{(\{U\})}_\mathrm{tot}(U_\ell)\approx4(r-1)/r^2$. If $r\gg1$, we have $\P^{(\{U\})}_\mathrm{tot}(U_\ell)\approx4/r$, which is a significantly higher value than the single-channel peak heights $4/r^2$ found in equation~\eqref{peak heights}.


\section{Spectral branching filter}\label{Section: Branching}

In this section we study a filtering device with multiple independently controllable outputs. The device is based on a star graph with $n=2r$ lines, see Figure~\ref{Fig.2r.}:
\begin{itemize}
\item Line 1 is \emph{input}.
\item Lines $2,\ldots,r$ are \emph{outputs}.
\item Line $r+1$ is \emph{drain}.
\item Lines $r+2,\ldots,n$ are \emph{controlling lines}, subjected to adjustable external potentials $U_2,\ldots,U_r$.
\end{itemize}
\begin{figure}[h]
\begin{center}
\includegraphics[width=4.5cm]{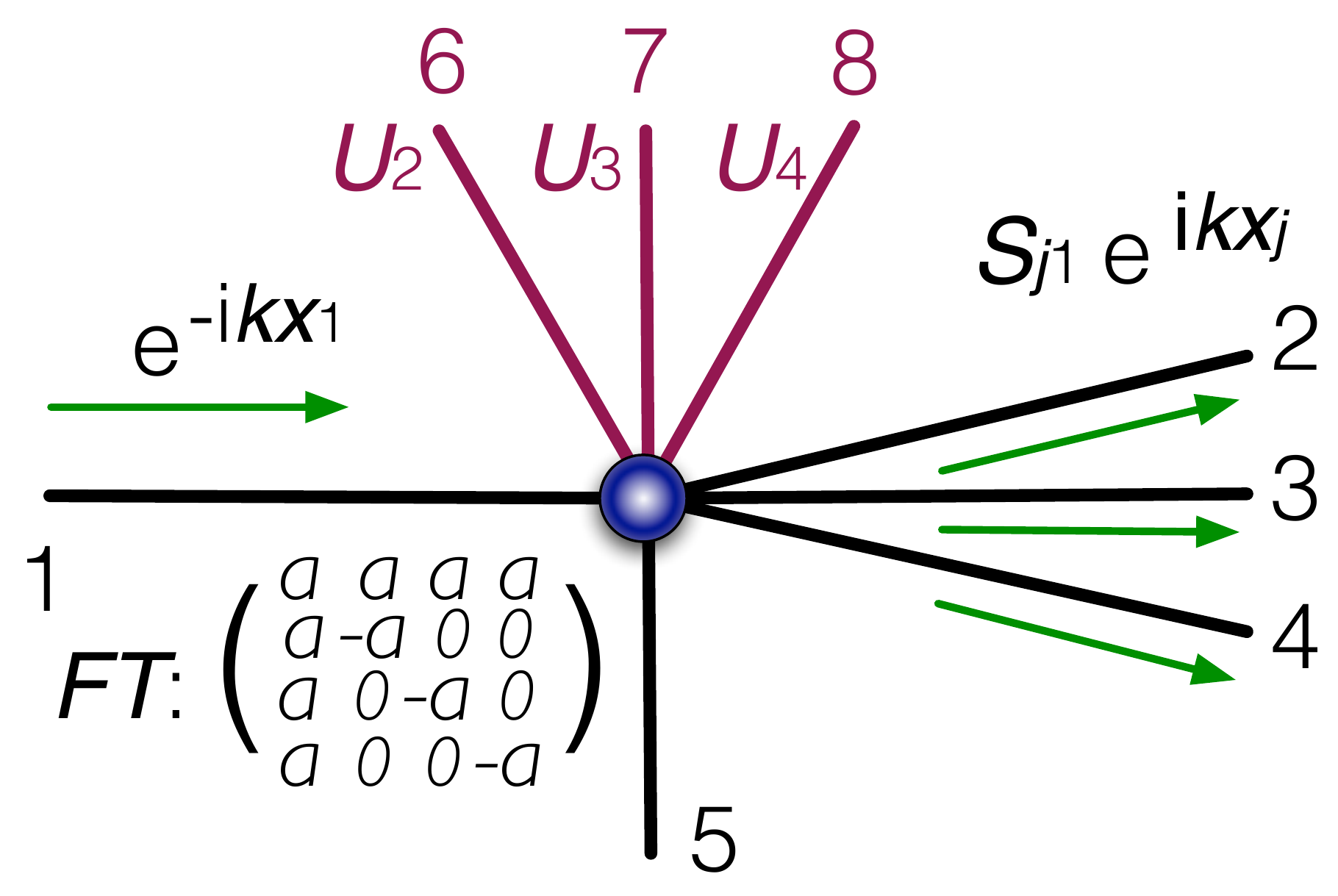}
\end{center}
\caption{Scheme of a quantum branching spectral filter based on a $n=2r$ star graph (for $r=4$). For every $j=2,\ldots,r$, the transmission to the output line $j$ is controlled by the external potential $U_j$ on the line $r+j$.}
\label{Fig.2r.}
\end{figure}
In the center of the graph there is a F\"ul\"op--Tsutsui vertex coupling given by the boundary conditions
\begin{equation}\label{b.c.2r.}
\left(\begin{array}{cc}
I^{(r)} & T \\
0 & 0
\end{array}\right)\Psi'(0)=
\left(\begin{array}{cc}
0 & 0 \\
-T^* & I^{(r)}
\end{array}\right)\Psi(0)
\qquad\text{for}\quad
T=a\begin{pmatrix}
1 & 1 & 1 & \cdots & 1 \\
1 & -1 & 0 & \cdots & 0 \\
1 & 0 & -1 &  & 0 \\
1 & \vdots &   & \ddots & \vdots  \\
1 & 0 & 0 & \cdots & -1
\end{pmatrix}\,,
\end{equation}
where $a>0$ is a parameter such that $a^2\gg1$.

For any $j=2,\ldots,r$, the transmission amplitude of the channel $1\to j$, denoted by $\S_{j1}^{(U)}(E)\equiv[\S(E;0,\ldots,0,U_1,\ldots,U_r)]_{j1}$, is
\begin{equation}\label{S_j1 branch}
\S_{j1}^{(U)}(E)=
\frac{2}{\left(1+a^2 r\right)\left(\frac{1}{a^2}+\xi_j\right)}\left(\xi_j\frac{1+\sum_{\ell=2}^{r}\frac{1}{\frac{1}{a^2}+\xi_\ell}}{1+\sum_{\ell=2}^{r}\frac{\xi_\ell}{\frac{1}{a^2}+\xi_\ell}}-1\right)
\,,
\end{equation}
where $\xi_\ell=\sqrt{1-\frac{U_\ell}{E}}$ for all $\ell=2,\ldots,r$, and the transmission amplitude input~$\to$~drain equals
\begin{equation}
\S_{r+1,1}^{(U)}(E)=\frac{2a}{1+a^2(k+1)}\,.
\end{equation}
With regard to equation~\eqref{S_j1 branch}, it holds
\begin{equation}
\lim_{E\to\infty}\S_{j1}^{(U)}(E)=0 \qquad\forall j=2,\ldots,r.
\end{equation}
Let the nonzero control potentials be mutually different in the sense
\begin{equation}\label{not close U.}
a^2\sqrt{\left|1-\frac{U_{\ell}}{U_j}\right|}\gg1 \quad \text{for all $j,\ell=1,\ldots,r$, $j\neq\ell$ such that $U_j\neq0$},
\end{equation}
and let $h_0$ be the number of the control potentials set to $0$, i.e., $h_0=\#\left\{\ell\in\{2,\ldots,r\}\,|\,U_\ell=0\right\}$.
Then, with regard to the assumption $a^2\gg1$,
\begin{align}
\lim_{E\to0}\S_{j1}^{(U)}(E)&=\frac{2}{1+a^2 r}\cdot\frac{1+h_0+\frac{1}{a^2}}{1+k+\frac{1+k-h_0}{a^2}}\approx0\,, \\ \label{S_j1 U=0}
\lim_{E\to U_j}\S_{j1}^{(U)}(E)&=\frac{-2}{\frac{1}{a^2}+r}\approx\frac{-2}{r}\,,
\end{align}
and furthermore, for every $j,\ell$ such that $U_j\neq0\neq U_\ell$, it holds
\begin{equation}
\lim_{E\to U_\ell}\S_{j1}^{(U)}(E)\approx\frac{2}{r(r-1)}\,.
\end{equation}
The transmission probability to the output line $j$ is given by $\P_{j}^{(U)}(E)=|\S_{j1}^{(U)}(E)|^2$.
If the condition~\eqref{not close  U.} is satisfied, then for every $j=2,\ldots,r$ such that $U_j\neq0$, we have
\begin{align}
\lim_{E\to\infty}\P_{j}^{(U)}(E)&=0\,, \label{lim P branch,infty} \\
\lim_{E\to0}\P_{j}^{(U)}(E)&\approx0\,, \\
\lim_{E\to U_j}\P_{j}^{(U)}(E)&\approx\frac{4}{r^2}\,, \label{lim U_j} \\
\lim_{E\to U_\ell}\P_{j}^{(U)}(E)&\approx\frac{4}{r^2(r-1)^2} \qquad\text{for } U_\ell\neq0\,, \label{lim U_ell}
\end{align}
and $\P_{j1}^{(U)}(E)\approx0$ for all $E$ except for a certain small neighborhoods of nonzero potentials $U_\ell\neq0$.
In other words, if $U_j\neq0$, then the probability $\P_{j}^{(U)}(E)$ has a principal sharp peak of the height $4/r^2$ at $E=U_j$, and secondary, significantly smaller, sharp peaks of the height $4/[r(r-1)]^2$ at $E=U_\ell\neq0$.
The situation is illustrated in Figure~\ref{Fig. P branch}. 
Therefore, the device constructed for $a\gg1$ works as a branching band-pass filter. Every output $j=2,\ldots,r$ is controlled by a dedicated control potential $U_j$ (put on line $j+r$) which determines the position of the principal passband.

\begin{figure}[h]
  \begin{tabular}{cc}
  \includegraphics[width=6.5cm]{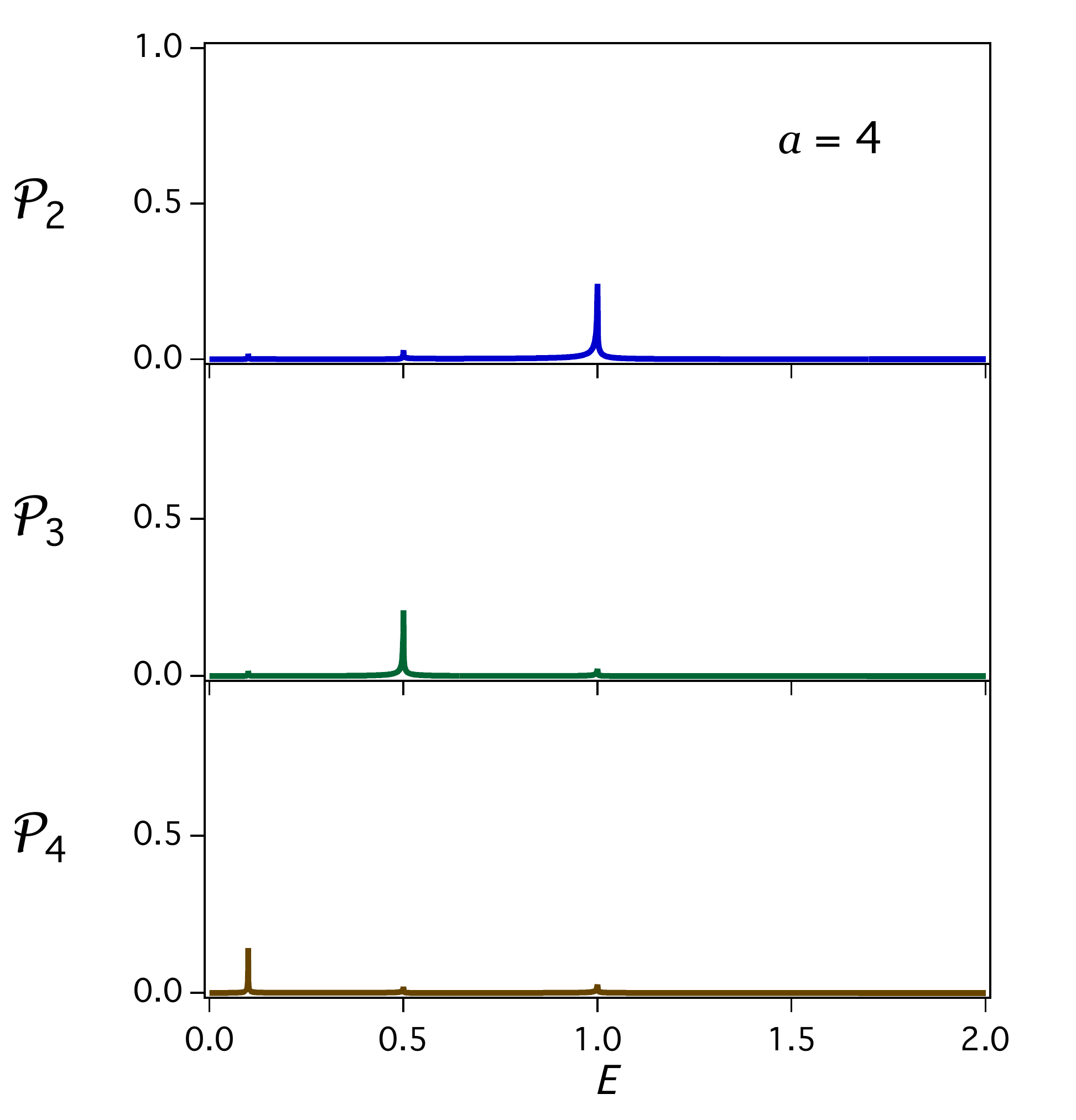} & \includegraphics[width=6.5cm]{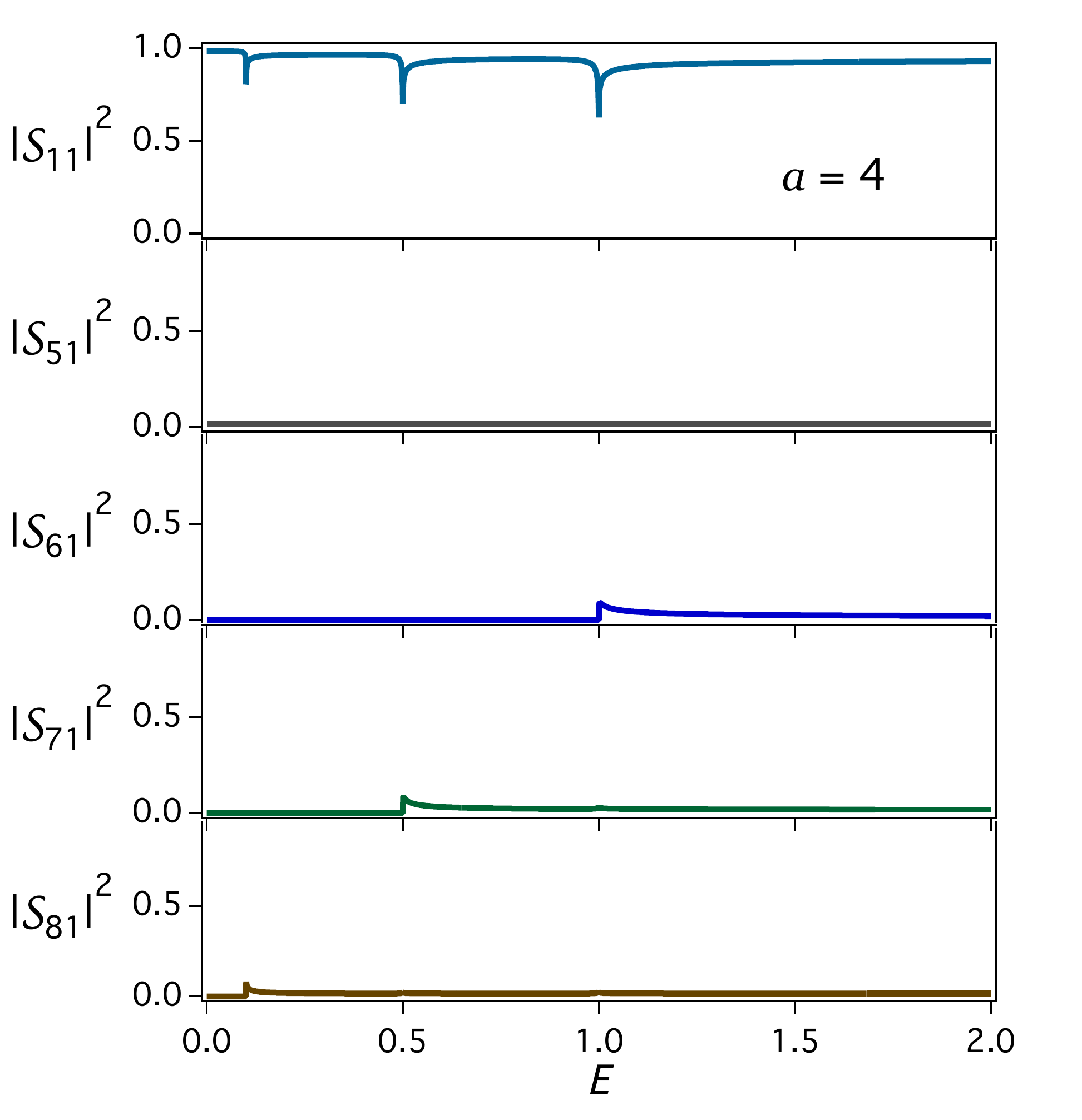}
  \end{tabular}
  \caption{Characteristics of the branching spectral filter, obtained from the graph in Fig.~\ref{Fig.2r.} for $n=2r=8$ and $a=4$. The control potentials are set to $U_2=1$, $U_3=0.5$, $U_4=0.1$. The left figure shows the transmission probabilities $\P_j^{(\{U\})}(E)$ to the outputs $j=2,3,4$, the right figure shows the reflection probability $|S_{11}^{(U)}(E)|^2$, the probability of transmission to the drain and the probabilities of transmission to the controlling lines $|S_{\ell 1}^{(U)}(E)|^2$ for $\ell=r+2,\ldots,n$.}
  \label{Fig. P branch}
\end{figure}

\paragraph*{Identical potentials}

If some of the control potentials are of the same nonzero value, or if their values are very close, it has effect on the secondary peaks at that energy. Namely, equation~\eqref{lim U_ell} is changed to
\begin{equation}
\lim_{E\to U_\ell}\P_{j}^{(U)}(E)\approx\frac{4}{r^2}\cdot\left(\frac{h_\ell}{r-h_\ell}\right)^2 \qquad\text{for } U_\ell\neq0\,,
\end{equation}
where $h_\ell$ can be roughly interpreted as the number of the control potentials that are approximately equal to $U_\ell$, i.e., $h_\ell\approx\#\left\{i\in\{2,\ldots,r\}\,|\,a^2\sqrt{1-U_{i}/U_\ell}\ll1\right\}$. Therefore, the height of the secondary peak at $E=U_\ell$ indicates the number of control potentials equal to $U_\ell$ (or to a certain very close value).

Let us emphasize that the heights of the principal peaks are essentially unaffected.

\paragraph*{The effect of $U_j=0$}

Setting the control potential $U_j$ on the controlling line $j+r$ to $0$ essentially results in cancelling the corresponding principal peak. The limit~\eqref{S_j1 U=0} is changed to
\begin{equation}
\lim_{E\to0}\S_{j1}^{(U)}(E)=\frac{2}{1+a^2 r}\cdot\frac{1+h_0+\frac{1}{a^2}}{1+k+\frac{1+k-h_0}{a^2}}\approx0\,,
\end{equation}
and, therefore, the principal peak height is changed from the value given by equation~\eqref{lim U_j} to
\begin{equation}
\lim_{E\to U_j}\P_{j}^{(U)}(E)\approx0\,.
\end{equation}

If all the control potentials are zero, then $\P_{j}^{(\{U\})}(E)=0$ for all $j=2,\ldots,r$, and the incoming particles are mostly reflected.

\paragraph*{Another operation mode}

One may consider the operation mode when the drain line be subjected to a nonzero potential $U_{1}$. In such a situation a second principal peak at the energy $E=U_1$ appears at every output $j=2,\ldots,r$. More precisely speaking, if $U_j\neq0$ and satisfies the condition $a^2\sqrt{|1-U_0/U_j|}\gg1$, then
\begin{equation}
\lim_{E\to U_1}\P_{j}^{(U)}(E)\approx\frac{4}{r^2}\,,
\end{equation}
whereas the relations~\eqref{lim P branch,infty}--\eqref{lim U_ell} are unaffected.
Therefore, the device with a nonzero potential at the drain line works as a dual-band branching filter with two principal transmission probability peaks at the outputs: the first peak is at $E=U_j\neq0$, where $j$ denotes the output, the second one is at $E=U_1$. Both peaks are essentially of the same height $4/r^2$.

\begin{remark}
The spectral filter designed in Section~\ref{Section: n=4} is obviously a special case of the filter studied in this section, corresponding to $r=2$.
\end{remark}


\section{Practical realization of F\"ul\"op--Tsutsui vertices}\label{Section: Approximation}

The function of the controllable filter devices designed in the previous sections is based on the threshold resonance effects in quantum star graphs with ``exotic'' couplings in the vertices, namely those of the F\"ul\"op--Tsutsui type. It should be emphasized that standard vertex couplings (the free and the $\delta$-coupling) would not work that way. It is 
therefore essential, for the proposed designs to be viable, that the required F\"ul\"op--Tsutsui vertices can be 
physically realized. 
This problem has been addressed in~\cite{CET10} and \cite{CT10}, where it was proved that any F\"ul\"op--Tsutsui coupling
can be approximately constructed by assembling a few short lines in a web with $\delta$-couplings in the vertices and, in some cases, certain vector potentials on the lines. It essentially solves the problem, because the $\delta$-couplings themselves have a simple physical interpretation and can be well approximated by steep smooth potentials \cite{Ex96b}.

In this section we revisit the results of~\cite{CET10} and \cite{CT10}, proposing a construction which is simpler, and, therefore, easier to be experimentally realized. The new construction uses less $\delta$-couplings than the scheme from~\cite{CET10}, and at the same time, in contrast to the scheme from~\cite{CT10}, it requires no vector potentials in case that the matrix $T$ from the boundary conditions~\eqref{FT} is real.

The approximation is constructed in the following way, cf. Figure~\ref{Scheme}:

\begin{figure}[h]
\begin{center}
\includegraphics[width=3.5cm]{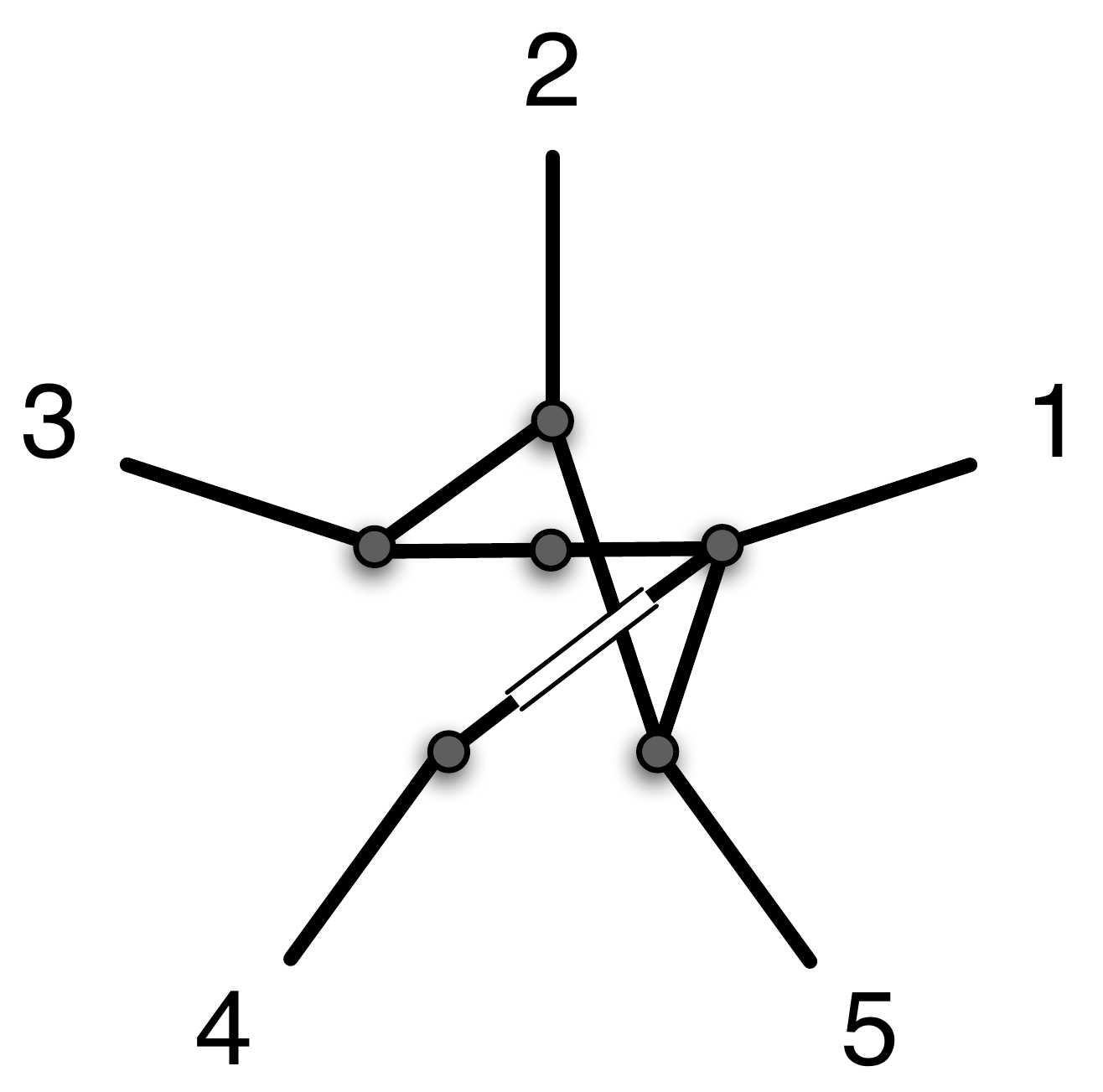}
\end{center}
\caption{An example of the approximation for $n=5$. Certain half lines are connected by short lines (in this example 2--3, 2--5, \ldots), other are without a direct link (e.g., 1--2, 4--5). The bullets represent $\delta$~potentials, the double line symbolizes a vector potential. The connecting lines are generally of different lengths, and each of them carries either a $\delta$-interactions (cf. link 1--3), or a vector potential (cf. 1--4), or none of them (cf. 1--5, 2--3, 2--5), depending on the parameters of the approximated coupling.}
\label{Scheme}
\end{figure}

\begin{itemize}
\item Take $n$ halflines, each parametrized by
$x\in[0,+\infty)$, with the endpoints denoted as $v_j$, $j=1,\ldots,n$.
\item For certain pairs $\{j,\ell\}$ ($j,\ell\in\{1,\ldots,n\}$), join halfline
endpoints $v_j,v_\ell$ by a connecting line of the length $d/\gamma_{j,\ell}$,
where $d$ is a length parameter and $1/\gamma_{j,\ell}$ is a (positive) length coefficient.
The pairs $\{j,\ell\}$ such that the vertices $v_j,v_\ell$ are connected, as well as the values of $\gamma_{j,\ell}$, will be specified later. If the endpoints $v_j,v_\ell$ are not connected, we may formally set $\gamma_{j,\ell}=0$, which represents the ``infinite length'' $d/\gamma_{j,\ell}=\infty$. Note that $\gamma_{j,\ell}=\gamma_{\ell,j}$ for all $j,\ell$. The center of the line connecting $v_j$ and $v_\ell$ will be denoted by $w_{j,\ell}$ (therefore, $w_{j,\ell}\equiv w_{\ell,j}$).

\item Put a $\delta$-coupling with the parameter $\alpha_j(d)$ at the vertex $v_j$ for each $j=1,..., n$.
\item At certain points $w_{j,\ell}$ (to be specified later) place a $\delta$-interaction
with a parameter $\beta_{j,\ell}(d)\neq0$, at the remaining points $w_{j,\ell}$ formally set $\beta_{j,\ell}(d)=0$. (Note that $\beta_{j,\ell}(d)=\beta_{\ell,j}(d)$.)
\item On certain connecting lines not carrying a $\delta$-interaction (i.e., those with $\beta_{j,\ell}(d)=0$) put a constant vector
potential supported by the interval $\left[d/(4\gamma_{j,\ell}),3d/(4\gamma_{j,\ell})\right]$. Its strength will be denoted by $A_{(j,\ell)}(d)$ (if the variable $x$ on the line grows from $v_j$ to $v_\ell$) or $A_{(\ell,j)}(d)$ (if the variable $x$ grows from $v_\ell$ to $v_j$). Both orientation are allowed, and obviously $A_{(\ell,j)}(d)=-A_{(j,\ell)}(d)$. On the remaining lines formally set $A_{(j,\ell)}(d)=0$.
\end{itemize}
In the rest of the section we will specify the dependence of $\gamma_{j,\ell}$, $\alpha_j(d)$, $\beta_{j,\ell}(d)$ and $A_{(j,\ell)}(d)$ on the matrix $T$ and on the length parameter $d$. In order to simplify the notation, we will write $\alpha_j$ instead of $\alpha_j(d)$; the same applies to $\beta_{j,\ell}$, $A_{(j,\ell)}$. We will also use symbol $N_j$ standing for the set containing indices of all the lines that are joined to the $j$-th one, i.e. $\ell\in N_j\Leftrightarrow \gamma_{j,\ell}\neq0$.

For all $j\in\{1,\ldots,n\}$, the wave function on the $j$-th half line is denoted by the symbol $\psi_j(x)$, $x\in[0,\infty)$. The wave function on the line connecting the points $v_j$ and $v_\ell$ is denoted by $\varphi_{(j,\ell)}$, $x\in[0,d/\gamma_{j,\ell}]$. Obviously, one has to take the line orientation into account. We adopt the following convention: If the value $0$ corresponds to the endpoint $v_j$ and the value $d/\gamma_{j,\ell}$ to the endpoint $v_\ell$, then the wave function is denoted by $\varphi_{(j,\ell)}$. In the case of the opposite direction, the wave function is denoted by $\varphi_{(\ell,j)}$. Therefore, $\varphi_{(\ell,j)}(x)=\varphi_{(j,\ell)}\left(\frac{d}{\gamma_{j,\ell}}-x\right)$ for all $x\in\left[0,\frac{d}{\gamma_{j,\ell}}\right]$.

Now we can proceed to determining $\gamma_{j,\ell}$, $\alpha_j(d)$, $\beta_{j,\ell}(d)$ and
$A_{(j,\ell)}(d)$.
At first, let us express the effect of the $\delta$-couplings involved. The $\delta$-coupling at the endpoint of the $j$-th half line ($j=1,\ldots, n$) implies
\begin{subequations}\label{II.serie}
\begin{gather}
\psi_j(0)=\varphi_{(j,\ell)}(0)=\varphi_{(\ell,j)}\left(\frac{d}{\gamma_{j,\ell}}\right)\quad \text{for all}\; \ell\in N_j\,, \label{II.serie.a}
\\
\psi_j'(0)+\sum_{\ell\in N_j}\varphi_{(j,\ell)}'(0)=\alpha_j\cdot\psi_j(0)\,, \label{II.serie.b}
\end{gather}
\end{subequations}
and the $\delta$-interaction in the point $w_{j,\ell}$ means
\begin{subequations}\label{II.serie..}
\begin{gather}
\varphi_{(j,\ell)}\left(\left(\frac{d}{2\gamma_{j,\ell}}\right)_-\right)=\varphi_{(j,\ell)}\left(\left(\frac{d}{2\gamma_{j,\ell}}\right)_+\right)=:\varphi_{(j,\ell)}\left(\frac{d}{2\gamma_{j,\ell}}\right)\quad \text{for all}\; \ell\in N_j\,, \label{II.serie..a}
\\
\varphi_{(j,\ell)}'\left(\left(\frac{d}{2\gamma_{j,\ell}}\right)_+\right)-\varphi_{(j,\ell)}'\left(\left(\frac{d}{2\gamma_{j,\ell}}\right)_-\right)=\beta_{j,\ell}\cdot\varphi_{(j,\ell)}\left(\frac{d}{2\gamma_{j,\ell}}\right)\,, \label{II.serie..b}
\end{gather}
\end{subequations}
where $\varphi(x_+)$ and $\varphi(x_-)$ denote the right-sided limit and the left-sided limit, respectively, of $\varphi$ at $x$.

Further relations which will help us to find the parameters come from Taylor expansions. Consider first the
case without the vector potentials (the effect of the potential will be taken into account later). We have
\begin{eqnarray}\label{bez pot.}
&& \begin{aligned}
&\varphi^0_{(j,\ell)}\left(\frac{d}{2\gamma_{j,\ell}}\right)
=\varphi^0_{(j,\ell)}(0)+\frac{d}{2\gamma_{j,\ell}}\, \varphi^{0\prime}_{(j,\ell)}(0)+\O(d^2)\,,\\
&\varphi^{0\prime}_{(j,\ell)}\left(\left(\frac{d}{2\gamma_{j,\ell}}\right)_-\right)
=\varphi^{0\prime}_{(j,\ell)}(0)+\O(d)\,,
\end{aligned}
\end{eqnarray}
\begin{eqnarray}\label{bez pot..}
&& \begin{aligned}
&\varphi^0_{(j,\ell)}\left(\frac{d}{\gamma_{j,\ell}}\right)
=\varphi^0_{(j,\ell)}\left(\frac{d}{2\gamma_{j,\ell}}\right)+\frac{d}{2\gamma_{j,\ell}}\, \varphi^{0\prime}_{(j,\ell)}\left(\left(\frac{d}{2\gamma_{j,\ell}}\right)_+\right)+\O(d^2)\,,\\
&\varphi^{0\prime}_{(j,\ell)}\left(\frac{d}{\gamma_{j,\ell}}\right)
=\varphi^{0\prime}_{(j,\ell)}\left(\left(\frac{d}{2\gamma_{j,\ell}}\right)_+\right)+\O(d)\,,
\end{aligned}
\end{eqnarray}
where the superscript $0$ refers to zero vector potential on the short line connecting $v_j$ and $v_\ell$.
If the vector potential on a connecting line is nonzero, the relations \eqref{bez pot.} and \eqref{bez pot..} are changed according to the following lemma.
\begin{lemma}\label{potencial}
Let $\lambda\in\R$, $s,t\in\R$, $m\in\R\backslash\{0\}$.
Let $\psi^0_{s,t}$ be the solution of the initial value problem
\begin{equation}
-\frac{\hbar^2}{2m}\frac{\d^2}{\d x^2}\psi=\lambda\psi\,,\qquad \psi(0)=s,\quad {\psi}'(0)=t\,,
\end{equation}
and $\psi^A_{s,t}$ be the solution of the initial value problem
\begin{equation}
\frac{1}{2m}\left(-\i\hbar\frac{\d}{\d x}-qA(x)\right)^2\psi=\lambda\psi\,,\qquad \psi(0)=s,\quad {\psi}'(0)=t \, ,
\end{equation}
for a certain $q\in\R$ and a function $A:\R\to\R$ such that $0\notin\overline{\mathrm{supp}(A)}$, where $\overline{\mathrm{supp}(A)}$ denotes the closure of the support of $A$.
Then
\begin{equation}
\psi^A_{s,t}(x)=\e^{i\frac{q}{\hbar}\int_0^x A(\xi)\d\xi}\cdot\psi^0_{s,t}(x)
\qquad\text{for all } x\in\R\,.
\end{equation}
\end{lemma}
The statement can be proved easily, cf. \cite{SMMC99} or \cite{CET10}. Lemma~\ref{potencial} implies that the constant vector potential $A$ supported by the interval $[L_1,L_2]$ shifts the phase of the wave function and of its derivative at $x>L_2$ by $\i qA(L_2-L_1)/\hbar$. 
Therefore, equations~\eqref{bez pot.} and \eqref{bez pot..} are transformed into
\begin{subequations}\label{III.serie}
\begin{align}
&\varphi_{(j,\ell)}\left(\frac{d}{2\gamma_{j,\ell}}\right) =\e^{\i \frac{d}{4\gamma_{j,\ell}}\frac{qA_{(j,\ell)}}{\hbar}}\left(\varphi_{(j,\ell)}(0)+\frac{d}{2\gamma_{j,\ell}}\, \varphi_{(j,\ell)}'(0)\right)+\O(d^2)\,, \label{III.serie.a} \\
&\varphi_{(j,\ell)}'\left(\left(\frac{d}{2\gamma_{j,\ell}}\right)_-\right)=\e^{\i \frac{d}{4\gamma_{j,\ell}}\frac{qA_{(j,\ell)}}{\hbar}}\cdot\varphi_{(j,\ell)}'(0)+\O(d)\,, \label{III.serie.b}
\end{align}
\end{subequations}
\begin{subequations}\label{III.serie..}
\begin{align}
&\varphi_{(j,\ell)}\left(\frac{d}{\gamma_{j,\ell}}\right) =\e^{\i \frac{d}{4\gamma_{j,\ell}}\frac{qA_{(j,\ell)}}{\hbar}}\left(\varphi_{(j,\ell)}\left(\frac{d}{2\gamma_{j,\ell}}\right)+\frac{d}{2\gamma_{j,\ell}}\, \varphi_{(j,\ell)}'\left(\left(\frac{d}{2\gamma_{j,\ell}}\right)_+\right)\right)+\O(d^2)\,, \label{III.serie..a}\\
&\varphi_{(j,\ell)}'\left(\frac{d}{\gamma_{j,\ell}}\right)=\e^{\i \frac{d}{4\gamma_{j,\ell}}\frac{qA_{(j,\ell)}}{\hbar}}\cdot\varphi_{(j,\ell)}'\left(\left(\frac{d}{2\gamma_{j,\ell}}\right)_+\right)+\O(d)\,. \label{III.serie..b}
\end{align}
\end{subequations}

Equations \eqref{II.serie}, \eqref{II.serie..}, \eqref{III.serie}
and~\eqref{III.serie..} connect values of
the wave functions and their derivatives at all the vertices $v_j$, $j=1, ..., n$. In the next step we
eliminate the terms with the ``auxiliary'' functions
$\varphi_{(j,\ell)}$.
To do so, we at first use equations~\eqref{II.serie..b} and \eqref{III.serie.b} to obtain
\begin{multline}
\varphi_{(j,\ell)}'\left(\left(\frac{d}{2\gamma_{j,\ell}}\right)_+\right)=\varphi_{(j,\ell)}'\left(\left(\frac{d}{2\gamma_{j,\ell}}\right)_-\right)+\beta_{j,\ell}\cdot\varphi_{(j,\ell)}\left(\frac{d}{2\gamma_{j,\ell}}\right)
\\
=\e^{\i \frac{d}{4\gamma_{j,\ell}}\frac{qA_{(j,\ell)}}{\hbar}}\cdot\varphi_{(j,\ell)}'(0)+\beta_{j,\ell}\cdot\varphi_{(j,\ell)}\left(\frac{d}{2\gamma_{j,\ell}}\right)+\O(d)
\,.
\end{multline}
Then equation~\eqref{III.serie..a} gives
\begin{multline}\label{fijl}
\varphi_{(j,\ell)}\left(\frac{d}{\gamma_{j,\ell}}\right)
=\e^{\i \frac{d}{4\gamma_{j,\ell}}\frac{qA_{(j,\ell)}}{\hbar}}\left(\varphi_{(j,\ell)}\left(\frac{d}{2\gamma_{j,\ell}}\right)+\frac{d}{2\gamma_{j,\ell}}\, \left(\e^{\i \frac{d}{4\gamma_{j,\ell}}\frac{qA_{(j,\ell)}}{\hbar}}\cdot\varphi_{(j,\ell)}'(0)+\beta_{j,\ell}\cdot\varphi_{(j,\ell)}\left(\frac{d}{2\gamma_{j,\ell}}\right)\right)\right)+\O(d^2)
\\
=
\e^{\i \frac{d}{4\gamma_{j,\ell}}\frac{qA_{(j,\ell)}}{\hbar}}\left(\left(1+\frac{d}{2\gamma_{j,\ell}}\beta_{j,\ell}\right)\varphi_{(j,\ell)}\left(\frac{d}{2\gamma_{j,\ell}}\right)+
\frac{d}{2\gamma_{j,\ell}}\, \e^{\i \frac{d}{4\gamma_{j,\ell}}\frac{qA_{(j,\ell)}}{\hbar}}\cdot\varphi_{(j,\ell)}'(0)\right)+\O(d^2)\,.
\end{multline}
We substitute 
for $\varphi_{(j,\ell)}(d/(2\gamma_{j,\ell}))$ in equation~\eqref{fijl} from equation~\eqref{III.serie.a}, and find
\begin{equation}
\varphi_{(j,\ell)}\left(\frac{d}{\gamma_{j,\ell}}\right)=
\e^{\i \frac{d}{2\gamma_{j,\ell}}\frac{qA_{(j,\ell)}}{\hbar}}\left(\left(1+\frac{d}{2\gamma_{j,\ell}}\beta_{j,\ell}\right)\varphi_{(j,\ell)}(0)+
\frac{d}{2\gamma_{j,\ell}}\left(2+\frac{d}{2\gamma_{j,\ell}}\beta_{j,\ell}\right)\cdot\varphi_{(j,\ell)}'(0)\right)+\O(d^2)\,,
\end{equation}
then we apply equation~\eqref{II.serie.a} to rewrite the equation above in the form
\begin{equation}
\psi_\ell(0)=
\e^{\i \frac{d}{2\gamma_{j,\ell}}\frac{qA_{(j,\ell)}}{\hbar}}\left(\left(1+\frac{d}{2\gamma_{j,\ell}}\beta_{j,\ell}\right)\psi_j(0)+
\frac{d}{2\gamma_{j,\ell}}\left(2+\frac{d}{2\gamma_{j,\ell}}\beta_{j,\ell}\right)\cdot\varphi_{(j,\ell)}'(0)\right)+\O(d^2)\,,
\end{equation}
hence we obtain
\begin{equation}
d\varphi_{(j,\ell)}'(0)=
-2\gamma_{j,\ell}\frac{1+\frac{d}{2\gamma_{j,\ell}}\beta_{j,\ell}}{2+\frac{d}{2\gamma_{j,\ell}}\beta_{j,\ell}}\psi_j(0)+
\e^{-\i \frac{d}{2\gamma_{j,\ell}}\frac{qA_{(j,\ell)}}{\hbar}}\frac{2\gamma_{j,\ell}}{2+\frac{d}{2\gamma_{j,\ell}}\beta_{j,\ell}}\cdot\psi_\ell(0)+\O(d^2)\,,
\end{equation}
and finally we sum the last equation over $k$, which yields
\begin{equation}\label{suma}
d\sum_{\ell=1}^n\varphi_{(j,\ell)}'(0)=-\sum_{\ell=1}^n 2\gamma_{j,\ell}\frac{1+\frac{d}{2\gamma_{j,\ell}}\beta_{j,\ell}}{2+\frac{d}{2\gamma_{j,\ell}}\beta_{j,\ell}}\psi_j(0)
+\e^{-\i \frac{d}{2\gamma_{j,\ell}}\frac{qA_{(j,\ell)}}{\hbar}}\frac{2\gamma_{j,\ell}}{2+\frac{d}{2\gamma_{j,\ell}}\beta_{j,\ell}}\cdot\psi_\ell(0)+\O(d^2)\,.
\end{equation}
According to equation \eqref{II.serie.b}, the LHS of equation~\eqref{suma} is equal to $-d\psi_j'(0)+d\alpha_j\psi_j(0)$. This fact allows us to eliminate the term $\sum_{\ell=1}^n\varphi_{(j,\ell)}'(0)$ from equation~\eqref{suma}, which leads to
\begin{equation}\label{j-ty radek}
d\psi_j'(0)=\left(d\alpha_j+\sum_{\ell=1}^n 2\gamma_{j,\ell}\frac{1+\frac{d}{2\gamma_{j,\ell}}\beta_{j,\ell}}{2+\frac{d}{2\gamma_{j,\ell}}\beta_{j,\ell}}\right)\psi_j(0)-
\sum_{\ell=1}^n\,\e^{-\i \frac{d}{2\gamma_{j,\ell}}\frac{qA_{(j,\ell)}}{\hbar}}\frac{2\gamma_{j,\ell}}{2+\frac{d}{2\gamma_{j,\ell}}\beta_{j,\ell}}\cdot\psi_\ell(0)
+\O(d^2)
\end{equation}
for all $j=1,\ldots,n$.

Our objective is to choose $\alpha_j(d)$, $\gamma_{j,\ell}$, $\beta_{j,\ell}(d)$ and $A_{(j,\ell)}(d)$ in such a way that the system of equations~\eqref{j-ty radek} with $j=1,\ldots,n$ is equivalent to the boundary conditions~\eqref{FT} in the limit $d\to0$.
It is convenient to adopt the following convention on a shift of the column indices of $T$:
\begin{notation}\label{indT}
The lines of the matrix $T$ are indexed from 1 to $r$, the columns
are indexed from $r+1$ to $n$.
\end{notation}
Taking this convention into account, we rewrite the boundary conditions~\eqref{FT}
as a system of $n$ equations,
\begin{subequations}\label{FTFT}
\begin{gather}
\psi_j'(0)+\sum_{\ell=r+1}^n t_{j\ell}\psi'_\ell(0)
=0 \qquad j=1,\ldots,r\,, \label{j<=m} \\
0=-\sum_{\ell=1}^r \overline{t_{\ell j}}\psi_\ell(0)+\psi_j(0)
\qquad j=r+1,\ldots,n \label{j>=m+1} \,.
\end{gather}
\end{subequations}
The rows of the system~\eqref{FTFT} are of two types. We start with the type \eqref{j>=m+1}, corresponding to
$j\geq r+1$. We see immediately that if the parameters $\alpha_j$, $\gamma_{j,\ell}$, $\beta_{j,\ell}$ and $A_{(j,\ell)}$ are chosen such that
\begin{gather}
\gamma_{j,\ell}=0\qquad\text{for $\ell\geq r+1$}\,, \label{implicitni_g}\\
\e^{-\i \frac{d}{2\gamma_{j,\ell}}\frac{qA_{(j,\ell)}}{\hbar}}\frac{2\gamma_{j,\ell}}{2+\frac{d}{2\gamma_{j,\ell}}\beta_{j,\ell}}=\overline{t_{\ell j}} \qquad \text{for $\ell\leq r$}\,, \label{implicitni_b}\\
d\alpha_j+\sum_{\ell=1}^n 2\gamma_{j,\ell}-\sum_{\ell=1}^n\frac{2\gamma_{j,\ell}}{2+\frac{d}{2\gamma_{j,\ell}}\beta_{j,\ell}}=1\,, \label{implicitni_c}
\end{gather}
then \eqref{j-ty radek} tends to
\begin{equation}\label{j>=m+1.}
0=\psi_j(0)-\sum_{\ell=1}^r \overline{t_{\ell j}}\psi_\ell(0) \qquad \text{for every } j=r+1,\ldots, n \, ,
\end{equation}
in the limit $d\to0$, and, therefore, equations~\eqref{j-ty radek} for $j\geq r+1$ are equivalent to equations~\eqref{j>=m+1}.
It suffices to find a solution of the system of equations~\eqref{implicitni_g}--\eqref{implicitni_c}.
Equation~\eqref{implicitni_g} implies that there is no connecting line between $v_j$ and $v_\ell$ if $\ell\geq r+1$.
If $\ell\leq r$, we use equation~\eqref{implicitni_b}; it is easy to check that it is satisfied by
\begin{subequations}\label{spoj}
\begin{align}
\gamma_{j,\ell}&=|t_{\ell j}|\,, \label{gamma} \\
A_{(j,\ell)}&=\left\{\begin{array}{cl}
0 & \text{for $t_{\ell j}\in\R$} \\
\frac{2}{d}\cdot\frac{\hbar}{q}|t_{\ell j}|\arg t_{\ell j} & \text{for $t_{\ell j}\in\C\backslash\R$}
\end{array}\right. \label{A} 
\, ,\\
\beta_{j,\ell}&=\left\{\begin{array}{cl}
0 & \text{for $t_{\ell j}>0 \vee t_{\ell j}\notin\R$} \\
-\frac{8}{d}|t_{\ell j}| & \text{for $t_{\ell j}<0$}
\end{array}\right. \label{beta}
\, .
\end{align} 
\end{subequations}

\begin{remark}\label{delta xor vector}
For every $j\geq r+1$, the point $v_j$ is connected to a point $v_\ell$ if and only if $\ell\leq r$ and, moreover, $t_{\ell j}\neq0$. Furthermore, if $v_j$ and $v_\ell$ are connected, then the $\delta$-interaction in $w_{j,\ell}$ is present iff $t_{\ell j}<0$, and the vector potential on the short line connecting $v_j$ and $v_\ell$ is present iff $t_{\ell j}\notin\R$. Consequently, none of the short lines connected to $v_j$, $j\geq r+1$, carries both the $\delta$-coupling and a vector potential.
\end{remark}

Finally, we use equation~\eqref{implicitni_c} together with~\eqref{spoj} to find $\alpha_j$:
\begin{equation}\label{alpha}
\alpha_j=\frac{1}{d}\left(1-\sum_{k=1}^r |t_{kj}|+2\sum_{\{k\leq r|t_{kj}<0\}} t_{kj}\right)\,.
\end{equation}

The case $j\geq r+1$ is now solved. Let us proceed to the case $j\leq r$. We
substitute for $\psi_\ell'(0)$ in the left-hand side of equation~\eqref{j<=m} from equation~\eqref{j-ty radek}, and in this way we get
\begin{multline}\label{j-ty radek1}
\psi_j'(0)+\sum_{k=r+1}^n t_{jk}\cdot\psi_k'(0)
\\
=\left(\alpha_j+\frac{1}{d}\sum_{\ell=1}^n 2\gamma_{j,\ell}\frac{1+\frac{d}{2\gamma_{j,\ell}}\beta_{j,\ell}}{2+\frac{d}{2\gamma_{j,\ell}}\beta_{j,\ell}}\right)\psi_j(0)
-\frac{1}{d}\sum_{\ell=1}^n\e^{-\i \frac{d}{2\gamma_{j,k}}\frac{qA_{(j,\ell)}}{\hbar}}\frac{2\gamma_{j,\ell}}{2+\frac{d}{2\gamma_{j,\ell}}\beta_{j,\ell}}\psi_\ell(0)+\\
+\sum_{k=r+1}^n
t_{jk}\left(\left(\alpha_k+\frac{1}{d}\sum_{\ell\in N_k}2\gamma_{l\ell}\frac{1+\frac{d}{2\gamma_{k,\ell}}\beta_{k,\ell}}{2+\frac{d}{2\gamma_{k,\ell}}\beta_{k,\ell}}\right)\psi_k(0)
-\frac{1}{d}\sum_{\ell=1}^n\e^{-\i \frac{d}{2\gamma_{k,\ell}}\frac{qA_{(k,\ell)}}{\hbar}}\frac{2\gamma_{k,\ell}}{2+\frac{d}{2\gamma_{k,\ell}}\beta_{k,\ell}}\psi_\ell(0)\right)+\O(d)\,.
\end{multline}
Then we substitute all the already derived parameters from equations~\eqref{implicitni_g}, \eqref{spoj} and \eqref{alpha} into equation~\eqref{j-ty radek1}. After a manipulation we arrive at
\begin{multline}\label{j-ty radek11}
\psi_j'(0)+\sum_{k=r+1}^n t_{jk}\cdot\psi_k'(0)
\\
=\left(\alpha_j+\frac{1}{d}\sum_{\ell=1}^r 2\gamma_{j,\ell}\frac{1+\frac{d}{2\gamma_{j,\ell}}\beta_{j,\ell}}{2+\frac{d}{2\gamma_{j,\ell}}\beta_{j,\ell}}+\frac{1}{d}\sum_{k=r+1}^n |t_{jk}|-\frac{2}{d}\sum_{\{k>r|t_{jk}<0\}} t_{jk}-\frac{1}{d}\sum_{k=r+1}^n |t_{jk}|^2
\right)\psi_j(0)\\
-\frac{1}{d}\sum_{\ell\leq r,\ell\neq j}\left(\e^{-\i \frac{d}{2\gamma_{j,\ell}}\frac{qA_{(j,\ell)}}{\hbar}}\frac{2\gamma_{j,\ell}}{2+\frac{d}{2\gamma_{j,\ell}}\beta_{j,\ell}}
+\sum_{k=r+1}^nt_{jk}\overline{t_{\ell k}}\right)\psi_\ell(0)+\O(d)\,.
\end{multline}
Similarly as in the case $j\geq r+1$, we require that the system~\eqref{j-ty radek11} is equivalent to the system~\eqref{j<=m} in the limit $d\to0$. In this way we obtain a system of (sufficient) conditions on the sought parameters $\alpha_j$, $\gamma_{j\ell}$ and $A_{(j,\ell)}$ ($j\leq r$), namely
\begin{gather}
\e^{-\i \frac{d}{2\gamma_{j,\ell}}\frac{qA_{(j,\ell)}}{\hbar}}\frac{2\gamma_{j,\ell}}{2+\frac{d}{2\gamma_{j,\ell}}\beta_{j,\ell}}
+\sum_{k=r+1}^nt_{jk}\overline{t_{\ell k}}=0 \qquad \text{for } \ell\leq r\,, \label{g,b,A} \\
\alpha_j+\frac{1}{d}\sum_{\ell=1}^r 2\gamma_{j,\ell}\frac{1+\frac{d}{2\gamma_{j,\ell}}\beta_{j,\ell}}{2+\frac{d}{2\gamma_{j,\ell}}\beta_{j,\ell}}+\frac{1}{d}\sum_{k=r+1}^n |t_{jk}|-\frac{2}{d}\sum_{\{k>r|t_{jk}<0\}} t_{jk}-\frac{1}{d}\sum_{k=r+1}^n |t_{jk}|^2=0\,. \label{a}
\end{gather}
Hence we derive explicit expressions for $\alpha_j$, $\gamma_{j,\ell}$, $\beta_{j,\ell}$ and $A_{(j,\ell)}$ ($j\leq r$). From~\eqref{g,b,A} we have
\begin{equation}\label{gamma j<=m}
\gamma_{j,\ell}=\sum_{k=r+1}^nt_{jk}\overline{t_{\ell k}}=(TT^*)_{j\ell}\,,
\end{equation}
therefore, there is a connecting line between $v_j$ and $v_\ell$ if and only if $(TT^*)_{j\ell}\neq0$.
Furthermore, for $(TT^*)_{j\ell}\neq0$, we use equation~\eqref{g,b,A} to determine $\beta_{j,\ell}$ and $A_{(j,\ell)}=0$:
\begin{subequations}\label{beta,A j<=m}
\begin{itemize}
\item If $(TT^*)_{j\ell}<0$, then
\begin{equation}
\beta_{j,\ell}=0\,, \quad A_{(j,\ell)}=0\,.
\end{equation}
\item If $(TT^*)_{j\ell}>0$, then
\begin{equation}
\beta_{j,\ell}=-\frac{8}{d}(TT^*)_{j\ell}\,, \quad A_{(j,\ell)}=0\,.
\end{equation}
\item If $(TT^*)_{j\ell}\notin\R$, then
\begin{equation}
\beta_{j,\ell}=0\,, \quad A_{(j,\ell)}=-\frac{2}{d}\cdot\frac{\hbar}{q}\left|(TT^*)_{j\ell}\right|\cdot\arg\left(-(TT^*)_{j\ell}\right)\,.
\end{equation}
\end{itemize}
\end{subequations}
Note that on none of the short connecting lines there is both the $\delta$-coupling and a vector potential (similarly as in the case $j\geq r+1$, cf. Remark~\ref{delta xor vector}).

Finally, we substitute $\gamma_{j,\ell}$, $\beta_{j,\ell}$ and $A_{(j,\ell)}$ found above into equation~\eqref{a}, and obtain the expression for $\alpha_j$, $j\leq r$:
\begin{equation}\label{alpha j<=m}
\alpha_j=\frac{1}{d}\left(\sum_{k=r+1}^n |t_{jk}|^2-\sum_{k=r+1}^n |t_{jk}|+\sum_{\{k>r|t_{jk}<0\}} 2t_{jk}
-\sum_{\ell=1,\ell\neq j}^r\left|(TT^*)_{j\ell}\right|-\sum_{\{\ell\leq r|\ell\neq j \wedge (TT^*)_{j\ell}>0\}}^r 2(TT^*)_{j\ell}\right)\,. \end{equation}

\subsubsection*{Example of the construction}

We conclude this section by a practical demonstration. We use the technique to construct approximations of the F\"ul\"op--Tsutsui couplings needed for the filtering devices designed in this paper.

Let us consider the coupling used in Section~\ref{Section: n=3}, corresponding to $T=\begin{pmatrix}a&b\end{pmatrix}$. Since $r=1$, there is no connecting line between the endpoints $v_2$ and $v_3$ due to equation~\eqref{implicitni_g}, cf. also Remark~\ref{delta xor vector}. If we substitute $t_{12}=a$, $t_{13}=b$ into equation~\eqref{gamma}, we obtain $\gamma_{21}=|a|$, $\gamma_{31}=|b|$, therefore, the endpoints $v_1$ and $v_2$ are connected by a short line of the length $d/|a|$, the points $v_1$ and $v_3$ by a short line of the length $d/|b|$, see Figure~\ref{fig:6a}. Since we have assumed $a,b\in\R$, we get $A_{(3,1)}=A_{(3,2)}=0$ from equation~\eqref{A}; moreover, the device is contructed for $a,b>0$, hence $\beta_{3,1}=\beta_{3,2}=0$ due to equation~\eqref{beta}. Therefore, each short line carries neither a $\delta$-interaction, nor a vector potential. The parameters of the $\delta$-couplings in the endpoints $v_2$ and $v_3$ follow from equation~\eqref{alpha}: $\alpha_2=(1-a)/d$, $\alpha_3=(1-b)/d$. Finally, the parameter of the $\delta$-coupling in the endpoint $v_1$ is given by equation~\eqref{alpha j<=m}: $\alpha_1=(a^2+b^2-a-b)/d$. The quantity $d$ appearing in the formulae is the length parameter. The small size limit $d\to 0$, together with the $\delta$~potentials parameters properly scaled according to the formulae above, effectively produces the required F\"ul\"op--Tsutsui vertex coupling.

\begin{figure}[h]
\begin{center}
  \includegraphics[width=5cm]{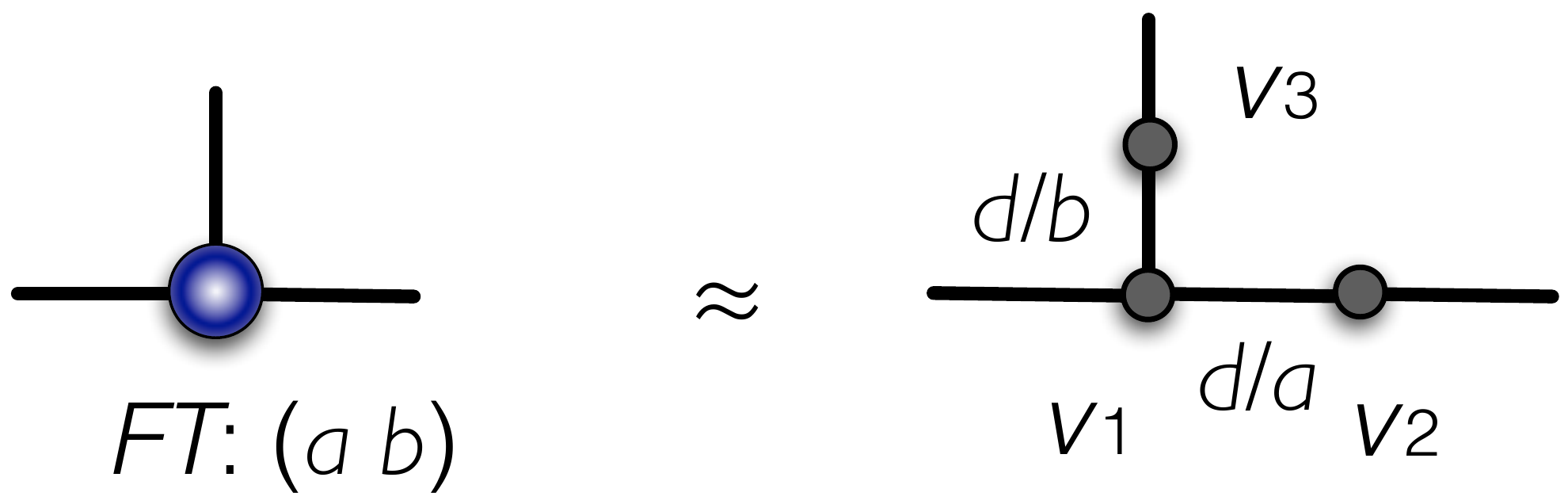}
\end{center}
  \caption[Text excluding the matrix]{Approximate construction of the F\"ul\"op--Tsutsui couplings for $n=3$ and $T=\begin{pmatrix}a&b\end{pmatrix}$. 
  For the special choice $a=1$ corresponding to the filter constructed in Sect.~\ref{Section: n=3} (cf. eq.~\eqref{a=1,b>>1}), 
  the parameters of the $\delta$~potentials in the points $v_1$, $v_2$ and $v_3$ are
  $\alpha_1=b(b-1)/d$, $\alpha_2=0$ (i.e., there is no $\delta$-interaction in the point $v_2$) and $\alpha_3=(1-b)/d$.
  }
  \label{fig:6a}
\end{figure}

The approximation of the coupling corresponding to $T=\begin{pmatrix}a&a\\a&-a\end{pmatrix}$, used in Section~\ref{Section: n=4}, can be obtained in the same way. The result is illustrated in Figure~\ref{fig:6b}.

\begin{figure}[h]
\begin{center}
  \includegraphics[width=5cm]{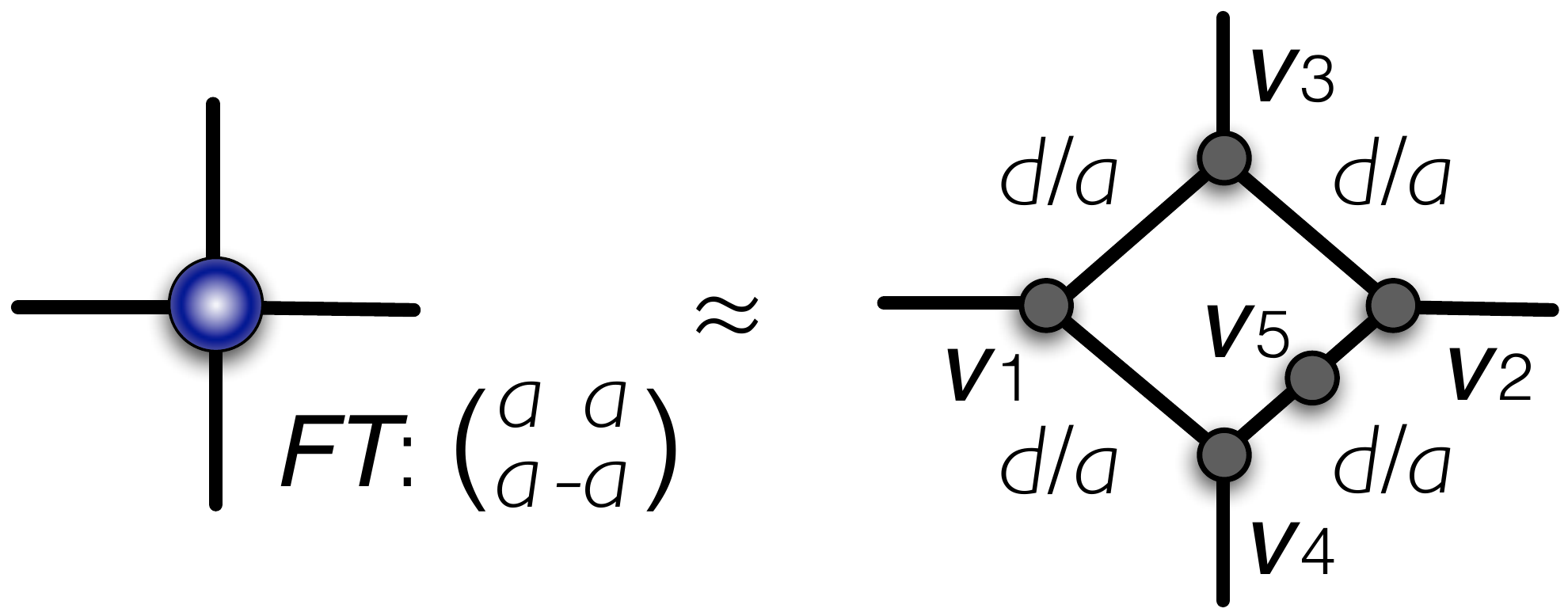}
\end{center}
  \caption[Text excluding the matrix]{Approximate construction of the F\"ul\"op--Tsutsui couplings for $n=4$ and $T=\begin{pmatrix}a&a\\a&-a\end{pmatrix}$. 
  The parameters of the $\delta$~potentials in the points $v_1$, $v_2$, $v_3$, $v_4$ and $w_{2,4}$ are $\alpha_1=2a(a-1)/d$, $\alpha_2=2a(a-2)/d$, $\alpha_3=(1-2a)/d$, $\alpha_4 =(1-4a)/d$ and $\beta_{2,4}=-8a/d$, respectively.}
  \label{fig:6b}
\end{figure}

Finally, let us briefly comment on the case when $T=aG$ for a unitary $r\times r$ matrix $G$; the corresponding F\"ul\"op--Tsutsui coupling is essential for constructing the device designed in Section~\ref{Section: n=2r}. It holds:
\begin{itemize}
\item If $j,\ell\geq r+1$, then the endpoints $v_j$ and $v_\ell$ are not connected due to equation~\eqref{implicitni_g},
\item If $j,\ell\leq r$, then the endpoints $v_j$ and $v_\ell$ are not connected, which follows from the property $TT^*=a^2 I$ and equation~\eqref{gamma j<=m}.
\end{itemize}
Consequently, the short lines connect only pairs $v_j$, $v_\ell$ such that $\ell\leq r<j$. In other words, they form a bipartite graph in which $\{v_1,\ldots,v_r\}$ and $\{v_{r+1},\ldots,v_n\}$ are independent sets of vertices, see Figure~\ref{Fig:bipartite}. In addition, if the matrix $G$ is chosen as a multiple of an Hadamard matrix, then the connecting lines are all of identical length $da/\sqrt{r}$ for $d$ being the length parameter.

\begin{figure}[h]
\begin{center}
\includegraphics[width=6cm]{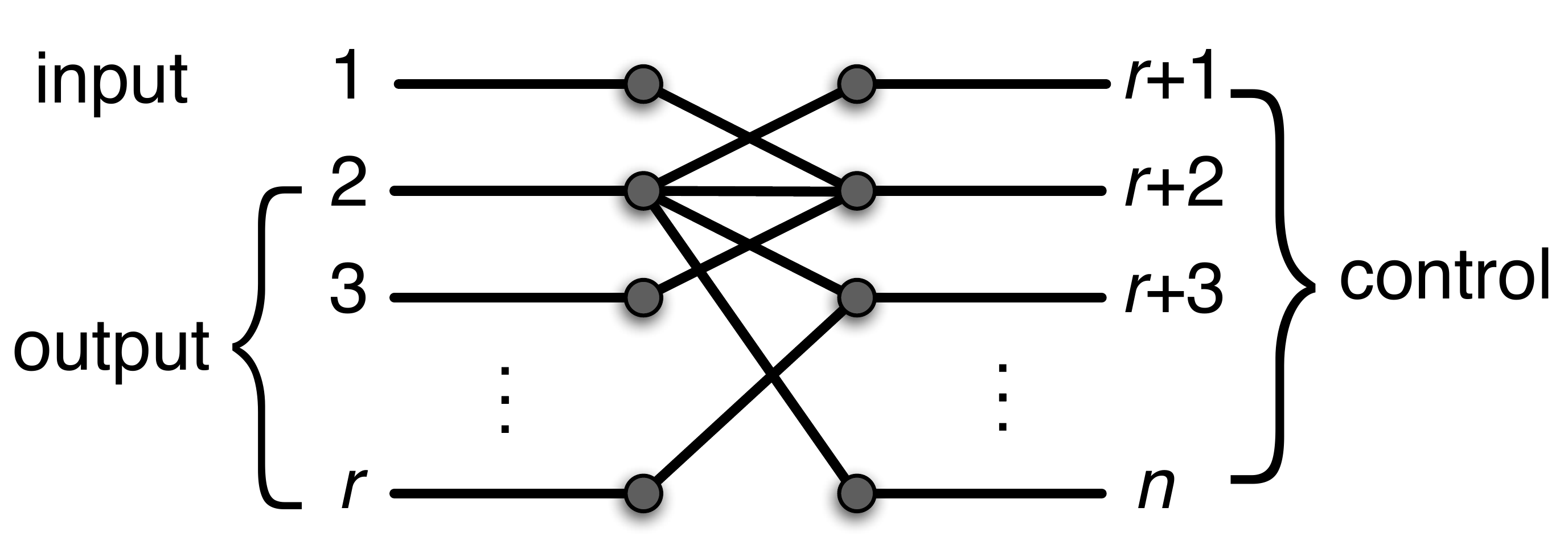}
\end{center}
\caption{Illustration of the approximate construction of the F\"ul\"op--Tsutsui couplings for $n=2r=8$ and $T=aG$, $G\in\mathrm{U}(r)$. Couplings of this type are used in Sect.~\ref{Section: n=2r}. The approximating graph in the center has a bipartite structure. The input and the output lines are on the left hand side, the controlling lines are on the right hand side. Some of the connecting lines carry a $\delta$-interaction or a vector potential, depending on the phases of $t_{j\ell}$.}
\label{Fig:bipartite}
\end{figure}


\section{Scattering matrix of the approximation}\label{Section: S-matrix}

Now we demonstrate the convergence of the approximating scheme proposed in Section~\ref{Section: Approximation}. The proof is performed in terms of the scattering matrices. We calculate the on-shell scattering matrix of the approximating graph, and show that it tends to the scattering matrix of the approximated F\"ul\"op--Tsutsui coupling in the limit $d\to0$.
This approach is different from the previous papers on this topic \cite{CET10,CT10}, where the proof was either omitted, or consisted in performing a full, but very lengthy and cumbersome, demonstration of the norm-resolvent convergence.
In addition, the convergence of the scattering matrices is proved here in a general setting when the lines coupled in the vertex carry potentials $V_1,\ldots,V_n$, which is another 
novel point of this paper.

In what follows, we denote the on-shell scattering matrix of the approximating graph by
\begin{equation}
\S_d(E;V_1,\ldots,V_n)\,.
\end{equation}
The matrix $\S_d(E;V_1,\ldots,V_n)$ will be calculated in the following three steps:
\begin{enumerate}
\item We derive relations between the on-shell boundary values $\psi_{j}(0)$ and derivatives $\psi_{j}'(0)$ by analyzing the physical properties of the internal web.
\item We substitute the final-state wave function components $\Psi^{(\ell)}_{j}(x)$, expressed in terms of $\S_d(E;V_1,\ldots,V_n)$ (cf. eq.~\eqref{psi_ji}), into the system of equations obtained in the first step.
\item We solve the system, in this way we find $\S_d(E;V_1,\ldots,V_n)$.
\end{enumerate}

\emph{Step 1.}
Let $\psi_j$ and $\varphi_{j\ell}$  ($j=1,\ldots,n$, $\ell\in N_j$) be the components of a wave function on the approximated graph, corresponding to a given energy $E$. They satisfy the boundary conditions in the vertices; our goal is to derive relations among the boundary values $\psi_{j}(0)$, $\psi_{j}'(0)$.
The $\delta$-coupling at the point $v_j$ implies
\begin{equation}\label{phipsi}
\sum_{h\in N_j}\varphi_{j h}'(0)+\psi_j'(0)=\alpha_j\psi_j(0) \qquad\text{for all $j=1,\ldots,n$}\,,
\end{equation}
where $\varphi_{j h}(x)$ is the wave function component on the line connecting $v_j$ and $v_h$, with the convention $x=0$ at $v_j$ and $x=d/\gamma_{j,h}$ at $v_h$.

We want to express $\varphi_{j h}'(0)$ in terms of $\psi_j(0)$ and $\psi_h(0)$; the result will be then substituted into equation~\eqref{phipsi}. Let us assume at first that the short line connecting $v_j$ and $v_h$ carries a $\delta$-interaction with parameter $\beta_{j,h}$ in the center. The corresponding wave function component is given by
\begin{equation}\label{phiCD}
\varphi_{j h}(x)=\left\{\begin{array}{cl}
C^+\e^{\i kx}+C^-\e^{-\i kx} & \text{for all}\quad x\in\left[0,\frac{d}{2\gamma_{j,h}}\right]\,, \\ [1em]
D^+\e^{\i kx}+D^-\e^{-\i kx} & \text{for all}\quad x\in\left[\frac{d}{2\gamma_{j,h}},\frac{d}{\gamma_{j,h}}\right]\,,
\end{array}\right.
\end{equation}
where $C^+,C^-$, $D^+,D^-$ are certain constants and
\begin{equation}\label{k}
k=\sqrt{\frac{2mE}{\hbar^2}}\,.
\end{equation}
Everywhere in this section, the symbol $k$ stands exclusively for the quantity given by equation~\eqref{k}.
From equation~\eqref{phiCD}, we have
\begin{subequations}\label{elimCD}
\begin{align}
\psi_j(0)=&\varphi_{j h}(0)=C^++C^-\,, \\
\psi_h(0)=&\varphi_{j h}\left(\frac{d}{\gamma_{j,h}}\right)=D^+\e^{\i k\frac{d}{\gamma_{j,h}}}+D^-\e^{-\i k\frac{d}{\gamma_{j,h}}}\,.
\end{align}
Furthermore, the $\delta$-interaction in the center of the line means
$$
\varphi_{j h}((\frac{d}{2\gamma_{j,h}})_-)=\varphi_{j h}((\frac{d}{2\gamma_{j,h}})_+)
\qquad\text{and}\qquad
\varphi_{j h}'((\frac{d}{2\gamma_{j,h}})_+)-\varphi_{j h}'((\frac{d}{2\gamma_{j,h}})_-)=\beta_{j,h}\varphi_{j h}(\frac{d}{2\gamma_{j,h}})\,,
$$
hence
\begin{gather}
C^+\e^{\i k\frac{d}{2\gamma_{j,h}}}+C^-\e^{-\i k\frac{d}{2\gamma_{j,h}}}=D^+\e^{\i k\frac{d}{2\gamma_{j,h}}}+D^-\e^{-\i k\frac{d}{2\gamma_{j,h}}}\,, \\
\i kC^+\e^{\i k\frac{d}{2\gamma_{j,h}}}-\i kC^-\e^{-\i k\frac{d}{2\gamma_{j,h}}}-\i kD^+\e^{\i k\frac{d}{2\gamma_{j,h}}}+\i kD^-\e^{-\i k\frac{d}{2\gamma_{j,h}}}=\beta_{j,h}\left(C^+\e^{\i k\frac{d}{2\gamma_{j,h}}}+C^-\e^{-\i k\frac{d}{2\gamma_{j,h}}}\right)\,.
\end{gather}
\end{subequations}
By solving the system~\eqref{elimCD}, we find $C^+,C^-$ and $D^+,D^-$ in terms of $\psi_j(0)$ and $\psi_h(0)$.
Then we substitute $C^+,C^-$ into the equation
\begin{equation}
\varphi_{j h}'(0)=\i k\left(C^+-C^-\right)\,,
\end{equation}
which follows from~\eqref{phiCD}, and in this way we find
\begin{equation}\label{phi' delta}
\varphi_{j h}'(0)=\frac{k}{\sin k\frac{d}{\gamma_{j,h}}+\frac{\beta_{j,h}}{k}\sin^2 k\frac{d}{2\gamma_{j,h}}}\left(-\psi_j(0)\left(\cos k\frac{d}{\gamma_{j,h}}+\frac{\beta_{j,h}}{2k}\sin k\frac{d}{\gamma_{j,h}}\right)+\psi_h(0)\right)\,.
\end{equation}
Recall that the formula~\eqref{phi' delta} has been derived on the assumption that the short line connecting $v_j$ and $v_h$ carries a $\delta$-interaction with the parameter $\beta_{j,h}$. The solution for the case when the line carries a vector potential $A_{(j,h)}$ can be obtained using Lemma~\ref{potencial}. Both solutions can be merged into the general formula
\begin{equation}\label{phi'}
\varphi_{j h}'(0)=\frac{k}{\sin k\frac{d}{\gamma_{j,h}}+\frac{\beta_{j,h}}{k}\sin^2 k\frac{d}{2\gamma_{j,h}}}\left(-\psi_j(0)\left(\cos k\frac{d}{\gamma_{j,h}}+\frac{\beta_{j,h}}{2k}\sin k\frac{d}{\gamma_{j,h}}\right)+\e^{-\i\frac{d}{2\gamma_{j,h}}A_{(j,h)}}\psi_h(0)\right)
\end{equation}
which is applicable to any of the three possible situations (a $\delta$-interaction, a vector potential, none of them), it suffices to set $\beta_{j,h}=0$ or $A_{(j,h)}=0$ if necessary.
Now we sum equations~\eqref{phi'} over $h$ and apply~\eqref{phipsi}, in this way we get $n$ equations relating $\psi_j(0)$ and $\psi_j'(0)$:
\begin{multline}\label{PsiPsi'}
-k\sum_{h\in N_j}\frac{\cos k\frac{d}{\gamma_{j,h}}+\frac{\beta_{j,h}}{2k}\sin k\frac{d}{\gamma_{j,h}}}{\sin k\frac{d}{\gamma_{j,h}}+\frac{\beta_{j,h}}{k}\sin^2 k\frac{d}{2\gamma_{j,h}}}\cdot\psi_j(0)
+k\sum_{h\in N_j}\frac{\e^{-\i\frac{d}{2\gamma_{j,h}}A_{(j,h)}}}{\sin k\frac{d}{\gamma_{j,h}}+\frac{\beta_{j,h}}{k}\sin^2 k\frac{d}{2\gamma_{j,h}}}\psi_h(0)
+\psi_j'(0)=\alpha_j\psi_j(0)\,, \\
j=1,\ldots,n\,.
\end{multline}

\emph{Step 2.} If a particle with energy $E$ comes from the $\ell$-th line into the internal approximating web, the final-state wave function naturally obeys relations~\eqref{PsiPsi'}. Let $\Psi^{(\ell)}_j(x)$ denote the components of the final-state wave function on the half lines and $\Psi^{(\ell)}_j(x)$, $(\Psi^{(\ell)}_j)'(0)$ be the boundary values and derivatives.
We express $\Psi_j(0)$ and $\Psi_j'(0)$ in terms of the entries of $\S_d(E;V_1,\ldots,V_n)$ using the definition of the scattering matrix~\eqref{psi_ji}:
\begin{gather}
\Psi^{(\ell)}_j(0)=\left\{\begin{array}{cl}
\frac{1}{\sqrt{k_\ell}}\left(1+[\S_d]_{\ell\ell}\right) & \text{for $j=\ell$}, \\
\frac{1}{\sqrt{k_j}}[\S_d]_{j \ell} & \text{for $j\neq \ell$},
\end{array}\right. \label{RT }
\\[1em]
(\Psi^{(\ell)}_j)'(0)=\left\{\begin{array}{cl}
\i\sqrt{k_\ell}\left(-1+[\S_d]_{\ell\ell}\right) & \text{for $j=\ell$}, \\
\i\sqrt{k_j}[\S_d]_{j \ell} & \text{for $j\neq \ell$},
\end{array}\right. \label{RT '}
\end{gather}
where $k_j=\sqrt{2m(E-V_j)/\hbar^2}$ for all $j=1,\ldots,n$.

If we multiply equation~\eqref{PsiPsi'} by $d$, then for all $j=1,\ldots,n$ replace $\psi_j(0)$ and $\psi_j'(0)$ by $\Psi^{(\ell)}_j(0)$ and $(\Psi^{(\ell)}_j)'(0)$, respectively, and finally, substitute for $\Psi^{(\ell)}_j(0)$ and $(\Psi^{(\ell)}_j)'(0)$ from equations~\eqref{RT } and \eqref{RT '}, we obtain a system of $n^2$ equations for $[\S_d]_{j\ell}$, $j,\ell=1,\ldots,n$. They take different forms for $j\neq\ell$ and $j=\ell$:

\begin{subequations}\label{Z jl}
\noindent $\bullet$\quad Case $j\neq \ell$
\begin{multline}\label{Z j<>l}
\left(-kd\sum_{h\in N_j}\frac{\cos k\frac{d}{\gamma_{j,h}}+\frac{\beta_{j,h}}{2k}\sin k\frac{d}{\gamma_{j,h}}}{\sin k\frac{d}{\gamma_{j,h}}+\frac{\beta_{j,h}}{k}\sin^2 k\frac{d}{2\gamma_{j,h}}}-\alpha_j d+\i k_j d\right)\frac{1}{\sqrt{k_j}}[\S_d]_{j \ell}+\\
+kd\sum_{h\in N_j}\frac{\e^{-\i\frac{d}{2\gamma_{j,h}}A_{(j,h)}}}{\sin k\frac{d}{\gamma_{j,h}}+\frac{\beta_{j,h}}{k}\sin^2 k\frac{d}{2\gamma_{j,h}}}\frac{1}{\sqrt{k_h}}[\S_d]_{h\ell}
\\
=-kd\frac{\e^{-\i\frac{d}{2\gamma_{j,\ell}}A_{(j,\ell)}}}{\sin k\frac{d}{\gamma_{j,\ell}}+\frac{\beta_{j,\ell}}{k}\sin^2 k\frac{d}{2\gamma_{j,\ell}}}\frac{1}{\sqrt{k_\ell}}\cdot\chi_{N_j}(\ell)\,,
\end{multline}
where $\chi_{N_j}(\ell)=1$ for $\ell\in N_j$ and $\chi_{N_j}(\ell)=0$ otherwise.

\noindent $\bullet$\quad Case $j=\ell$
\begin{multline}\label{Z j=l}
\left(-kd\sum_{h\in N_\ell}\frac{\cos k\frac{d}{\gamma_{\ell,h}}+\frac{\beta_{\ell,h}}{2k}\sin k\frac{d}{\gamma_{\ell,h}}}{\sin k\frac{d}{\gamma_{\ell,h}}+\frac{\beta_{\ell,h}}{k}\sin^2 k\frac{d}{2\gamma_{\ell,h}}}-\alpha_\ell d+\i k_\ell d\right)\frac{1}{\sqrt{k_\ell}}[\S_d]_{\ell\ell}+\\
+kd\sum_{h\in N_\ell}\frac{\e^{-\i\frac{d}{2\gamma_{\ell,h}}A_{(\ell,h)}}}{\sin k\frac{d}{\gamma_{\ell,h}}+\frac{\beta_{\ell,h}}{k}\sin^2 k\frac{d}{2\gamma_{\ell,h}}}\frac{1}{\sqrt{k_h}}[\S_d]_{h\ell}
\\
=\left(kd\sum_{h\in N_\ell}\frac{\cos k\frac{d}{\gamma_{\ell,h}}+\frac{\beta_{\ell,h}}{2k}\sin k\frac{d}{\gamma_{\ell,h}}}{\sin k\frac{d}{\gamma_{\ell,h}}+\frac{\beta_{\ell,h}}{k}\sin^2 k\frac{d}{2\gamma_{\ell,h}}}+\alpha_\ell d+\i k_\ell d\right)\frac{1}{\sqrt{k_\ell}}
\end{multline}\end{subequations}

\emph{Step 3.} The sought scattering matrix $\S_d(E;V_1,\ldots,V_n)$ can be obtained by solving the system of equations~\eqref{Z jl}.
We observe that the system~\eqref{Z jl} can be rewritten in the matrix form
\begin{equation}
(ZD^{-1}+\i dD)\S_d=-ZD^{-1}+\i dD,
\end{equation}
where $D=\diag(\sqrt{k_1},\ldots,\sqrt{k_n})$ is the matrix introduced by equation~\eqref{K} and
\begin{equation}\label{Z}
Z_{j \ell}=\left\{\begin{array}{cl}
kd\frac{\e^{-\i\frac{d}{2\gamma_{j,\ell}}A_{(j,\ell)}}}{\sin k\frac{d}{\gamma_{j,\ell}}+\frac{\beta_{j,\ell}}{k}\sin^2 k\frac{d}{2\gamma_{j,\ell}}}\cdot\chi_{N_j}(\ell) & \text{for $j\neq \ell$}\,, \\
-kd\sum_{h\in N_\ell}\frac{\cos k\frac{d}{\gamma_{\ell,h}}+\frac{\beta_{\ell,h}}{2k}\sin k\frac{d}{\gamma_{\ell,h}}}{\sin k\frac{d}{\gamma_{\ell,h}}+\frac{\beta_{\ell,h}}{k}\sin^2 k\frac{d}{2\gamma_{\ell,h}}}-\alpha_\ell d & \text{for $j=\ell$}\,.
\end{array}\right.
\end{equation}
Therefore,
\begin{equation}\label{SZ}
\S_d(E;V_1,\ldots,V_n)=(ZD^{-1}+\i dD)^{-1}(-ZD^{-1}+\i dD)=-I+2\i dD\left(Z+\i dD^2\right)^{-1}D.
\end{equation}

\begin{remark}
Note that the matrices $Z$ and $D$ occuring in equation~\eqref{SZ} have very different meanings. The matrix $Z$ represents primarily the \emph{internal} properties of the approximating graph, and depends on the approximated boundary conditions, on $d$ and on $E$. The matrix $D$ represents only the \emph{external} setting, and depends on $V_1,\ldots,V_n$ and on $E$.
\end{remark}

It suffices to 
determine $Z$. This will be done by substituting the values of the coupling parameters $\alpha_j$, $\beta_{j,\ell}$, the length factors $\gamma_{j,\ell}$ and the vector potential strengths $A_{(j,\ell)}$ into equation~\eqref{Z}. All these parameters are given by equations \eqref{spoj}, \eqref{alpha}, \eqref{gamma j<=m}, \eqref{beta,A j<=m}, \eqref{alpha j<=m}. We shall keep in mind that the expressions are different when $j,\ell$ belong to $\{1,\ldots,r\}$ and when $j,\ell$ belong to $\{r+1,\ldots,n\}$, and moreover, $Z_{j\ell}$ is defined in different ways for $j\neq\ell$ and for $j=\ell$. Therefore, the formula for $Z_{j \ell}$ is divided into several expressions:

\begin{subequations}\label{Z form}
\noindent $\bullet$\quad
Case $j,\ell\leq r$ and $j\neq \ell$
\begin{equation}
Z_{j \ell}=\left\{\begin{array}{cl}
-kd\,\e^{\i\arg (TT^*)_{j \ell}}\left(\sin\frac{kd}{|(TT^*)_{j \ell}|}\right)^{-1} & \text{if $(TT^*)_{j \ell}<0$ or $(TT^*)_{j \ell}\notin\R$}\,, \\
-kd\,\e^{\i\arg (TT^*)_{j \ell}}\left(\sin\frac{kd}{|(TT^*)_{j \ell}|}-8\frac{|(TT^*)_{j \ell}|}{kd}\sin^2\frac{kd}{2|(TT^*)_{j \ell}|}\right)^{-1} & \text{if $(TT^*)_{j \ell}>0$}\,, \\
0 & \text{if $(TT^*)_{j \ell}=0$}\,.
\end{array}\right.
\end{equation}
$\bullet$\quad
Case $j=\ell\leq r$
\begin{multline}
Z_{\ell \ell}=
-\sum_{\{h\leq r,h\neq \ell|(TT^*)_{h \ell}<0\vee t_{h \ell}\notin\R\}}kd\,\cotg\frac{kd}{|(TT^*)_{h \ell}|} \\
-\sum_{\{h\leq r,h\neq \ell|(TT^*)_{h \ell}<0\}}kd\frac{\cos\frac{kd}{|(TT^*)_{h \ell}|}-4\frac{|(TT^*)_{h \ell}|}{kd}\sin\frac{kd}{|(TT^*)_{h \ell}|}}{\sin\frac{kd}{|(TT^*)_{h \ell}}-8\frac{|(TT^*)_{h \ell}|}{kd}\sin^2\frac{kd}{2|(TT^*)_{h \ell}|}}\\
-\left(\sum_{h=r+1}^n |t_{\ell h}|^2-\sum_{h=r+1}^n |t_{\ell h}|+\sum_{\{h>r|t_{\ell h}<0\}} 2t_{\ell h}
-\sum_{\{h\leq r,h\neq \ell\}}\left|(TT^*)_{\ell h}\right|-\sum_{\{h\leq r|h\neq \ell\wedge (TT^*)_{\ell h}>0\}} 2(TT^*)_{\ell h}\right)\,.
\end{multline}
$\bullet$\quad
Case $j\leq r<\ell$
\begin{equation}\label{j<=m<l}
Z_{j \ell}=\left\{\begin{array}{cl}
kd\,\e^{\i\arg t_{j \ell}}\left(\sin\frac{kd}{|t_{j \ell}|}\right)^{-1} & \text{if $t_{j \ell}>0$ or $t_{j \ell}\notin\R$}\,, \\
kd\,\e^{\i\arg t_{j \ell}}\left(\sin\frac{kd}{|t_{j \ell}|}-8\frac{|t_{j \ell}|}{kd}\sin^2\frac{kd}{2|t_{j \ell}|}\right)^{-1} & \text{if $t_{j \ell}<0$}\,, \\
0 & \text{if $t_{j \ell}=0$}\,.
\end{array}\right.
\end{equation}
$\bullet$\quad
Case $\ell\leq r<j$
\begin{equation}
Z_{j \ell}=\overline{Z_{\ell j}}\,,\qquad\text{where $Z_{\ell j}$ is given by equation~\eqref{j<=m<l}}\,.
\end{equation}
$\bullet$\quad
Case $j,\ell>r$ and $j\neq \ell$
\begin{equation}
Z_{j \ell}=0\,.
\end{equation}
$\bullet$\quad
Case $j=\ell>r$
\begin{multline}
Z_{\ell \ell}=
-\sum_{\{h\leq r|t_{h \ell}>0\vee t_{h \ell}\notin\R\}}kd\,\cotg\frac{kd}{|t_{h\ell}|}
-\sum_{\{h\leq r|t_{h \ell}<0\}}kd\frac{\cos\frac{kd}{|t_{h\ell}|}-4\frac{|t_{h\ell}|}{kd}\sin\frac{kd}{|t_{h\ell}|}}{\sin\frac{kd}{|t_{h\ell}|}-8\frac{|t_{h\ell}|}{kd}\sin^2\frac{d}{2|t_{h\ell}|})}\\
-\left(1-\sum_{h=1}^r |t_{h\ell}|+\sum_{\{h\leq r|t_{h\ell}<0\}} 2t_{h\ell}\right)\,.
\end{multline}
\end{subequations}

Equations~\eqref{Z form} above fully determine the matrix $Z$. Therefore, for a given $T$, $V_1,\ldots,V_n$, $E$ and $d$, we can easily find $\S_d(E;V_1,\ldots,V_n)$ using the formula~\eqref{SZ} together with equations~\eqref{k_j}, \eqref{K} and \eqref{k}.

If $d\to0$, i.e., $kd\ll1$, we can expand each of the expressions~\eqref{Z form}. The calculation leads to
\begin{equation}\label{Zaprox}
Z=\left(\begin{array}{cc}
-TT^* & T \\
T^* & -I^{(n-r)}
\end{array}\right)+\O\left((kd)^2\right)\,.
\end{equation}
When the matrix~\eqref{Zaprox} is substituted into equation~\eqref{SZ}, one obtains the following approximative formula for $\S_d(E;V_1,\ldots,V_n)$:
\begin{multline}\label{S_d aprox}
\S_d(E;V_1,\ldots,V_n)=
-I+2\left(\begin{array}{c}
Q_{(1)} \\
Q_{(2)}T^*
\end{array}\right)
\left(Q_{(1)}^2+TQ_{(2)}^2T^*\right)^{-1}
\left(\begin{array}{cc}
Q_{(1)} & TQ_{(2)}
\end{array}\right)+\\
+2\i d\left(\begin{array}{c}
Q_{(1)}W^{-1}TQ_{(2)}^2 \\
Q_{(2)}\left(I-T^*W^{-1}TQ_{(2)}^2\right)
\end{array}\right)
\left(\begin{array}{cc}
Q_{(2)}^2T^*W^{-1}Q_{(1)} & \left(I-Q_{(2)}^2T^*W^{-1}T\right)Q_{(2)}
\end{array}\right)+\O(d^2)\,,
\end{multline}
where $W=Q_{(1)}^2+TQ_{(2)}^2T^*$ and $Q_{(1)}$, $Q_{(2)}$ are the diagonal matrices introduced in Proposition~\ref{S-matrix FT}.
In equation~\eqref{S_d aprox}, the scattering matrix of the approximating graph is found up to a small error. The error tends to zero for $d\to0$,
i.e.,
\begin{equation}\label{lim S}
\lim_{d\to0}\S_d(E;V_1,\ldots,V_n)=-I+2\left(\begin{array}{c}
Q_{(1)} \\
Q_{(2)}T^*
\end{array}\right)
\left(Q_{(1)}^2+TQ_{(2)}^2T^*\right)^{-1}
\left(\begin{array}{cc}
Q_{(1)} & TQ_{(2)}
\end{array}\right)\,.
\end{equation}
Now we compare the limit~\eqref{lim S} with the scattering matrix of the approximated star graph. Since the vertex of the approximated graph carries a F\"ul\"op--Tsutsui coupling, described by the boundary conditions~\eqref{FT}, its scattering matrix $\S(E;V_1,\ldots,V_n)$ is given by formula~\eqref{SFTpot}. Comparing equations~\eqref{lim S} and \eqref{SFTpot}, we find
\begin{equation}
\lim_{d\to0}\S_d(E;V_1,\ldots,V_n)=\S(E;V_1,\ldots,V_n)\,.
\end{equation}
To sum up, the scattering matrix of the approximating system converges to the scattering matrix of the requested F\"ul\"op--Tsutsui vertex coupling 
in the limit $d\to0$.


\section{Conclusions}\label{Section: Conclusion}

The resonance phenomenon found in this paper allows to control the transmission of a quantum particle along a line by external potentials put on other lines connected to the channel at a vertex. We have demonstrated that for certain well-chosen F\"ul\"op--Tsutsui vertices, the resonance is strong enough to enable a design of quantum spectral filters with fine characteristics.
The scattering characteristics are very different depending on the matrix $T$ determining the F\"ul\"op--Tsutsui coupling in the graph vertex. Therefore, it is likely that many more controllable quantum devices based on this concept can be obtained for other choices of $T$.

Some of the designed filters have rather small transmission probabilities inside the passbands. It would be useful to find how to increase the probabilities, and in this way to improve the filter characteristics.

Macroscopic controllability can be achieved also by an application of a magnetic field, for instance, by tuning a magnetic flux through rings carried by the controlling lines, or through ring-shaped controlling lines. It turns out that the use of suitably chosen F\"ul\"op--Tsutsui couplings together with the application of magnetic field lead to transmission characteristics with adjustable high peaks. We plan to address this problem in a subsequent article.

It seems that the idea can be extended to the control of the spin flow, because the mathematical description of a spin one-half particle propagating in a quantum star graph with $n$ lines is equivalent to the description of a spinless particle in a quantum star graph with $2n$ lines. Exploring the spin-filtering mechanism could be of interest.

Finally, there are known parallels between the particle propagation in quantum graphs and the wave propagation in waveguides.
Consequently, we are convinced that the results of this paper can be adapted to the control of a transmission of macroscopic
electromagnetic microwaves in a thin waveguide.

\section*{Acknowledgements}
We thank Prof. Atushi Tanaka for helpful suggestions.  This research was supported  by the Japan Ministry of Education, Culture, Sports, Science and Technology under the Grant number 21540402.



\begin{thebibliography}{99}

\bibitem{AG05}
S.~Albeverio, F.~Gesztesy, R.~H\o egh-Krohn, and H.~Holden:
``Solvable Models in Quantum Mechanics'', 2nd ed. with appendix by P. Exner,
AMS Chelsea, R. I., 2005.

\bibitem{EKST08}
P.~Exner, J.P.~Keating, P.~Kuchment, T.~Sunada, A.~Teplyaev, eds.:
\emph{Analysis on Graphs and Applications},
AMS ``Proc. of Symposia in Pure Math.'' Ser., vol.~77,
Providence, R.I., 2008, \textit{and references therein}.

\bibitem{LBHS10}
M.~Lawniczak, S.~Bauch, O.~Hul and L.~Sirko:
{Experimental investigation of the enhancement factor for microwave irregular networks with preserved 
and broken time reversal symmetry in the presence of absorption.},
\emph{Phys. Rev.} \textbf{E 81}, 046204 (5pp) (2010).

\bibitem{SCdL03}
A. G. M. Schmidt, B. K. Cheng and M. G. E. da Luz:
Green function approach for general quantum graphs,
\emph{J. Phys. A: Math. Gen.} \textbf{36} (2003), L545--L551.

\bibitem{CET09}
T.~Cheon, P.~Exner and O.~Turek:
{Spectral filtering in quantum Y-junction},
\emph{J. Phys. Soc. Jpn.} \textbf{78} (2009), 124004 (7pp).

\bibitem{CET10}
T.~Cheon, P.~Exner and O.~Turek: 
Approximation of a general singular vertex coupling in quantum graphs, 
\emph{Ann. Phys. (NY)} \textbf{325} (2010), 548--578.

\bibitem{CT10}
T.~Cheon and O.~Turek:
{Fulop-Tsutsui interactions on quantum graphs}, 
\emph{Phys. Lett.} \textbf{A 374} (2010), 4212--4221.

\bibitem{Ex96b}
P.~Exner: Weakly coupled states on branching graphs, \emph{Lett.
Math. Phys.} \textbf{38} (1996), 313--320.

\bibitem{KS99}
V.~Kostrykin, R.~Schrader:
{Kirchhoff's rule for quantum wires},
\emph{J. Phys. A: Math. Gen.} \textbf{32} (1999), 595--630.

\bibitem{Ha00}
M. Harmer: Hermitian symplectic geometry and extension theory,
\emph{J. Phys. A: Math. Gen.}, \textbf{33} (2000), 9193--9203.

\bibitem{KS00}
V.~Kostrykin, R.~Schrader: Kirchhoff's rule for quantum wires. II:
The Inverse Problem with Possible Applications to Quantum
Computers, \emph{Fortschr. Phys.} \textbf{48} (2000), 703--716.

\bibitem{FT00}
T. F\"ul\"op, I. Tsutsui: A free particle on a circle with point interaction, 
\emph{Phys. Lett.} \textbf{A264} (2000), 366--374.

\bibitem{NS00}
K.~Naimark, M.~Solomyak:
{Eigenvalue estimates for the weighted Laplacian
on metric trees}, 
\emph{Proc. London Math. Soc.} \textbf{80} (2000), 690--724.

\bibitem{SS02}
A.V.~Sobolev, M.~Solomyak:
{Schr\"{o}dinger operator on homogeneous metric
trees: spectrum in gaps}, 
\emph{Rev. Math. Phys.} \textbf{14} (2002), 421--467.

\bibitem{SMMC99}
T.~Shigehara, H.~Mizoguchi, T.~Mishima, T.~Cheon: Realization of a
four parameter family of generalized one-dimensional contact
interactions by three nearby delta potentials with renormalized
strengths, \emph{IEICE Trans. Fund. Elec. Comm. Comp. Sci.}
\textbf{E82-A} (1999), 1708--1713.

\end{thebibliography}
\end{document}